\documentclass[12pt,english]{amsart}
\usepackage[T1]{fontenc}
\usepackage[latin9]{luainputenc}
\usepackage{geometry}
\usepackage{babel}
\usepackage{url}
\usepackage{amsbsy}
\usepackage{amstext}
\usepackage{amsthm}
\usepackage{amssymb}
\usepackage{graphicx}
\usepackage[unicode=true,pdfusetitle,
 bookmarks=true,bookmarksnumbered=false,bookmarksopen=false,
 breaklinks=true,pdfborder={0 0 1},backref=false,colorlinks=false]
 {hyperref}

\makeatletter
\numberwithin{equation}{section}
\numberwithin{figure}{section}
\theoremstyle{plain}
\newtheorem*{thm*}{\protect\theoremname}
\theoremstyle{plain}
\newtheorem{thm}{\protect\theoremname}[section]
\theoremstyle{definition}
\newtheorem{defn}[thm]{\protect\definitionname}
\theoremstyle{remark}
\newtheorem{rem}[thm]{\protect\remarkname}
\theoremstyle{definition}
\newtheorem{example}[thm]{\protect\examplename}
\theoremstyle{plain}
\newtheorem{lem}[thm]{\protect\lemmaname}
\theoremstyle{plain}
\newtheorem{assumption}[thm]{\protect\assumptionname}
\theoremstyle{plain}
\newtheorem{prop}[thm]{\protect\propositionname}

\usepackage{xurl}

\@ifundefined{showcaptionsetup}{}{%
 \PassOptionsToPackage{caption=false}{subfig}}
\usepackage{subfig}
\makeatother

\providecommand{\assumptionname}{Assumption}
\providecommand{\definitionname}{Definition}
\providecommand{\examplename}{Example}
\providecommand{\lemmaname}{Lemma}
\providecommand{\propositionname}{Proposition}
\providecommand{\remarkname}{Remark}
\providecommand{\theoremname}{Theorem}

\begin{document}
\global\long\def\B{\mathcal{B}}%

\global\long\def\R{\mathbb{R}}%

\global\long\def\Q{\mathbb{Q}}%

\global\long\def\Z{\mathbb{Z}}%

\global\long\def\C{\mathbb{C}}%

\global\long\def\N{\mathbb{N}}%

\global\long\def\RR{\mathbb{\overline{R}}}%

\global\long\def\E{\mathcal{E}}%

\global\long\def\V{\mathcal{V}}%

\global\long\def\P{\mathcal{P}}%

\global\long\def\VR{\mathcal{V}_{\mathcal{R}}}%

\global\long\def\dom{\Omega}%

\global\long\def\sdom{\mathcal{S}}%

\global\long\def\spec#1{\textrm{\text{Spec}}\left(#1\right)}%

\global\long\def\meig#1{\mathcal{\lambda}^{\Gamma}\left(#1\right)}%

\global\long\def\deig#1{\lambda\left(#1\right)}%

\global\long\def\sreg{\Sigma^{\text{reg}}}%

\global\long\def\sf#1{\mathrm{Sf}_{#1}}%

\global\long\def\mult#1#2{\mathrm{Mult}_{#1}\left(#2\right)}%

\global\long\def\L{\mathbf{L}}%

\begin{center}
\textbf{\large{}SPECTRAL CURVES OF QUANTUM GRAPHS WITH}{\large\par}
\par\end{center}

\begin{center}
\textbf{\large{}$\delta_{s}$ TYPE VERTEX CONDITIONS}{\large\par}
\par\end{center}

\vspace{140bp}

\begin{center}
\emph{\Large{}Gilad Sofer}{\Large\par}
\par\end{center}

\thispagestyle{empty}

\newpage{}
\title[]{Spectral Curves of Quantum Graphs with $\delta_{s}$ Type Vertex
Conditions}

\maketitle
\vspace{70bp}

\begin{center}
{\Large{}Research Thesis In Partial Fulfillment of The Requirements
for the Degree of Master of Science in Mathematics }{\Large\par}
\par\end{center}

\vspace{70bp}

\begin{center}
\emph{\Large{}Gilad Sofer}\vspace{250bp}
\par\end{center}

\begin{center}
{\Large{}Submitted to the Senate of the Technion - Israel Institute
of Technology }{\Large\par}
\par\end{center}

\vspace{5bp}

\begin{center}
{\Large{}Elul, 5782 , Haifa, September, 2022 }{\Large\par}
\par\end{center}

\thispagestyle{empty}

\newpage{}

~\vspace{50bp}

\begin{center}
{\Large{}The Research Thesis Was Done Under The Supervision of Ram
Band in the Faculty of Mathematics}{\Large\par}
\par\end{center}

\vspace{60bp}

\begin{center}
{\Large{}The Generous Financial Help of the Gutwirth Fellowship and}{\Large\par}
\par\end{center}

\begin{center}
{\Large{}the Technion is Gratefully Acknowledged}{\Large\par}
\par\end{center}

\vspace{60bp}

\begin{center}
I sincerely thank Ram Band for his wonderful guidance and support
throughout this experience\vspace{60bp}
\par\end{center}

\textbf{Conferences. }

\vspace{6bp}

Parts of the Thesis appeared in the following.
\begin{enumerate}
\item QGraph conference, Stockholm University, December 2021.\\
\emph{\url{https://staff.math.su.se/kurasov/QGRAPH/Meeting1221.html}}
\item Ergodic Operators and Quantum Graphs, Simons Center for Geometry and
Physics, Stony Brook University, New York, June 2022.\\
\emph{\url{https://scgp.stonybrook.edu/archives/32892}}
\item Heat Kernels on Graphs and Manifolds, Center for Distance Learning,
Bregenz, August 2022.\\
\emph{\url{https://www.mat-dyn-net.eu/en/spectral-geometry-2022}}
\item Workshop on Functional Analysis, Operator Theory and Dynamical Systems,
Potsdam University, September 2022.\\
\emph{\url{https://www.math.uni-potsdam.de/institut/veranstaltungen/details/veranstaltungsdetails/workshop-on-functional-analysis-operator-theory-and-dynamical-systems}}
\end{enumerate}
\newpage{}

\textbf{Acknowledgment.}

\vspace{6bp}

I had the pleasure to work with several great collaborators in my
studies, and some parts of the thesis are based on these collaborations.

Some results regarding the Robin-Neumann gap are based on joint work
with Ram Band, Holger Schanz, and Uzy Smilansky {[}10{]}, with a preprint
available at

\emph{\url{https://arxiv.org/abs/2212.12531}}

Some results regarding the spectral flow are based on joint work with
Ram Band and Marina Prokhorova {[}9{]} (to be submitted in the future).

\thispagestyle{empty}

\newpage{}

\tableofcontents{}

\thispagestyle{empty}

\newpage{}

\listoffigures

\thispagestyle{empty}\newpage{}

\section*{Abstract}

\vspace{120bp}

In this Thesis we study the behavior of spectral curves of quantum
graphs under certain families of vertex conditions, called the $\delta_{s}$
family, which we define in this work. We focus on studying two main
quantities related to the spectral curves, known as the Robin-Neumann
gap and the spectral flow. We show that these quantities hold information
about the the spectral curves, the behavior of the corresponding eigenfunctions,
and the geometry of the graph itself.

For a specific subset of the $\delta_{s}$ family which is known as
the $\delta$ family, we study the Robin-Neumann gap, which measures
the total increase in the eigenvalues with respect to the perturbation
parameter. We use this quantity to show that the growth of the spectral
curves is uniformly bounded, and that on average it is linear, with
proportionality factor determined by the geometry of the graph.

For the general $\delta_{s}$ family of vertex conditions, we study
a quantity known as the spectral flow, which counts the number of
oriented intersections of the spectral curves with some given horizontal
cross section. We use this quantity to prove an index theorem which
connects between a generalized nodal deficiency of the eigenfunctions
and the stability index of a generalized Dirichlet-to-Neumann map.
We also show that the spectral flow holds information about the graph
topology.

\setcounter{page}{1}

\newpage{}

\section*{Nomenclature}

$\R$ - The set of Real numbers.

$\overline{\mathbb{R}}$ -- The extended real line $\mathbb{R}\cup\left\{ \infty\right\} $.

$\N$ - The set of Natural numbers.

$-\triangle$ - The non-negative Laplacian.

$\Omega$ - A bounded domain in $\R^{N}$ with piecewise smooth boundary.

$\frac{\partial f}{\partial n}\mid_{\partial\Omega}$ - The outwards
pointing normal derivative of $f$ along the boundary $\partial\Omega$.

$G=\left(\mathcal{V},\mathcal{E}\right)$ - A combinatorial graph
with vertex set $\mathcal{V}$ and edge set $\E$.

$V$ - Number of vertices in $\V$.

$E$ - Number of edges in $\E$.

$E_{v}$ - Set of edges connected to the vertex $v$.

$\deg\left(v\right)$ - Degree of the vertex $v$, or $\left|E_{v}\right|$.

$\Gamma$ - A compact, connected metric graph.

$\beta_{\Gamma}$ - First Betti number of $\Gamma$.

$L$ - Total length of the graph.

$H^{1}(\Gamma)$ - The Sobolev space $W^{1,2}(\Gamma)$.

$H_{0}^{1}(\Gamma)$ - The Sobolev space $W_{0}^{1,2}(\Omega)$.

$H^{2}(\Gamma)$ - The Sobolev space $W^{2,2}(\Gamma)$.

$L^{2}\left(X\right)$ - The Hilbert space $L^{2}$ on the space $X$.

$Dom\left(H\right)$ - The domain of the operator/sesquilinear form
$H$.

$\spec H$ - The spectrum of the operator $H$.

$\mult{\lambda}H$ - The multiplicity of the eigenvalue $\lambda$
in $\spec H$.

$H^{s}\left(t\right)$ - The Hamiltonian of the $\delta_{s}$ family
at time $t$.

$\lambda_{n}^{s}\left(t\right)$ - The value of the $n$th spectral
curve of the $\delta_{s}$ family at time $t$.

$k_{n}$ - The wave number of the $n$th eigenvalue, which satisfies
$k_{n}^{2}=\lambda_{n}$.

$d_{n}\left(\sigma\right)$ - The $n$th Robin-Neumann gap with coupling
parameter $\sigma$.

$\VR$ - The set of $\delta$ vertices.

$\left\langle a\right\rangle _{n}$ - The Cesaro sum of the sequence
$\left(a_{n}\right)_{n=1}^{\infty}$.

$\mathbb{T}^{E}$ - The $E$ dimensional torus $\mathbb{R}^{E}/2\pi\mathbb{Z}^{E}$.

$S^{\left(t\right)}\left(k\right)$ - The bond scattering matrix with
Robin parameter $t$ and wave number $k$.

$\L$ - Diagonal matrix of edge lengths $\left(\ell_{1},\ell_{1},...,\ell_{E},\ell_{E}\right)$.

$\Sigma$ - The secular manifold.

$\sreg$ - The regular part of the secular manifold.

$\mathcal{S}_{s}\left(f_{n}\right)$ - The set of $s$ points corresponding
to the $n$th eigenfunction.

$\phi_{s}\left(f_{n}\right)$ - The number of $s$ points corresponding
to the $n$th eigenfunction.

$\nu_{s}\left(f_{n}\right)$ - The number of $s$ domains corresponding
to the $n$th eigenfunction$.$

$\mathcal{D}_{s}\left(f_{n}\right)$ - The $s$ deficiency corresponding
to the $n$th eigenfunction , or $n-\nu_{s}\left(f_{n}\right)$.

$\Lambda_{s}\left(c\right)$ - The Robin map with Robin parameter
$s$ and spectral parameter $c$.

$L_{t}^{s}$ - The sesquilinear form corresponding to the $\delta_{s}$
family.

$\gamma^{s},\gamma^{s\star}$ - The $s$ traces. 

$\sf{\lambda}\left(I\right)$ - The spectral flow along the interval
$I$ through the horizontal line $\lambda$.

$Mor\left(H\right)$ - The Morse index of an operator $H$, or number
of negative eigenvalues.

$Pos\left(H\right)$ - The positive index of an operator $H$, or
number of positive eigenvalues.

\newpage{}

\section{Introduction\label{sec:Introduction}}

Given a `nice enough' family of self-adjoint operators $\left(H\left(\vec{R}\right)\right)_{\vec{R}\in M}$
(where $M$ is some smooth manifold), one can consider the collection
of eigenvalues of this family -- $\lambda_{n}\left(\vec{R}\right)$.
In the particular case of a one-parameter family $H\left(t\right)$,
this gives a sequence of spectral curves $\lambda_{n}\left(t\right)$.

The study of the behavior of spectral curves of families of self-adjoint
operators has a long history in spectral geometry, dating back to
the classical Hadamard type formula \cite{Had_book08}, which gives
a formula for the derivative of the spectral curves of the Dirichlet
Laplacian for a time dependent domain $\Omega\left(t\right)\subset\mathbb{R}^{N}$.
Since then, many works have been devoted to the subject, with numerous
works focusing on families of self-adjoint extensions of the Laplacian
on manifolds, metric graphs and discrete graphs.

While the spectral curves naturally give information about the eigenvalues
themselves, it turns out that in some cases they can also give information
about the corresponding eigenfunctions, and even about the geometry
of the underlying space.

A nice example from recent years is the study of the so called Robin-Neumann
gap on planar domains, which gives a connection between the change
in the spectral curves of the Robin Laplacian and geometric properties
of the domain \cite{RudWigYes_arxiv21}. Given a bounded domain $\Omega\subset\mathbb{R}^{2}$
with piecewise smooth boundary, one can consider the family of Laplacians
$\left(-\Delta\left(t\right)\right)_{t\geq0}$ under the Robin boundary
condition with coupling parameter $t$:
\begin{equation}
\frac{\partial f}{\partial n}|_{\partial\Omega}+tf|_{\partial\Omega}=0,\label{eq:-85}
\end{equation}
where $\frac{\partial f}{\partial n}$ is the outwards pointing normal
derivative of $f$. One then defines the sequence of Robin-Neumann
gaps $\left(d_{n}\left(t\right)\right)_{n=1}^{\infty}$ as the total
change of the $n$th spectral curve due to the Robin condition: 
\begin{equation}
d_{n}\left(t\right):=\lambda_{n}\left(t\right)-\lambda_{n}\left(0\right).\label{eq:-86}
\end{equation}

The following relation between the expectation value of the Robin-Neumann
gap and the geometry of $\Omega$ was recently proven in \cite{RudWigYes_arxiv21}:
\begin{equation}
\lim_{N\rightarrow\infty}\frac{1}{N}\sum_{n=1}^{N}d_{n}\left(t\right)=\frac{2\cdot\text{Length}\left(\partial\Omega\right)}{\text{Area\ensuremath{\left(\Omega\right)}}}t.\label{eq:-87}
\end{equation}

The spectral curves have also played an important role in the study
of nodal domains of Laplacian eigenfunctions. Given an open subset
$\Omega\subset\mathbb{R}^{n}$ and an eigenpair $\left(\lambda_{n},f_{n}\right)$
of a Schrödinger operator $H$, one can define the nodal set of $f_{n}$
by
\begin{equation}
\mathcal{S}_{\infty}\left(f_{n}\right)=\left\{ x\in\Omega:f_{n}\left(x\right)=0\right\} .\label{eq:-22}
\end{equation}

The nodal set partitions $\Omega\backslash\mathcal{S}_{\infty}\left(f_{n}\right)$
into connected components, which are called the nodal domains of $f_{n}$.
Denoting the number of nodal domains of $f_{n}$ by $\nu_{\infty}\left(f_{n}\right)$,
one can define the corresponding nodal deficiency of $f_{n}$:
\begin{equation}
\mathcal{D}_{\infty}\left(f_{n}\right):=n-\nu_{\infty}\left(f_{n}\right).\label{eq:-23}
\end{equation}

A famous theorem by Courant (\cite{Cou_ngwgmp23}) states that $\mathcal{D}_{\infty}\left(f_{n}\right)\geq0$.
This result was the corner stone to many later works devoted to studying
the nodal deficiency.

An important related family of results from recent years are the so
called \emph{nodal index theorems}, which link between the nodal
count of Laplacian eigenfunctions and the stability index of some
appropriate functional/operator. It turns out that for many of these
results, the proof relies on defining some proper parametric family
of self-adjoint operators and studying the behavior of the corresponding
spectral curves. 

An interesting example of such a result is the celebrated nodal-magnetic
theorem for quantum and discrete graphs (see \cite{Ber_apde13,BerWey_ptrsa14}),
which relates the stability index of the $n$th eigenvalue of the
magnetic Laplacian with respect to a magnetic perturbation to the
nodal count of the corresponding eigenfunction.

Another important example, which serves as a main motivation for this
work, is the following theorem:
\begin{thm*}
\cite{BerCoxMar_lmp19,Cox2017} Let $\left(\lambda_{n},f_{n}\right)$
be a simple eigenpair of the Neumann Laplacian on a bounded, Lipschitz
domain $\Omega\subset\mathbb{R}^{N}$. Then for $\epsilon>0$ small
enough, the following formula for the nodal deficiency holds:
\begin{equation}
\mathcal{D}_{\infty}\left(f_{n}\right)=Mor\left(\Lambda\left(\lambda_{n}+\epsilon\right)\right),\label{eq:-84}
\end{equation}
where $\Lambda\left(\lambda_{n}+\epsilon\right)$ is the perturbed
two sided Dirichlet to Neumann map evaluated at the nodal set of $f_{n}$,
and $Mor$ is the Morse index (number of negative eigenvalues).
\end{thm*}
Here, the perturbed Dirichlet to Neumann map (which is sometimes also
known as a Steklov type operator) is a linear map acting on $L^{2}\left(\partial\Omega\right)$
which, in essence, links between the boundary value of the eigenfunction
$f_{n}$ and its normal derivative. A generalized version of this
map is defined for quantum graphs in Section \ref{sec:Prelims}.

While the spectral curves do not appear explicitly in the statement
of the theorem above, a very natural proof is given by defining an
appropriate one-parameter family of operators and studying the corresponding
spectral curves.

In the present work, we study the behavior of the spectral curves
of the Laplacian on metric graphs. A metric graph $\Gamma$ is a one-dimensional
simplicial complex with a natural notion of $L^{2}\left(\Gamma\right)$,
which can then be used to define a Laplacian (See complete definition
in Section \ref{sec:Prelims}). We focus on studying the spectral
curves of the Laplacian under the $\delta_{s}$ family of vertex conditions,
which we define in this work. A large portion of the work is dedicated
specifically to the $\delta$ vertex condition, which has been studied
extensively in past works.

The work concerns with studying two main quantities related to the
behavior of the spectral curves. The first is the Robin-Neumann gap
of the $\delta$ family, which in essence captures the total vertical
change in the spectral curves along some given interval; The second
is the spectral flow of the $\delta_{s}$ family, which in essence
captures the number of spectral curves which intersect some given
horizontal cross section. Both of these quantities are defined precisely
for quantum graphs in Section \ref{sec:Prelims}.

For the Robin-Neumann gap, we show that the sequence of Robin-Neumann
gap functions is uniformly Lipschitz continuous, uniformly bounded,
and contains a uniformly convergent subsequence (Theorem \ref{thm:1.RNG-Lipschitz}).
We also show (Theorem \ref{thm:RNG-mean}) that the mean value of
the Robin-Neumann gap exists, and is given by a geometric expression
analogous to the one given in Formula (\ref{eq:-87}). In the process
of proving Theorem \ref{thm:RNG-mean}, we also prove a local Weyl
law, which provides an estimate for the mean value of the graph eigenfunctions
at the graph vertices (Theorem \ref{thm:Weyl-law}).

For the spectral flow, we provide an index theorem (Theorem \ref{thm:SF-index})
which relates a generalized nodal deficiency to the Morse index of
a generalized Dirichlet to Neumann map, both of which are defined
in Section \ref{sec:Prelims}. This result provides a substantial
generalization to the formula presented in Formula (\ref{eq:-84})
above for the graph setting, which allows one to study not only nodal
points of eigenfunctions (points such that $f_{n}=0$), but also Neumann
points (points such that $\frac{df_{n}}{dx}=0$) and a much wider
class of points. Moreover, we show that the spectral flow of the $\delta_{s}$
family is determined by the number of interaction points (Theorem
\ref{prop:SF-points}), and can be related to the first Betti number
of the graph (Theorems \ref{thm:SF-Betti1} and \ref{prop:SF-Betti2}).

The structure of the thesis is as follows. Section \ref{sec:Prelims}
is dedicated for preliminary background required in order to state
the main results. We then present the main results precisely in Section
\ref{sec:Main-Results}. Section \ref{sec:RNG-tools} is devoted to
explaining and developing tools necessary for the proof of the results
related to the Robin-Neumann gap, which are then proven in Section
\ref{sec:proof-of-rng}. Similarly, Section \ref{sec:SF-tools} is
devoted to explaining and developing the tools necessary for the proof
of the results related to the spectral flow, which are then proven
in Section \ref{sec:proof_for_sf}. Finally, Section \ref{sec:Discussion}
includes a discussion of the results and suggestions to possible future
research directions and applications of this work. Appendix \ref{sec:Appendices}
is devoted to deriving the sesquilinear form corresponding to the
$\delta_{s}$ condition which appears in this work. Appendix \ref{sec:DTN-comp}
presents an example for a computation for the Robin map presented
in this work.

\newpage{}

\section{\label{sec:Prelims}Preliminaries}

\subsection{Basic introduction to quantum graphs}
\begin{defn}
\label{def:Mgraph}A \emph{metric graph} is a pair $\Gamma=\left(G,\vec{\ell}\right)$,
where $G=\left(\mathcal{V},\mathcal{E}\right)$ is a combinatorial
graph and $\vec{\ell}\in\mathbb{R}_{+}^{E}$ is a vector of positive
edge lengths associated to the edges in $\mathcal{E}$. $\Gamma$
is considered as a metric space, such that each edge $e\in\mathcal{E}$
is identified with the interval $\left[0,\ell_{e}\right]$ (as shown
in Figure \ref{fig:A-metric-graph.}). The total length of the graph
is denoted by $L:=\sum_{e\in\mathcal{E}}\ell_{e}$.

For each vertex $v$, we denote the set of edges connected to $v$
by $E_{v}$, and moreover denote $\deg\left(v\right):=\left|E_{v}\right|$.
Furthermore, we denote $V=\left|\mathcal{V}\right|,E=\left|\mathcal{E}\right|$.
\end{defn}

\begin{figure}
\includegraphics[scale=0.7]{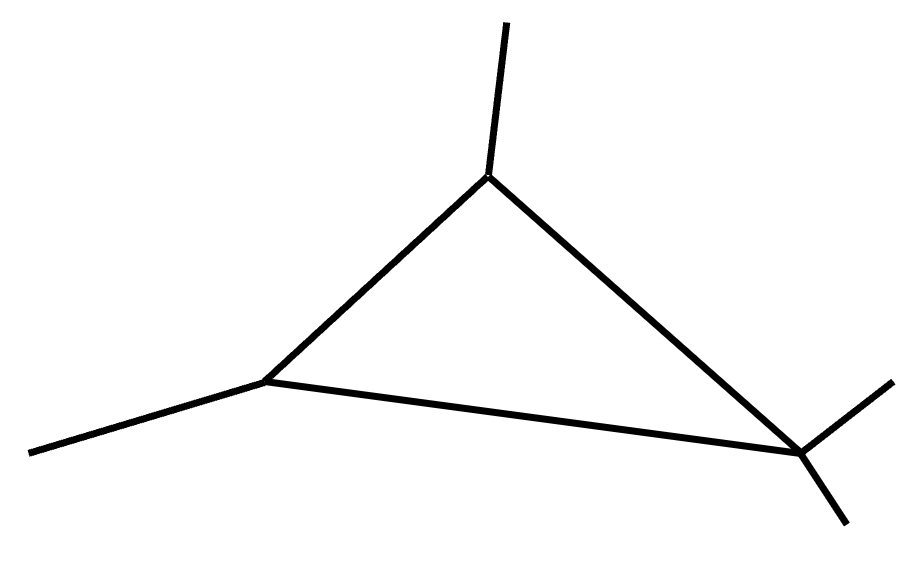}

\caption{A metric graph.\label{fig:A-metric-graph.}}
\end{figure}

\begin{defn}
\label{def:Qgraph}A \emph{Quantum Graph} is a metric graph $\Gamma$,
equipped with a self-adjoint differential operator $H$ (usually called
the \emph{Hamiltonian}), acting on the Sobolev space $H^{2}\left(\Gamma\right):=\oplus_{e\in\mathcal{E}}H^{2}\left(\left[0,\ell_{e}\right]\right)$.

One usually takes $H$ to be a Schrödinger operator -- $H=-\frac{d^{2}}{dx^{2}}+V\left(x\right)$
with $V\left(x\right)\in L^{\infty}\left(\Gamma\right)$, along with
certain vertex conditions such that $H$ is self-adjoint. Some common
examples for self-adjoint vertex conditions:
\end{defn}

\begin{enumerate}
\item Neumann-Kirchhoff (or standard) vertex condition: \label{enu:NK}
\begin{align}
 & f\text{ is continuous at \ensuremath{v} - \ensuremath{f|_{e}\left(v\right)=f|_{e'}\left(v\right),\forall e,e'\in E_{v}}},\label{eq:-15}\\
 & \text{Current conservation - \ensuremath{\sum_{e\in E_{v}}f'|_{e}\left(v\right)=0,}}\label{eq:-16}
\end{align}
where by convention the derivatives are taken in the outwards direction
from the vertex.
\item $\delta$ vertex condition with coupling parameter $\sigma\in\mathbb{R}$:
\begin{align}
 & \text{Continuity at \ensuremath{v} - \ensuremath{f|_{e}\left(v\right)=f|_{e'}\left(v\right):=f\left(v\right),\forall e,e'\in E_{v},}}\label{eq:-17}\\
 & \text{\text{Robin condition -}\ensuremath{\sum_{e\in E_{v}}f'|_{e}\left(v\right)=\sigma f\left(v\right)}.}\label{eq:-18}
\end{align}
Note that the Neumann-Kirchhoff condition is a particular case, with
$\sigma=0$. The $\delta$ condition has appeared in many previous
works, as a model for describing a singular potential barrier on the
graph.
\item For a vertex of degree two, one can define the $\delta'$ vertex condition
with coupling parameter $\sigma\in\mathbb{R}$:
\begin{align}
 & f'\text{ is continuous at \ensuremath{v} - \ensuremath{f'|_{e_{1}}\left(v\right)+f'|_{e_{2}}\left(v\right)=0,}}\label{eq:-19}\\
 & \text{\ensuremath{f|_{e_{1}}\left(v\right)-f|_{e_{2}}\left(v\right)=\sigma f'|_{e_{2}}\left(v\right)}.}\label{eq:-20}
\end{align}
Note that $\sigma=0$ once again corresponds to the Neumann-Kirchhoff
condition. The $\delta'$ condition can be thought of as an analog
of the $\delta$ condition, with the roles of the function and its
derivative reversed. It has been studied in the context of periodic
scattering arrays with the sign of the coupling parameter reversed
(see \cite{AvrExnLas_prl94}) and for certain lattice Kronig-Penney
models (see \cite{Exner1995}).
\end{enumerate}
\begin{rem}
All of the vertex conditions above are a particular case of the so
called $\delta_{s}$ family, which we define in Subsection \ref{subsec:delta-s-definition}.
\end{rem}

In the following work, we always take $\Gamma$ to be a compact, connected
metric graph. We also take our Hamiltonian to be the Laplacian $H=-\frac{d^{2}}{dx^{2}}$,
with one of the vertex conditions stated above.

By standard theory (see Section $3$ in \cite{BerKuc_graphs}), $H$
is self-adjoint, and its spectrum is infinite, discrete, and bounded
from below. We can thus write the spectrum of $H$ as $\lambda_{1}<\lambda_{2}\leq...\rightarrow\infty$,
with a complete orthonormal set of eigenfunctions $f_{1},f_{2},...$.

Moreover, in the case of the Neumann-Kirchhoff Laplacian, we shall
sometimes assume that the eigenpairs $\left(\lambda_{n},f_{n}\right)_{n=1}^{\infty}$
are generic.
\begin{defn}
\label{def:generic}An eigenpair $\left(\lambda_{n},f_{n}\right)_{n=1}^{\infty}$
of the Neumann-Kirchhoff Laplacian is said to be generic if $\lambda_{n}$
is simple, and $f_{n}$ and its derivative do not vanish at any vertex
of degree larger than two. This property is generic in the sense explained
in \cite{Alon_PhDThesis}.
\end{defn}

\bigskip

\subsection{The Robin-Neumann gap\label{subsec:RNG-definition}}

Let $\Gamma$ be a quantum graph, initially with Neumann-Kirchhoff
condition imposed at all vertices. We introduce a perturbation to
our initial graph, by selecting a finite subset of vertices $\VR\subset\mathcal{V}$,
and on this subset of vertices imposing the $\delta$ condition with
parameter $\sigma\in\mathbb{R}$ (recall (\ref{eq:-17}), (\ref{eq:-18})).

We denote our perturbed Hamiltonian by $H\left(\sigma\right)$, and
its eigenvalues by $\left(\lambda_{n}\left(\sigma\right)\right)_{n=1}^{\infty}$.
If $\sigma\geq0$, then $\spec{H\left(\sigma\right)}\subset[0,\infty)$,
and we can also denote the eigenvalues using the wave number $\left(k_{n}^{2}\left(\sigma\right)\right)_{n=1}^{\infty}:=\left(\lambda_{n}\left(\sigma\right)\right)_{n=1}^{\infty}$.

One can show that the eigenvalues of $H\left(\sigma\right)$ are non-decreasing
with respect to $\sigma$, see Proposition $3.1.6$ in \cite{BerKuc_graphs}.
To quantify this increase, we define the \emph{Robin-Neumann gap}
(or RNG) by
\begin{equation}
d_{n}\left(\sigma\right)=\lambda_{n}\left(\sigma\right)-\lambda_{n}\left(0\right).\label{eq:-21}
\end{equation}

This gives us an infinite sequence of functions $\left(d_{n}\left(\sigma\right)\right)_{n=1}^{\infty}$,
which quantifies the total increase in the spectrum of $H\left(\sigma\right)$
due to the $\delta$ perturbation (see Figure \ref{fig:RNG}).

\begin{figure}
\includegraphics[scale=0.8]{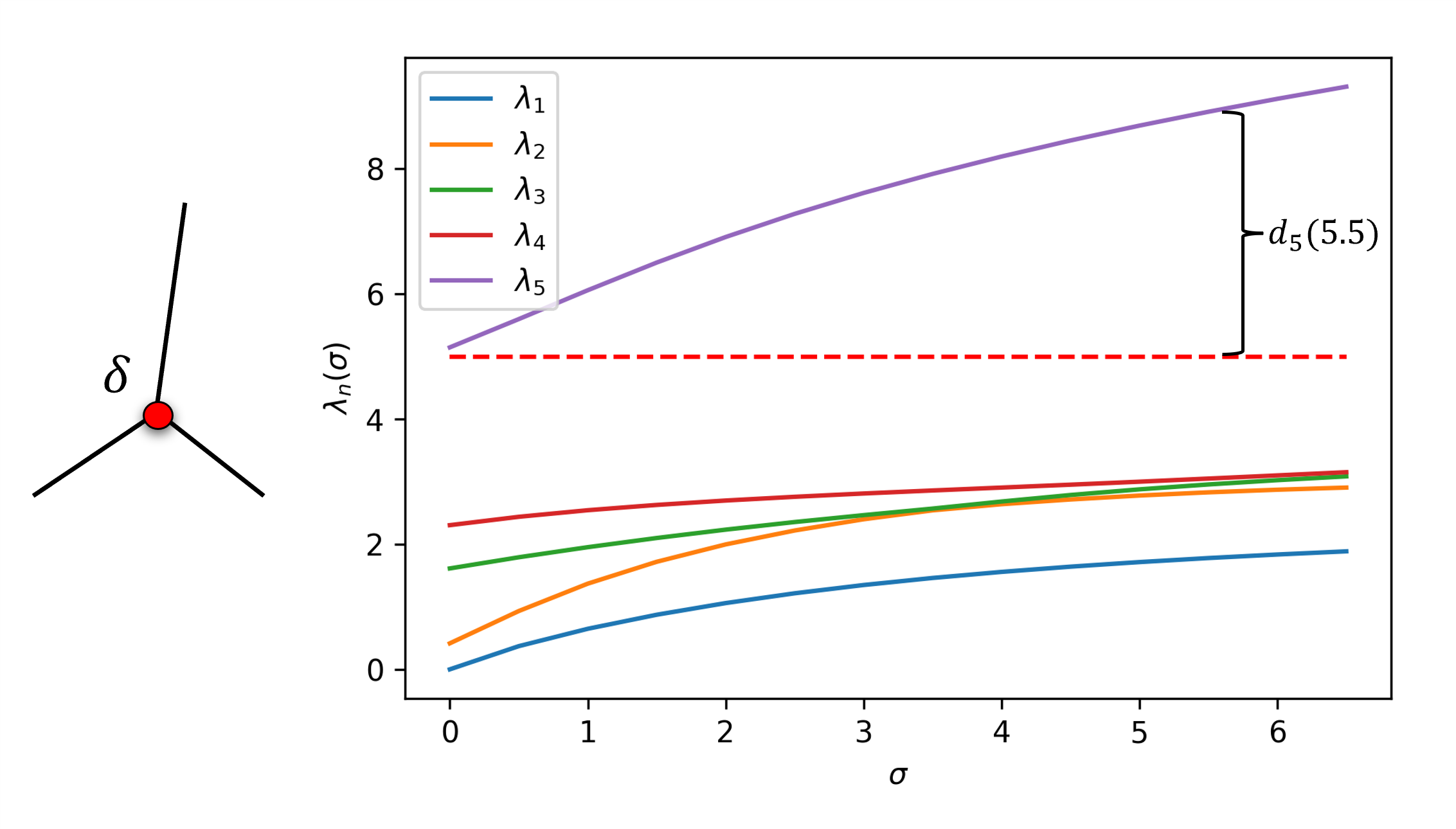}

\caption[Spectral curves of star graph with Robin-Neumann gap .]{The first five spectral curves of a star graph with a $\delta$ point
at the central vertex, and the fifth Robin-Neumann gap $d_{5}\left(5.5\right)$.\label{fig:RNG}}
\end{figure}

The RNG has been recently studied in the two-dimensional setting for
planar domains and the hemisphere in \cite{RudWigYes_arxiv21,RudWig_amq21},
and for non self-adjoint $\delta$ condition on star graphs in \cite{RivRoy_jphys20}.
\begin{rem}
\label{rem:ambiguity}The definition of $d_{n}\left(\sigma\right)$
might not be inclusive if $\lambda_{n}\left(\sigma\right)$ or $\lambda_{n}\left(0\right)$
is a multiple eigenvalue, since then it is not clear which branch
of the spectral curves one should follow. For the purpose of this
work, which deals with the collective behavior of the sequence $d_{n}\left(\sigma\right)$,
this ambiguity does not matter, since it only results in exchanging
the order of several elements of the sequence. Thus, one may choose
the eigenvalue branch arbitrarily in this case (and this possible
exchange in order will not affect our results).
\end{rem}

\bigskip

\subsection{Nodal domains, Neumann domains and $s$ domains on quantum graphs\label{subsec:sDomains-definition}}

The definition given in Section \ref{sec:Introduction} for nodal
domains on piecewise smooth domains in $\mathbb{R}^{N}$ can be naturally
extended for quantum graphs; Given an eigenpair $\left(\lambda_{n},f_{n}\right)$
of a Schrödinger operator, one can define the nodal set of $f_{n}$
by
\begin{equation}
\mathcal{S}_{\infty}\left(f_{n}\right)=\left\{ x\in\Gamma:f_{n}\left(x\right)=0\right\} .\label{eq:-24}
\end{equation}

Under our assumption that $f_{n}$ is a generic eigenfunction (see
Definition \ref{def:generic}), $\mathcal{S}_{\infty}\left(f_{n}\right)$
is a finite set of points of degree two, which are called \emph{nodal points}.
We denote the number of nodal points (which is known as the nodal
count) by $\phi_{\infty}\left(f_{n}\right)$.

As in the case of domains, the nodal set partitions $\Gamma\backslash\mathcal{S}_{\infty}\left(f_{n}\right)$
into connected components, which are called \emph{nodal domains}
(see Figure \ref{fig:nodal-domains-2}). We denote the number of nodal
domains of $f_{n}$ by $\nu_{\infty}\left(f_{n}\right)$. The nodal
deficiency of an eigenfunction $f_{n}$ is then defined as
\begin{equation}
\mathcal{D}_{\infty}\left(f_{n}\right):=n-\nu_{\infty}\left(f_{n}\right)\label{eq:-109}
\end{equation}

\begin{figure}
\includegraphics[scale=0.7]{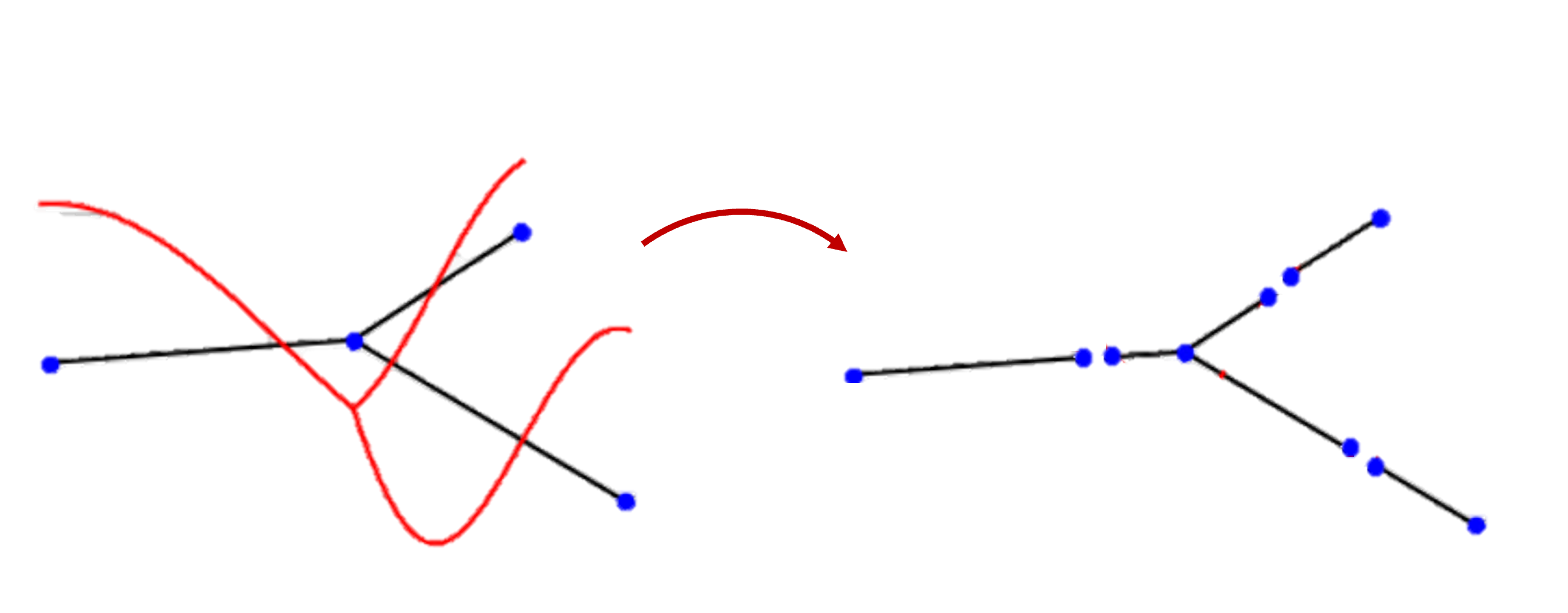}

\caption{Nodal domains for an eigenfunction of a star graph.\label{fig:nodal-domains-2}}
\end{figure}

Under the genericity assumption, the following version of Courant's
nodal theorem (see \cite{Courant23}) holds for the graph setting:
\begin{thm*}[\cite{Gnutzmann2003},\cite{Ber_cmp08}]
The following bounds hold:
\begin{equation}
0\leq\mathcal{D}_{\infty}\left(f_{n}\right)\leq\beta_{\Gamma},\label{eq:-25}
\end{equation}
where $\beta_{\Gamma}$ is the first Betti number of the graph, $\beta_{\Gamma}=E-V+1$.
\end{thm*}
A related quantity to the nodal count which has appeared in several
recent works (see \cite{AloBan_21ahp,AloBanBerEgg_lms20}) is the
so called Neumann count. Here, instead of partitioning $\Gamma$ at
the points on which $f=0$, one instead partitions the graph according
to the points on which $f'=0$. These points are known as Neumann
points, and the corresponding domains are known as Neumann domains.

Note that nodal points can be considered as points such that the corresponding
eigenfunction satisfies the Dirichlet condition. Similarly, Neumann
points are points such that the eigenfunction satisfies the Neumann
condition. These definition suggest towards a natural generalization,
which we shall call Robin points.
\begin{defn}
\label{def:s-points}Given $s\in\mathbb{R}$, we say that an eigenfunction
$f$ of the Neumann-Kirchhoff Laplacian satisfies the \emph{Robin condition}
at some point $x$ (of degree one or two) with Robin parameter $s$,
if $f$ satisfies
\begin{equation}
f'\left(x\right)=sf\left(x\right),\label{eq:-26}
\end{equation}
where the derivative is taken with respect to the natural orientation
induced by the parametrization of the corresponding edge as $\left[0,\ell_{e}\right]$.

The case $s=0$ corresponds to the Neumann condition. One can also
consider the Dirichlet condition as a particular case, by allowing
the value $s=\pm\infty$. Thus, from now on, we consider $s$ as a
parameter on the one point compactification of $\mathbb{R}$:
\begin{equation}
s\in\RR:=\mathbb{R}\cup\left\{ \infty\right\} .\label{eq:-27}
\end{equation}

By varying the value of $s$, this allows us to simultaneously consider
Neumann points ($s=0$), nodal points ($s=\infty$), and a much wider
collection of points. These Robin points will usually be referred
to as \emph{s points}, and we say that $f$ satisfies the \emph{s condition}
at these given points.
\end{defn}

\begin{rem}
Note that this definition only corresponds to vertices of degree one
or two. This agrees with the definition of nodal and Neumann points
under the genericity assumption. We do not define the notion of $s$
points for vertices of higher degree.
\end{rem}

\begin{rem}
Recall that these definitions rely on a fixed choice of orientation
on the edges of $\Gamma$. Mainly, for $s\neq0,\infty$, if one reverses
the orientation of an edge, then one has to replace the value $s$
with the value $-s$. For nodal/Neumann points, this does not matter.
\end{rem}

One can then define the $s$ set of $f_{n}$ analogously to the case
of domains:
\begin{equation}
\mathcal{S}_{s}\left(f_{n}\right)=\left\{ x\in\Gamma:f_{n}'\left(x\right)=sf_{n}\left(x\right)\right\} ,\label{eq:-28}
\end{equation}
where once again we allow the special case $s=\infty$. Similarly,
we denote the number of $s$ points by $\phi_{s}\left(f_{n}\right)$.
Moreover, if we partition our graph using the points of $\mathcal{S}_{s}\left(f_{n}\right)$,
we denote the corresponding number of $s$ domains by $\nu_{s}\left(f_{n}\right)$,
and the $s$ deficiency by
\begin{equation}
\mathcal{D}_{s}\left(f_{n}\right):=n-\nu_{s}\left(f_{n}\right).\label{eq:-29}
\end{equation}

Note that taking $s=\infty$ gives all the usual definitions of nodal
domains, and $s=0$ corresponds to Neumann domains.

The $s$ condition can also be written in terms of \emph{Prüfer angles};
By writing $s=\cot\left(\alpha\right)$ for $\alpha\in\left[0,\pi\right]$
(where $\alpha=0,\pi$ both correspond to the value $s=\infty$),
one may replace the parameter space $\RR$ with $\left[0,\pi\right]/\left\{ 0,\pi\right\} \cong S^{1}$.
These Prüfer angles will be used to define the $\delta_{s}$ family
below.

\bigskip

\subsection{The $\delta_{s}$ family of Hamiltonians\label{subsec:delta-s-definition}}

For the remainder of this work, we fix some orientation on the edges
of the graph, so that for every vertex of degree two, we have an ingoing
and outgoing edge. With this orientation in mind, we denote the two
sided values of a function $f$ at a vertex $v$ of degree two by
$f_{1,2}\left(v\right)$. Correspondingly, we denote the derivative
of $f_{2}$ at $v$ oriented into the \textbf{edge} by $f_{2}'\left(v\right)$
and the derivative of $f_{1}$ at $v$ oriented into the \textbf{vertex}
by $f_{1}'\left(v\right)$ (see Figure \ref{fig:Orientation}).

\begin{figure}
\includegraphics[scale=0.8]{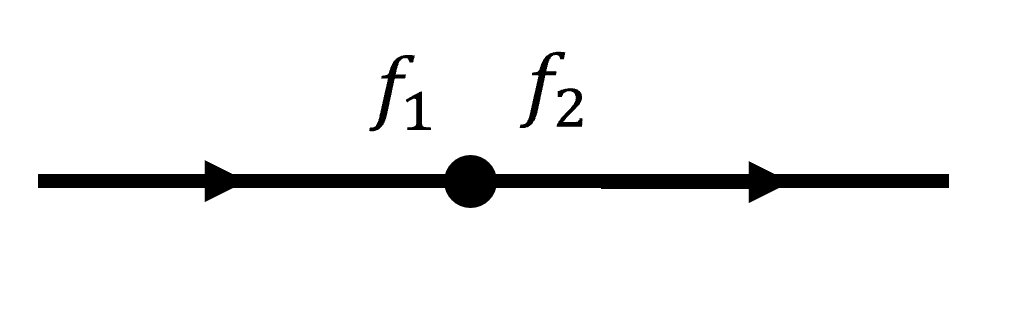}

\caption[Orientation on the edges.]{The orientation on the edges determines $f_{1},f_{2}$. $f_{1}'$
is directed into the vertex, while $f_{2}'$ is directed into the
edge. \label{fig:Orientation}}
\end{figure}

Throughout this work, we treat interior points of edges as artificial
vertices of degree two as well (and make no distinction between the
two).

Fix $s\in\RR$ and the corresponding Prüfer angle $\alpha$. Given
a subset $B\subset\mathcal{V}$ of vertices of degree two, we define
the following trace maps:
\begin{align}
 & \gamma_{1}^{s},\gamma_{1}^{s\star}\left(f\right):H^{2}\left(\Gamma\right)\rightarrow\ell^{2}\left(B\right)\label{eq:-75}\\
 & \left(\begin{array}{c}
\gamma_{1}^{s}\left(f\right)\left(v\right)\\
\gamma_{1}^{s\star}\left(f\right)\left(v\right)
\end{array}\right)=\left(\begin{array}{cc}
\cos\left(\alpha\right) & -\sin\left(\alpha\right)\\
\sin\left(\alpha\right) & \cos\left(\alpha\right)
\end{array}\right)\left(\begin{array}{c}
f_{1}\left(v\right)\\
f_{1}'\left(v\right)
\end{array}\right)\,\forall v\in B,\label{eq:-76}
\end{align}

and similarly define $\gamma_{2}^{s},\gamma_{2}^{s\star}$ with $f_{2}$.
We can heuristically think of the vector $\left(\gamma_{i}^{s},\gamma_{i}^{s\star}\right)$
at each vertex $v\in B$ as the rotation of the vector $\left(f_{i}\left(v\right),f_{i}'\left(v\right)\right)$
by the appropriate Prüfer angle.
\begin{example}
For $i=1,2$, $\gamma_{i}^{\infty}$ (which correspond to $\alpha=0$)
are the usual Dirichlet traces, while $\gamma_{i}^{0}$ (which correspond
to $\alpha=\frac{\pi}{2}$) are the Neumann traces. Moreover, $\gamma_{i}^{\infty\star}$
and $\gamma_{i}^{0\star}$ are the Neumann and Dirichlet traces, respectively.
\end{example}

The definition above can be extended so that $B$ may contain vertices
of degree one. If $v$ is of degree one, then only two of the four
trace maps above are naturally defined ($\gamma_{i}^{s},\gamma_{i}^{s\star}$
for some $i\in\left\{ 1,2\right\} $). In this case, for $j\neq i$,
we use the convention $\gamma_{j}^{s}=\gamma_{i}^{s}$ and $\gamma_{j}^{s\star}=0$.

From Definition \ref{def:s-points} we see that the $s$ points of
a function $f$ are characterized by $\gamma_{1}^{s}\left(f\right)=\gamma_{2}^{s}\left(f\right)=0$.
In the case where the vertex $v$ is of degree one or that the values
of the two sided traces agree on $v$, we simply denote
\begin{align}
 & \gamma^{s}:=\gamma_{1}^{s}=\gamma_{2}^{s}.\label{eq:-95}
\end{align}

Now, let $B\subset\mathcal{V}$ be a subset of vertices of degree
two as above. We define a one-parameter family of self-adjoint operators
$\left(H^{s}\left(t\right)\right)_{t\in\RR}$ by setting $H^{s}\left(t\right)=-\frac{d^{2}}{dx^{2}}$,
with domain consisting of functions $f\in H^{2}\left(\Gamma\right)$
satisfying the Neumann-Kirchhoff condition at $\mathcal{V}\backslash B$,
and satisfying the following vertex condition at $B$:
\begin{align}
 & \gamma_{1}^{s}\left(f\right)=\gamma_{2}^{s}\left(f\right),\label{eq:-6-1-1}\\
 & \gamma_{2}^{s\star}\left(f\right)-\gamma_{1}^{s\star}\left(f\right)=t\gamma^{s}\left(f\right),\label{eq:-7-1-1}
\end{align}
where the case $t=\infty$ is interpreted as satisfying the $s$ condition
at $B$:
\begin{equation}
\gamma_{1}^{s}\left(f\right)=\gamma_{2}^{s}\left(f\right)=0.\label{eq:-5-1}
\end{equation}

One can verify that $H^{s}\left(0\right)$ is simply the Neumann-Kirchhoff
Laplacian for all $s$, and so we sometimes denote $H_{0}:=H^{s}\left(0\right)$.
\begin{defn}
For fixed $s$, we call $\left(H^{s}\left(t\right)\right)_{t\in\RR}$
the $\delta_{s}$ family. 
\end{defn}

\begin{example}
For $s=\infty$, the $\delta_{\infty}$ family is exactly the $\delta$
type condition which was used to define the Robin-Neumann gap. For
$s=0$, the $\delta_{0}$ family is exactly the $\delta'$ type condition.
\end{example}

\bigskip

\subsection{The Robin map\label{subsec:DTN-definition}}

In this subsection we define the \emph{Robin map}, which serves as
a generalization to the two sided Dirichlet to Neumann map presented
in \cite{BerCoxMar_lmp19,Berkolaiko2022}.

Let $B\subset\mathcal{V}$ be a subset of vertices of degree one or
two. We usually take $B$ to be the $s$ set of some fixed eigenfunction.
The points in $B$ partition $\Gamma$ into subgraphs $\left\{ D_{i}\right\} _{i=1}^{n}$,
each one with a corresponding boundary $\left\{ \partial D_{i}\right\} _{i=1}^{n}$.

Denote $B_{i}=\partial D_{i}\cap B$. We define a collection of weights
$\chi_{i}:B_{i}\rightarrow\left\{ \pm1\right\} $, such that $\chi_{i}\left(v\right)=1$
if the orientation on $D_{i}$ points into $v$, and $\chi_{i}\left(v\right)=-1$
if the orientation on $D_{i}$ points out of $v$ (see Figure \ref{fig:chimap}).

\begin{figure}
\includegraphics[scale=0.7]{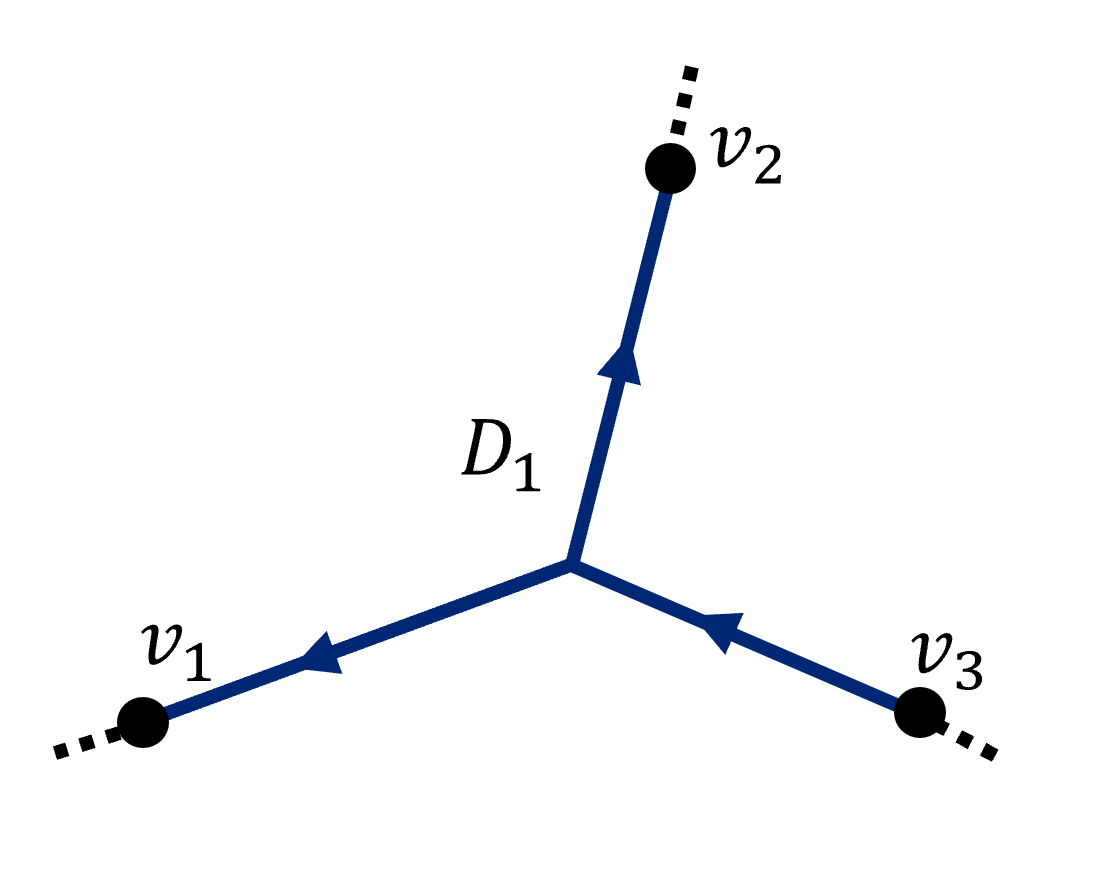}

\caption[Demonstration of the map $\chi_{D_{1}}$.]{Demonstration of the map $\chi_{D_{1}}$ according to the orientation
on the domain $D_{1}$ (The dashed lines are additional edges of the
graph). Here, $\chi_{D_{1}}\left(v_{1}\right)=\chi_{D_{1}}\left(v_{2}\right)=1$,
and $\chi_{D_{1}}\left(v_{3}\right)=-1$. \label{fig:chimap}}
\end{figure}

Fix $s\in\overline{\mathbb{R}}$ and $c\notin\spec{H^{s}\left(\infty\right)}$.
Let us first focus on defining the Robin map $\Lambda_{s}\left(c\right)$
on a single domain $D_{i}$ (which comes with its weight function
$\chi_{D_{i}}$), and then define it globally. More precisely, we
define a map $\Lambda_{s}^{D_{i}}$ which acts on the finite dimensional
vector space $\ell^{2}\left(B_{i}\right)$. 

For a vector $w\in\ell^{2}\left(B_{i}\right)$, the boundary value
problem
\begin{align}
 & -\frac{\partial^{2}u}{\partial x^{2}}=cu,\,\text{on }D_{i},\label{eq:-13-1}\\
 & \gamma^{s}\left(u\right)=w\,\text{on }B_{i},\label{eq:-14-1}\\
 & u\text{ satisfies the Neumann-Kirchhoff condition (\ref{enu:NK}) at }\partial D_{i}\backslash B_{i},
\end{align}
has a unique solution $u\in H^{2}\left(D_{i}\right)$, and we thus
let
\begin{equation}
\Lambda_{s}^{D_{i}}\left(c\right)w=\gamma^{s\star}\left(u\right).\label{eq:-15-1}
\end{equation}

In other words, $\Lambda_{s}^{D_{i}}\left(c\right)$ sends the $\gamma^{s}$
data of $u$ to the $\gamma^{s\star}$ data of $u$.

For each domain $D_{i}$, $\ell^{2}\left(B_{i}\right)$ can be identified
as a subspace of $\ell^{2}\left(B\right)\cong\mathbb{C}^{\left|B\right|}$.
If two domains $D_{i},D_{j}$ have a common boundary point, then $\ell^{2}\left(B_{i}\right)\cap\ell^{2}\left(B_{j}\right)\neq\left\{ 0\right\} $.

With this identification in mind, we can define the combined Robin
map $\Lambda_{s}\left(c\right)$ by the following block form:
\begin{equation}
\Lambda_{s}\left(c\right):=\sum_{D_{i}}\chi_{D_{i}}\Lambda_{s}^{D_{i}}\left(c\right).\label{eq:-79}
\end{equation}

This gives the following effective formula for the Robin map:
\begin{equation}
\Lambda_{s}\left(c\right)\gamma^{s}\left(f\right)=\gamma_{1}^{s\star}\left(f\right)-\gamma_{2}^{s\star}\left(f\right),\label{eq:-16-1}
\end{equation}
where the minus sign is due to the weights $\chi_{D}$ (see Figure
\ref{fig:DTNdomains}). This formula holds for vertices of degree
one or two, while keeping in mind that in the case of vertices of
degree one, one of the traces in the right hand side is zero by convention.

The motivation for adding the weights will become clear in Lemma \ref{lem:DTN-correspondence},
which shows the connection between the Robin map and the $\delta_{s}$
family (as already hinted by Equation (\ref{eq:-7-1-1})).

\begin{figure}
\includegraphics[scale=0.65]{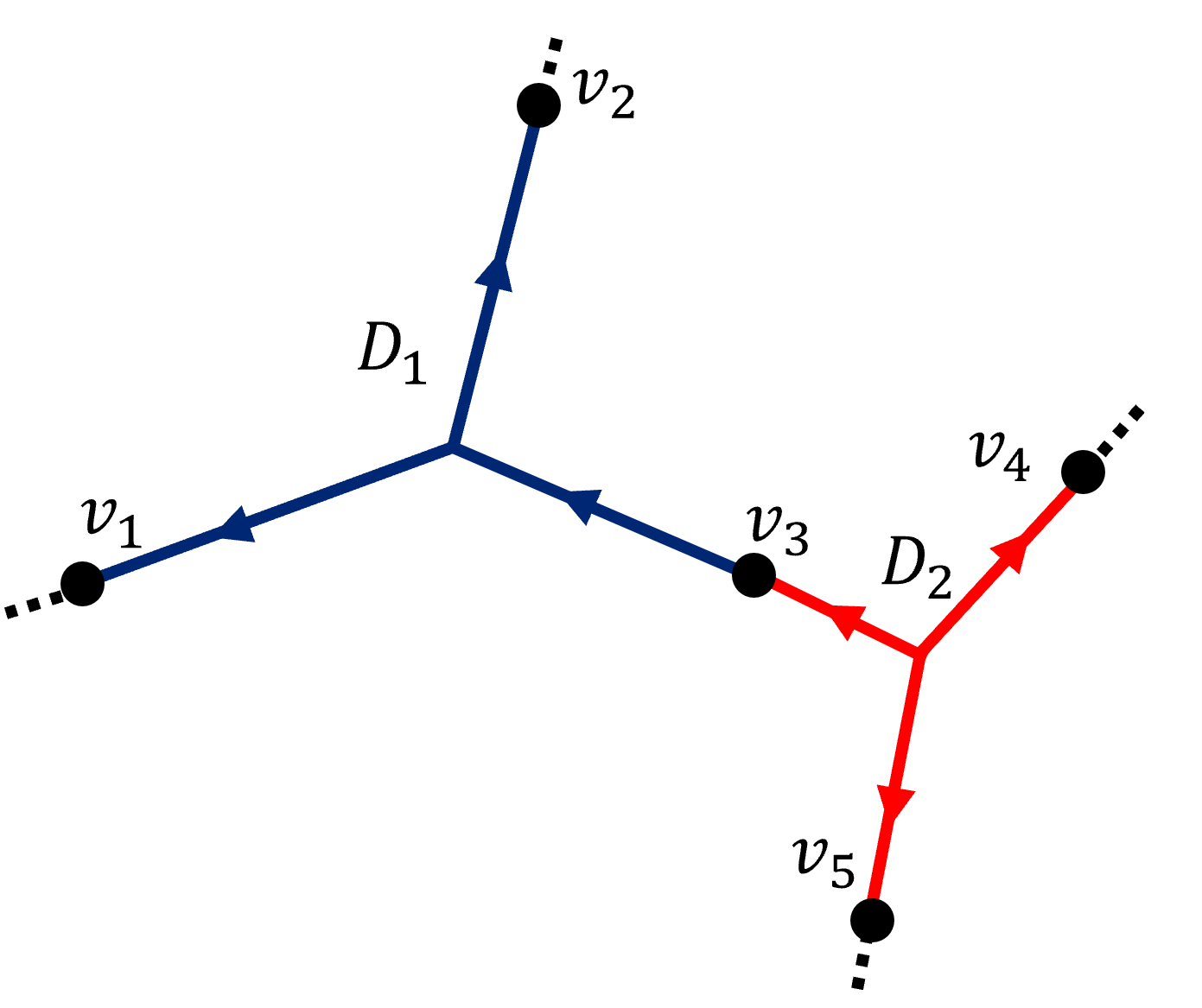}

\caption[Construction of $\Lambda_{s}\left(c\right)$ for a graph consisting
of two domains.]{Construction of $\Lambda_{s}\left(c\right)$ for a graph consisting
of two domains -- $D_{1},D_{2}$ (The dashed lines are additional
edges of the graph). Here, $B=\left\{ v_{1},...,v_{5}\right\} $,
$B_{1}=\left\{ v_{1},v_{2},v_{3}\right\} $, $B_{2}=\left\{ v_{3},v_{4,},v_{5}\right\} $.
The two domains share a common vertex $v_{3}$. The maps $\Lambda_{s}^{D_{1}}\left(c\right)$
and $\Lambda_{s}^{D_{2}}\left(c\right)$ are defined on different
three-dimensional subspaces of $\ell^{2}\left(B\right)\protect\cong\mathbb{C}^{5}$,
whose intersection is one-dimensional (and can be thought of as $\ell^{2}\left(v_{3}\right)$).
Since $\chi_{D_{1}}\left(v_{3}\right)=-1$ and $\chi_{D_{2}}\left(v_{3}\right)=1$,
we get the expression for $\Lambda_{s}\left(c\right)$ as in Equation
(\ref{eq:-16-1}). \label{fig:DTNdomains}}
\end{figure}

\begin{example}
For $s=\infty$, we get that at each $v\in B$:
\begin{align}
 & \gamma^{\infty}\left(f\right)\left(v\right)=f\left(v\right),\gamma^{\infty\star}\left(f\right)\left(v\right)=f'\left(v\right)\label{eq:-97}\\
\Rightarrow & \Lambda_{\infty}\left(c\right)f\left(v\right)=f_{1}'\left(v\right)-f_{2}'\left(v\right).\label{eq:-17-1}
\end{align}
\end{example}

which is exactly the two sided Dirichlet to Neumann map which appears
in (\ref{eq:-84}).
\begin{rem}
\label{rem:DTN-matrix}Since the Robin map acts on $\ell^{2}\left(B\right)$,
then it can be simply thought of as a square matrix of size $\left|B\right|$.
A concrete computation of such a matrix is presented in Appendix \ref{sec:DTN-comp}.
\end{rem}

\bigskip

\subsection{The spectral flow\label{subsec:SF-definition}}

Given a one-parameter family of Hamiltonians $\left(H\left(t\right)\right)_{t\in I}$
(where $I=\left[a,b\right]$ is some closed interval), whose spectral
curves depend continuously on the parameter $t$, one can define a
quantity known as the \emph{spectral flow}. Colloquially, this quantity
captures the number of oriented intersections of spectral curves with
respect to a given horizontal cross section along $I$ (see Figure
\ref{fig:The-spectral-flow}).

More precisely, given $\lambda\in\mathbb{R}$ and a partition $a=t_{0}<t_{1}<...<t_{N}=b$,
and $N$ intervals $\left(\left[a_{l},b_{l}\right]\right)_{l=1}^{N}$
with $a_{l}<\lambda<b_{l}$ such that
\begin{equation}
a_{l},b_{l}\notin\spec{H\left(t\right)}_{t=a}^{b},\forall t\in\left[t_{l-1},t_{l}\right],1\leq l\leq N,\label{eq:-30}
\end{equation}
we define the spectral flow through $\lambda$ by
\begin{equation}
\sf{\lambda}\left(I\right)=\sum_{l=1}^{N}\sum_{a_{l}\leq\alpha<\lambda}\left(\dim\ker\left(H\left(t_{l-1}\right)-\alpha\right)-\dim\ker\left(H\left(t_{l}\right)-\alpha\right)\right).\label{eq:-31}
\end{equation}

\begin{figure}
\includegraphics[scale=0.7]{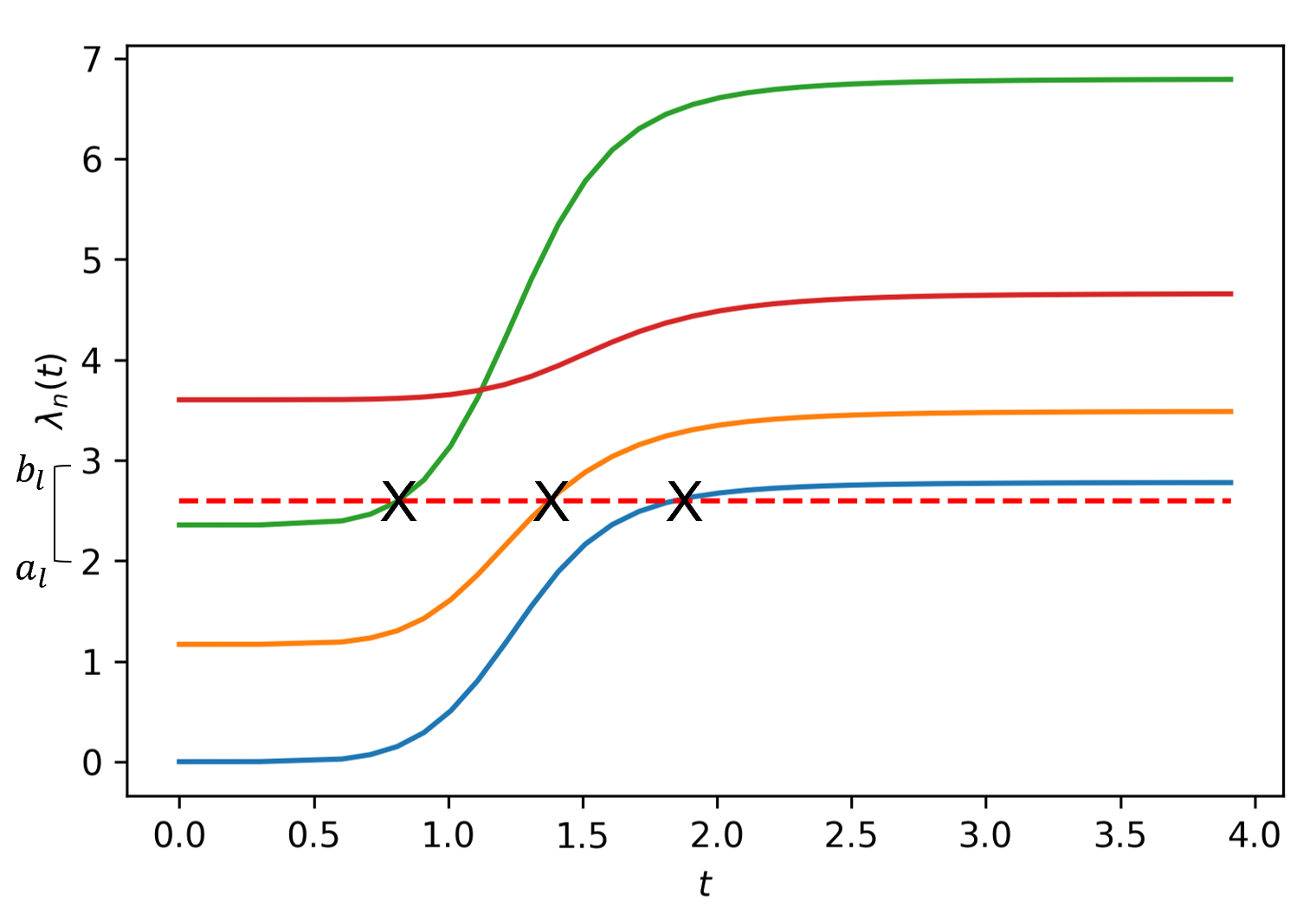}

\caption[The spectral flow.]{The spectral flow of the spectral curves through the horizontal cross
section $\lambda=2.5$. Here, $I=\left[0,4\right]$. One can choose,
for instance, a partition with $t_{l}=0.5l$ (as given by the hatch
marks on the horizontal axis) and $\left[a_{l},b_{l}\right]=\left[2,3\right],\forall l\in\left\{ 0,...,8\right\} $.
The marked intersections inside $\left[0.5,1\right]$, $\left[1,1.5\right]$
and $\left[1.5,2\right]$ give total spectral flow of three. \label{fig:The-spectral-flow}}
\end{figure}

One can show that the spectral flow does not depend on the choice
of partition, see e.g. \cite{B.BoossBavnbek2018}.

A nice recent work which lays the groundwork for the spectral flow
machinery in the context of quantum graphs is given in \cite{LatSuk_ams20}.

\newpage{}

\section{Statement of Main Results \label{sec:Main-Results}}

\subsection{Results concerning the Robin-Neumann gap}
\begin{thm}
\label{thm:1.RNG-Lipschitz}1. The sequence of functions $d_{n}\left(\sigma\right)$
is Lipschitz continuous in $\sigma\in\mathbb{R}$ with a uniform Lipschitz
constant.

2. For any compact set $A\subset\mathbb{R}$, the sequence of functions
$d_{n}\left(\sigma\right)$ is uniformly bounded on $A$, and there
exists a subsequence of functions $d_{n_{k}}\left(\sigma\right)$
which converges uniformly on $A$.
\end{thm}

\begin{defn}
\label{def:Cesaro}Given a sequence of numbers $\left(a_{n}\right)_{n=1}^{\infty}$,
we define the Cesaro sum (or Cesaro mean) of $a_{n}$ as
\begin{equation}
\left\langle a\right\rangle _{n}=\lim_{N\rightarrow\infty}\frac{1}{N}\sum_{n=1}^{N}a_{n},\label{eq:-40}
\end{equation}
\end{defn}

assuming that the limit exists.

The next theorem concerns with the Cesaro mean of the Robin-Neumann
gap: 
\begin{thm}
\label{thm:RNG-mean}$\left\langle d\right\rangle _{n}\left(\sigma\right)$
exists for all $\sigma\in\mathbb{R}$ and satisfies
\begin{equation}
\left\langle d\right\rangle _{n}\left(\sigma\right)=\frac{2\sigma}{L}\sum_{v\in\VR}\frac{1}{\deg\left(v\right)}.\label{eq:}
\end{equation}
\end{thm}

\begin{figure}
\includegraphics[scale=0.43]{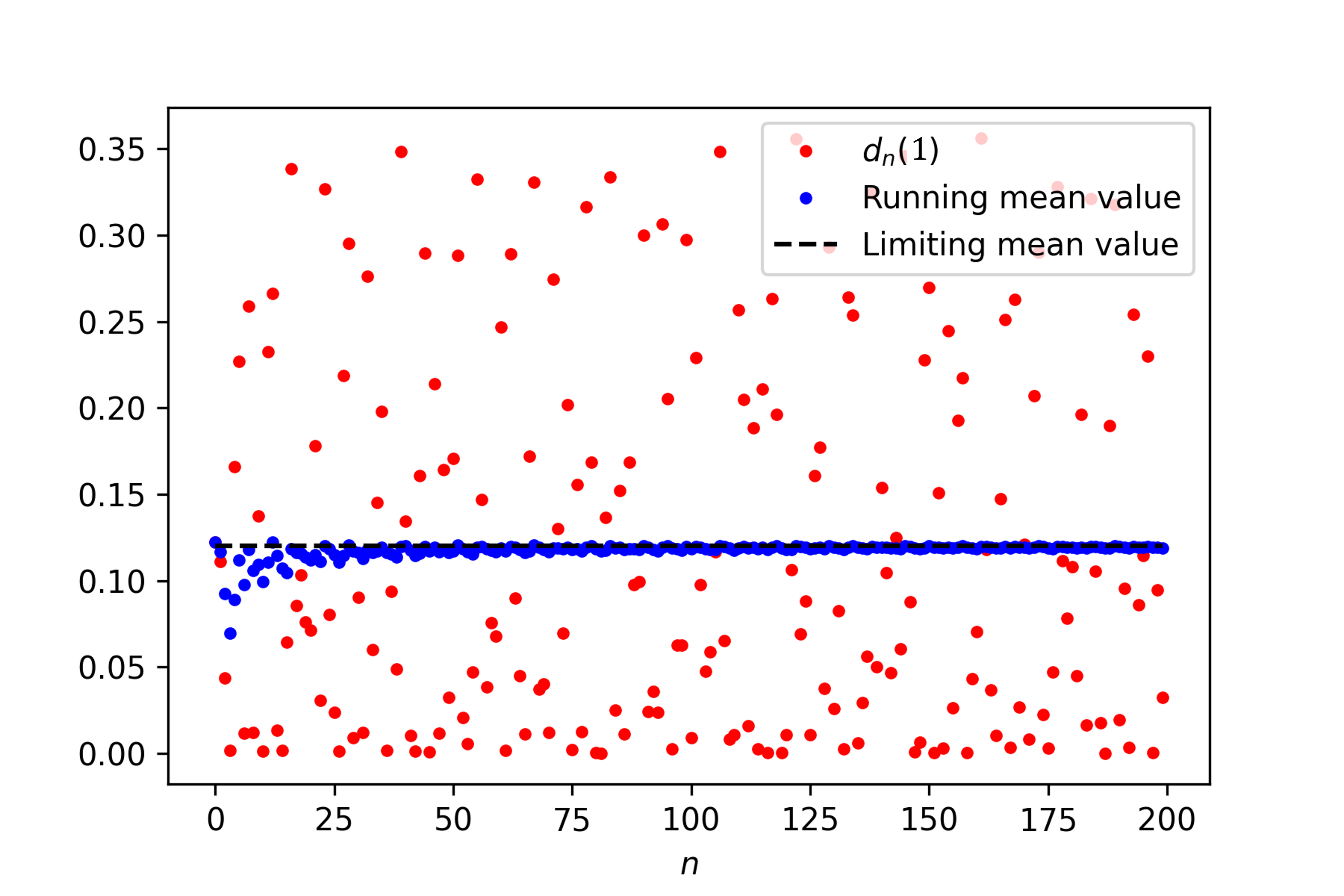}

\caption[Demonstration of Theorem \ref{thm:RNG-mean}.]{A demonstration of Theorem \ref{thm:RNG-mean}. For a star graph with
a single $\delta$ point at the central vertex, we plot the first
$200$ values of $d_{n}\left(\sigma\right)$ with $\sigma=1$ (red
points). We also plot the running mean value of the sequence (blue
points). The plot shows the rapid convergence of the mean value to
the limiting value $\frac{2}{L\cdot\deg\left(v\right)}$ (dashed black
line), as suggested by the theorem. One can also see that the sequence
is bounded, as suggested by Theorem \ref{thm:1.RNG-Lipschitz}.\label{fig:RNG-thm-demo}}
\end{figure}

For a demonstration of Theorems \ref{thm:1.RNG-Lipschitz} and \ref{thm:RNG-mean},
see Figure \ref{fig:RNG-thm-demo}.

In the process of proving Theorem \ref{thm:RNG-mean}, we also prove
the following local Weyl law, which estimates the Cesaro mean of several
quantities related to the eigenfunctions:
\begin{thm}
\label{thm:Weyl-law}Denote the $L^{2}$ normalized eigenfunctions
of the Neumann-Kirchhoff Laplacian by $f_{n}$ with eigenvalues $k_{n}^{2}$.
Then at each vertex $v\in\mathcal{V}$
\begin{equation}
\left\langle \left|f\left(v\right)\right|^{2}\right\rangle _{n}=\frac{2}{\deg\left(v\right)L}.\label{eq:-15-2}
\end{equation}

Moreover, if we write $f_{n}$ on each edge $e$ as
\begin{equation}
f_{n}^{\left(e\right)}\left(x\right)=A_{e}e^{ik_{n}x}+A_{\hat{e}}e^{ik_{n}\ell_{e}}e^{-ik_{n}x},\label{eq:-45-1}
\end{equation}
where the amplitudes $A_{e},A_{\hat{e}}$ depend on $n$, then
\begin{align}
 & \left\langle \left|A_{e}\right|^{2}\right\rangle _{n}=\left\langle \left|A_{\hat{e}}\right|^{2}\right\rangle _{n}=\frac{1}{2L},\forall e\in\mathcal{E}\label{eq:-16-2}\\
 & \left\langle A_{j}\overline{A_{e}}\right\rangle _{n}=\left\langle A_{\hat{j}}\overline{A_{\hat{e}}}\right\rangle _{n}=0\text{,\,\ensuremath{\forall}}j,e\in\mathcal{E},j\neq e.\label{eq:-17-2}
\end{align}
\end{thm}

\begin{rem}
A similar local Weyl law was recently obtained in \cite{Borthwick2022}
via heat kernel methods.

\bigskip
\end{rem}

\subsection{Results concerning the $\delta_{s}$ family}

The next several results are concerned with the relation between the
$\delta_{s}$ family and properties of eigenfunctions on the graph,
as well as geometric properties of the graph.

The first result is a generalization of the result presented for domains
in Equation (\ref{eq:-84}):
\begin{thm}
\label{thm:SF-index} Let $\left(\lambda_{n},f_{n}\right)$ be a generic
eigenpair of $H_{0}$ such that $\lambda_{n}>\frac{\pi}{\ell_{min}}$,
where $\ell_{min}$ is the length of the shortest edge of the graph.
Then for $\epsilon>0$ small enough, the $s$ deficiency is given
by the following formula:
\begin{equation}
\mathcal{D}_{s}\left(f_{n}\right)=\phi_{\infty}\left(f_{n}\right)-Pos\left(\Lambda_{s}\left(\lambda_{n}+\epsilon\right)\right),\label{eq:-6-1}
\end{equation}
where the Robin map $\Lambda_{s}\left(\lambda_{n}+\epsilon\right)$
is evaluated at the set of $s$ points of $f_{n}$, $Pos\left(\Lambda_{s}\left(\lambda_{n}+\epsilon\right)\right)$
denotes the number of its positive eigenvalues, and $\phi_{\infty}\left(f_{n}\right)$
is the number of nodal points of $f_{n}$. \\
For $s=\infty$, the following formula for the nodal deficiency holds:
\begin{equation}
\mathcal{D}_{\infty}\left(f_{n}\right)=Mor\left(\Lambda_{\infty}\left(\lambda_{n}+\epsilon\right)\right),\label{eq:-7-1}
\end{equation}
where $Mor\left(\Lambda_{\infty}\left(\lambda_{n}+\epsilon\right)\right)$
is the Morse index, or number of its negative eigenvalues.
\end{thm}

\begin{rem}
\label{rem:unnecessary}For $s=\infty$, the assumption that $\lambda_{n}>\frac{\pi}{\ell_{min}}$
is unnecessary, and so is the assumption that $f_{n}$ is generic. 
\end{rem}

\begin{thm}
\label{prop:SF-points}Consider the $\delta_{s}$ family placed on
the vertex set $B$. Then for every $c\notin\spec{H^{s}\left(\infty\right)}$
we have that
\begin{equation}
\sf c^{\delta_{s}}\left[-\infty,\infty\right]=\left|B\right|.\label{eq:-33}
\end{equation}

In particular, if the $\delta_{s}$ family is placed at the $s$ points
of an eigenfunction $f_{n}$, then
\begin{equation}
\sf c^{\delta_{s}}\left[-\infty,\infty\right]=\phi_{s}\left(f_{k}\right).\label{eq:-99}
\end{equation}
\end{thm}

\begin{rem}
\label{rem:Maslov-remark}In fact, one can show that Theorem \ref{prop:SF-points}
holds for every $c\in\mathbb{R}$ (without further restrictions).
The proof requires more abstract tools (see \cite{Band}). Nevertheless,
we discuss those briefly in Section \ref{sec:Discussion}.

The next two theorems relate the spectral flow to the topology of
the graph (see demonstrations in Figures \ref{fig:Demo-betti-1} and
\ref{fig:Demo-betti-2}).
\end{rem}

\begin{thm}
\label{thm:SF-Betti1}Consider the $\delta_{0}$ family placed at
an arbitrary subset of points of degree two $B\subset\Gamma$. Denote
the metric graph obtained by cutting $\Gamma$ at $B$ by $\Gamma_{cut}$.
Then for $\epsilon>0$ small enough
\begin{equation}
\sf{\epsilon}^{\delta_{0}}\left[0,\infty\right]=\beta_{\Gamma}-\beta_{\Gamma_{cut}}.\label{eq:-34}
\end{equation}

\begin{figure}
\caption[Demonstration of Theorem \ref{thm:SF-Betti1}.]{\label{fig:Betti}Demonstration of Theorem \ref{thm:SF-Betti1} on
a glasses graph for two choices of $\delta_{0}$ points.\label{fig:Demo-betti-1}}
\subfloat[The points are chosen so that cutting the graph \textquotedblleft destroys\textquotedblright{}
two cycles, and so the spectral flow through the dashed line $y=\epsilon$
is two.]{\includegraphics[scale=0.7]{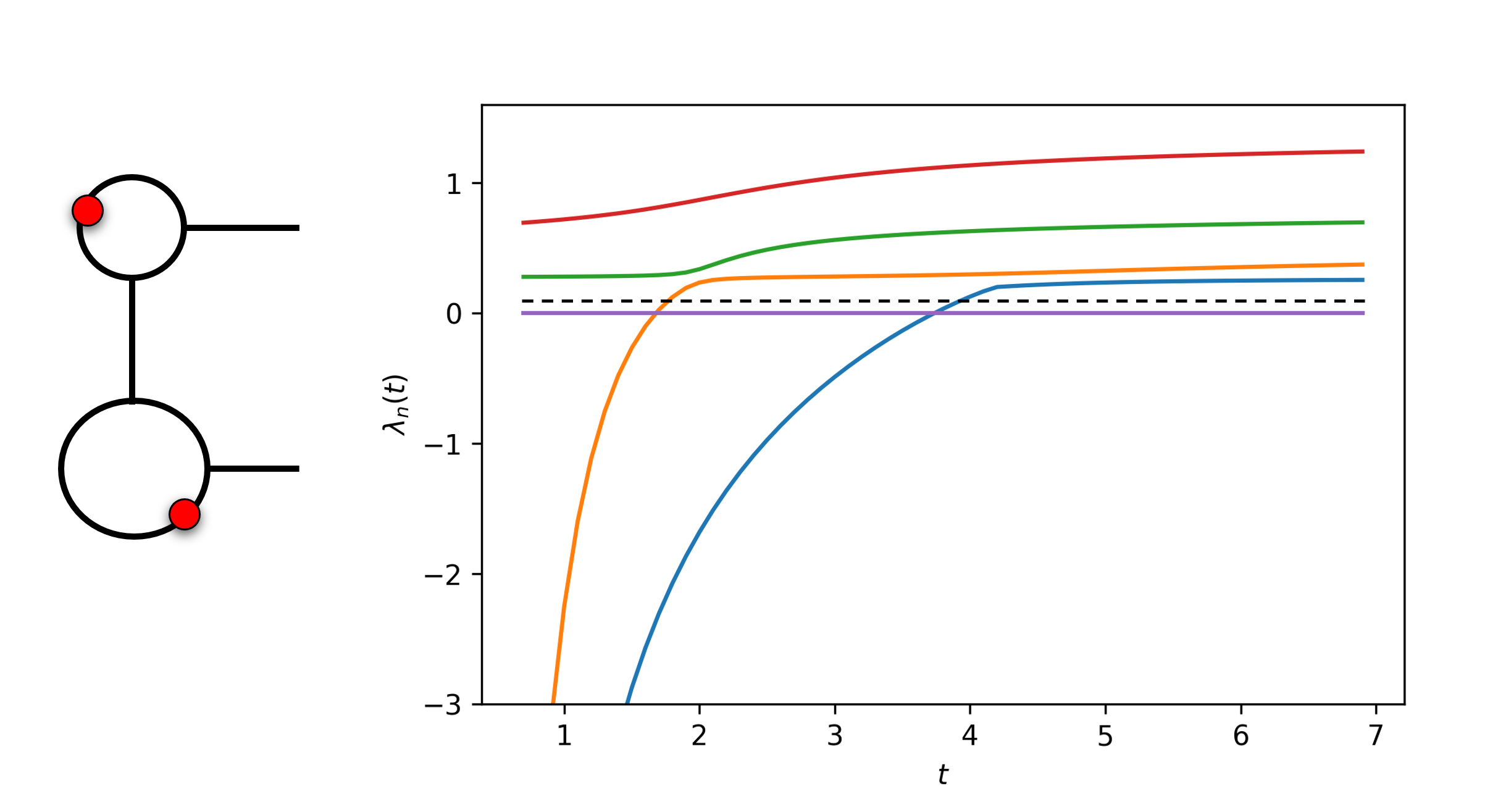}

}

\subfloat[The points are chosen so that cutting the graph \textquotedblleft destroys\textquotedblright{}
one cycle, and so the spectral flow through the dashed line $y=\epsilon$
is one. ]{\includegraphics[scale=0.7]{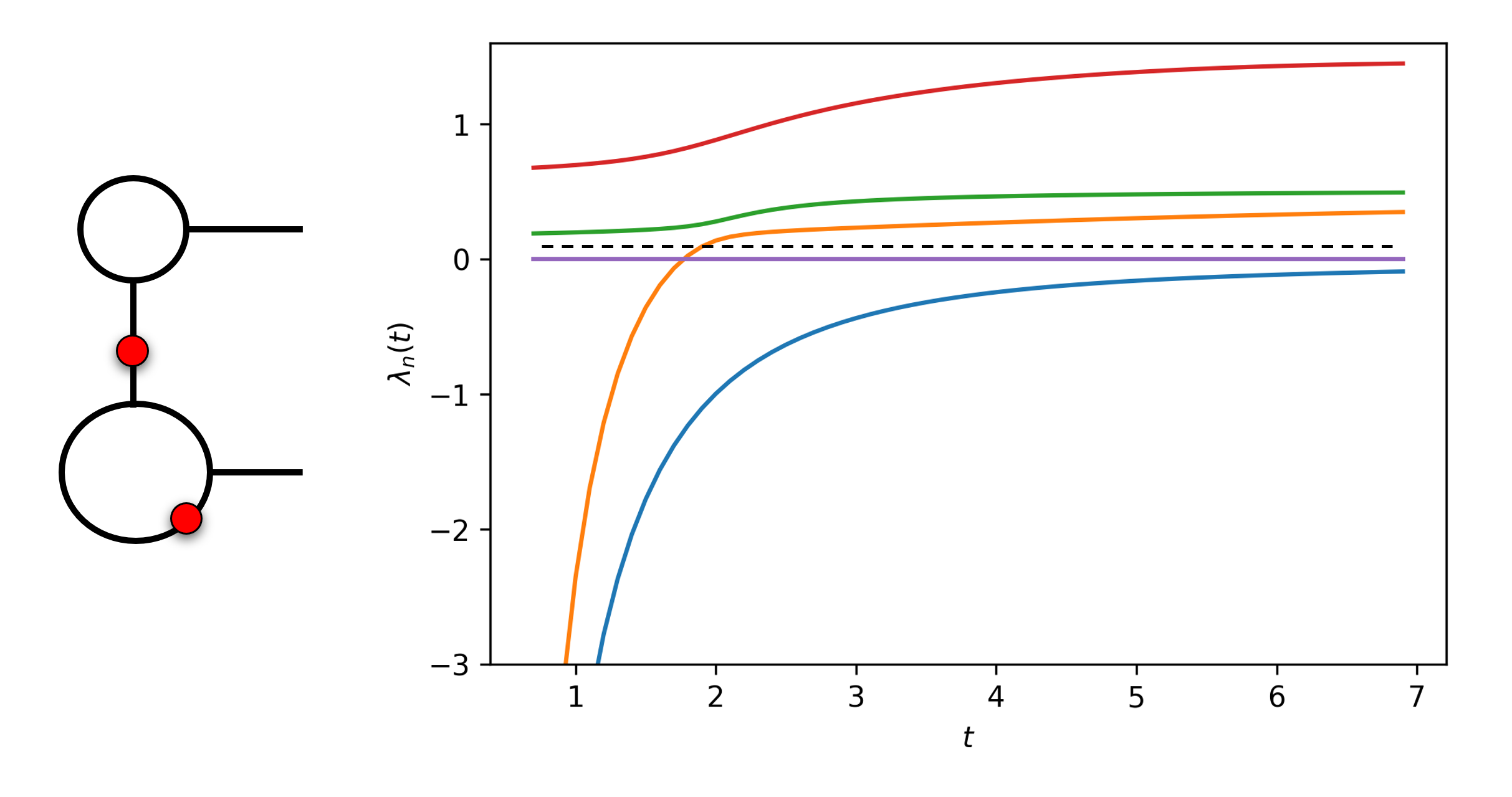}

}
\end{figure}
\end{thm}

\begin{thm}
\label{prop:SF-Betti2} If $\left(\lambda_{n},f_{n}\right)$ is a
generic eigenpair of $H_{0}$ as in Theorem \ref{thm:SF-index}, then
for the $\delta_{s}$ family with vertex set $B$ chosen as the $s$
points of $f_{n}$, there are exactly $\beta_{\Gamma}$ transversal
intersections of the spectral curves with the horizontal line $\lambda_{n}$.
\\
In particular, for tree graphs, there are no such intersections.
\end{thm}

\begin{figure}
\includegraphics[scale=0.7]{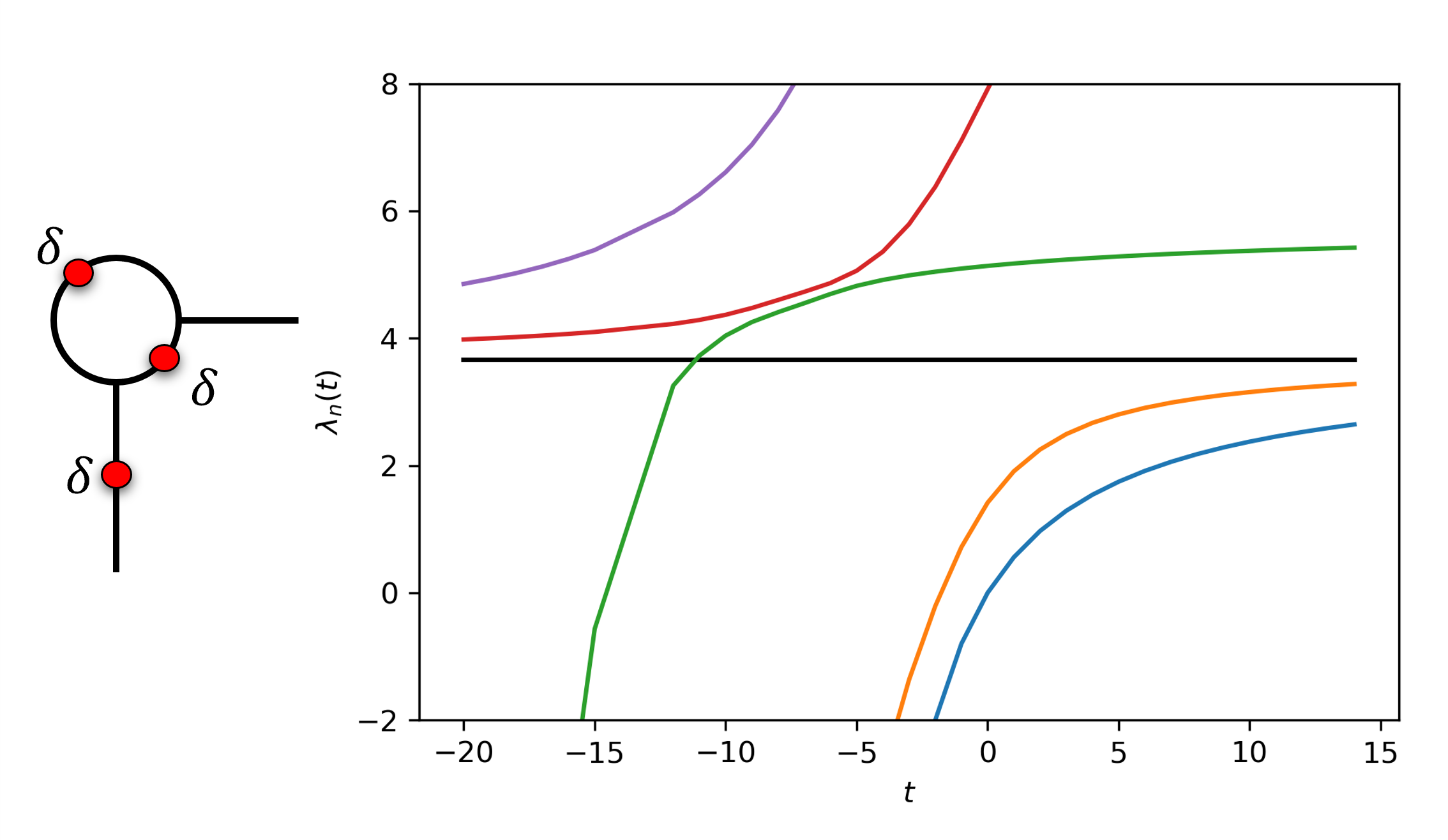}

\caption[Demonstration of Theorem \ref{prop:SF-Betti2}.]{Demonstration of Theorem \ref{prop:SF-Betti2} on a graph with $\beta=1$.
The $\delta$ points are placed at the nodal points of an eigenfunction,
and exactly one spectral curve intersects the horizontal line $\lambda_{k}$
(marked in black). \label{fig:Demo-betti-2}}
\end{figure}

\newpage{}

\section{Tools for proofs of Robin-Neumann gap Theorems \ref{thm:1.RNG-Lipschitz},
\ref{thm:RNG-mean} and \ref{thm:Weyl-law} \label{sec:RNG-tools}}

\subsection{Scattering formalism and the secular equation\label{subsec:Scattering-formalism}}

For this section and all sections concerning the RNG, we fix our family
of Hamiltonians to be the $\delta$ family, $H\left(t\right):=H^{\infty}\left(t\right)$.

An eigenfunction $f$ of $H_{0}:=H\left(0\right)$ with eigenvalue
$k^{2}>0$ can be written on the $j$th edge as
\begin{equation}
f_{j}\left(x\right)=a_{j}e^{ikx}+a_{\hat{j}}e^{ik\ell_{j}}e^{-ikx},\label{eq:-6}
\end{equation}

as seen in \cite{Berkolaiko_qg-intro17}. Thus, $f$ can be described
by the vector of coefficients $\vec{a}=\left(a_{1},a_{\hat{1}},...,a_{E},a_{\hat{E}}\right)\in\mathbb{C}^{2E}$,
which depends on the wave number $k$ (although we omit this $k$
dependence to avoid clutter).

These coefficients are known as the \emph{scattering amplitudes}.
We can think of $a_{j}$ as representing an outgoing wave along the
directed edge $j$, and of $a_{\hat{j}}$ as representing an ingoing
wave. Meaning, we think of $\hat{j}$ as the reversal of the directed
edge $j$, as originally described by Kottos and Smilanksy in \cite{KotSmi_prl00,KotSmi_ap99}.

\begin{figure}
\includegraphics[scale=0.6]{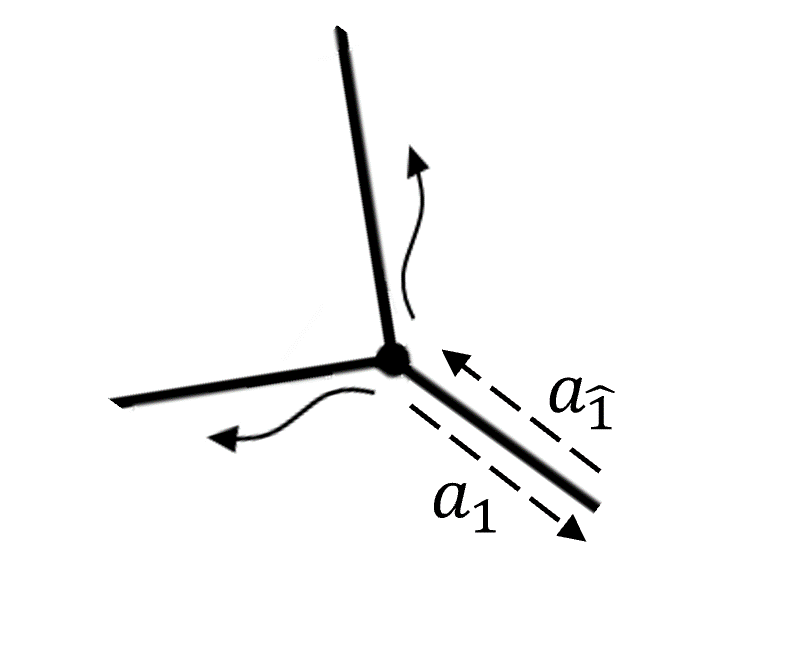}

\caption{Scattering amplitudes at a vertex.}
\end{figure}

If we write the vertex conditions imposed at a vertex $v$ at $x=0$,
with the edges exiting $v$ indexed from $1$ to $\deg\left(v\right)$,
the Neumann-Kirchhoff condition in (\ref{eq:-15}, \ref{eq:-16})
gives
\begin{align}
 & a_{1}+a_{\hat{1}}e^{ik\ell_{1}}=...=a_{\deg\left(v\right)}+a_{\hat{\deg\left(v\right)}}e^{ik\ell_{\deg\left(v\right)}}\,\,\left(\text{continuity}\right),\label{eq:-35}\\
 & \sum_{i=1}^{\deg\left(v\right)}a_{i}-\sum_{i=1}^{\deg\left(v\right)}a_{\hat{i}}e^{ik\ell_{i}}=0\,\,\left(\text{current conservation}\right).\label{eq:-36}
\end{align}

The first condition implies that for any fixed edge $1\leq j\leq\deg\left(v\right)$:
\begin{equation}
\deg\left(v\right)\left(a_{j}+a_{\hat{j}}e^{ik\ell_{j}}\right)=\sum_{i=1}^{\deg\left(v\right)}a_{i}+\sum_{i=1}^{\deg\left(v\right)}a_{\hat{i}}e^{ik\ell_{i}}.\label{eq:-37}
\end{equation}

Combining this with the expression for $\sum_{i=1}^{\deg\left(v\right)}a_{i}$
in (\ref{eq:-36}) and extracting $a_{j}$ gives
\begin{equation}
a_{j}=\frac{2}{\deg\left(v\right)}\sum_{i=1}^{\deg\left(v\right)}a_{\hat{i}}e^{ik\ell_{i}}-a_{\hat{j}}e^{ik\ell_{j}}.\label{eq:-12}
\end{equation}

This gives a simple linear relation between the different scattering
amplitudes.

Given two directed edges $j,j'$, we say that $j\rightarrow j'$ at
$v$ if the end vertex of $j$ is $v$ and the starting vertex of
$j'$ is $v$. With this in mind, we can define the following square
matrix of size $2E$:

\begin{equation}
S_{j'j}=\begin{cases}
\frac{2}{\deg\left(v\right)}-1 & j'=\hat{j}\\
\frac{2}{\deg\left(v\right)} & j\rightarrow j'\text{ at \ensuremath{v} and }j'\neq\hat{j}\\
0 & \text{Otherwise}
\end{cases}\label{eq:-3-1-1}
\end{equation}

We can thus write the linear relation in (\ref{eq:-12}) in the following
compact form:

\begin{equation}
Se^{ik\L}\vec{a}=\vec{a},\label{eq:-2}
\end{equation}

where $\L$ is the diagonal matrix $diag\left(\ell_{1},\ell_{1},\ell_{2},\ell_{2},...,\ell_{E},\ell_{E}\right)$.
The unitary matrix $S$ is known as the \emph{bond scattering matrix}
(For more details, see \cite{BerKuc_graphs,Alon_PhDThesis,Berkolaiko_qg-intro17,BanHarJoy_prep12,GnuSmi_ap06}).

Equation \ref{eq:-2} gives us the following fundamental result:
\begin{thm}
\label{thm:secular-equation}(\cite{KotSmi_prl00,KotSmi_ap99}) $k^{2}>0$
is an eigenvalue of $H_{0}$ if and only if $\det\left(I-Se^{ik\L}\right)=0$.
\end{thm}

More generally, for the $\delta$ Hamiltonian $H\left(t\right)$ with
$t\in\mathbb{R}$, a similar computation shows that $k^{2}>0$ is
an eigenvalue of $H\left(t\right)$ if and only if $\det\left(I-S^{\left(t\right)}\left(k\right)e^{ik\L}\right)=0$,
where this time the scattering matrix $S^{\left(t\right)}\left(k\right)$
depends on $k$, and is given by the following expression (see Equation
($3.6$) in \cite{GnuSmi_ap06}):
\begin{equation}
S_{j'j}^{\left(t\right)}=\begin{cases}
\frac{2}{\deg\left(v\right)+\frac{it}{k}}-1 & j'=\hat{j}\\
\frac{2}{\deg\left(v\right)+\frac{it}{k}} & j\rightarrow j'\text{ at \ensuremath{v} and }j'\neq\hat{j}\\
0 & \text{Otherwise}
\end{cases}\label{eq:-32}
\end{equation}

Throughout the following sections, we shall denote the $L^{2}$ normalized
eigenfunctions of $H\left(t\right)$ by $f_{n}^{\left(t\right)}$.

\bigskip

\subsection{The secular manifold and the Barra-Gaspard measure\label{subsec:secular-manifold}}

Motivated by the discussion above, we can define the following function
on $\mathbb{R}^{E}$:
\begin{equation}
\tilde{F}\left(\vec{y}\right)=\det\left(I-Se^{iy}\right),\label{eq:-3}
\end{equation}

where $y$ is the diagonal matrix $diag\left(y_{1},y_{1},...,y_{E},y_{E}\right)$.
This function is clearly $2\pi$ periodic in each of its components,
and we can thus consider it as a function $F\left(\vec{\kappa}\right)$,
where $\vec{\kappa}$ is a point on the torus $\mathbb{T}^{E}:=\mathbb{R}^{E}/2\pi\mathbb{Z}^{E}$.
This function is known as the ``\emph{secular function}''.

Given a vector of edge lengths $\vec{\ell}$, we can consider the
following linear flow on the torus:
\begin{equation}
\phi\left(k\right)=k\vec{\ell},k\in\mathbb{R},\label{eq:-102}
\end{equation}
where $k\vec{\ell}=\left(k\ell_{1},...,k\ell_{E}\right)$ is identified
with its class in $\mathbb{T}^{E}$, see also Figure \ref{fig:torus flow}.
Then by Theorem \ref{thm:secular-equation} presented above, we conclude
that $k^{2}>0$ is an eigenvalue of $H_{0}$ if and only if $F\left(\phi\left(k\right)\right)=0$.
We can thus define the set (see Figure \ref{fig:torus flow})
\begin{equation}
\Sigma=\left\{ \vec{\kappa}\in\mathbb{T}^{E}:F\left(\vec{\kappa}\right)=0\right\} .\label{eq:-38}
\end{equation}

\begin{figure}
\includegraphics[scale=0.9]{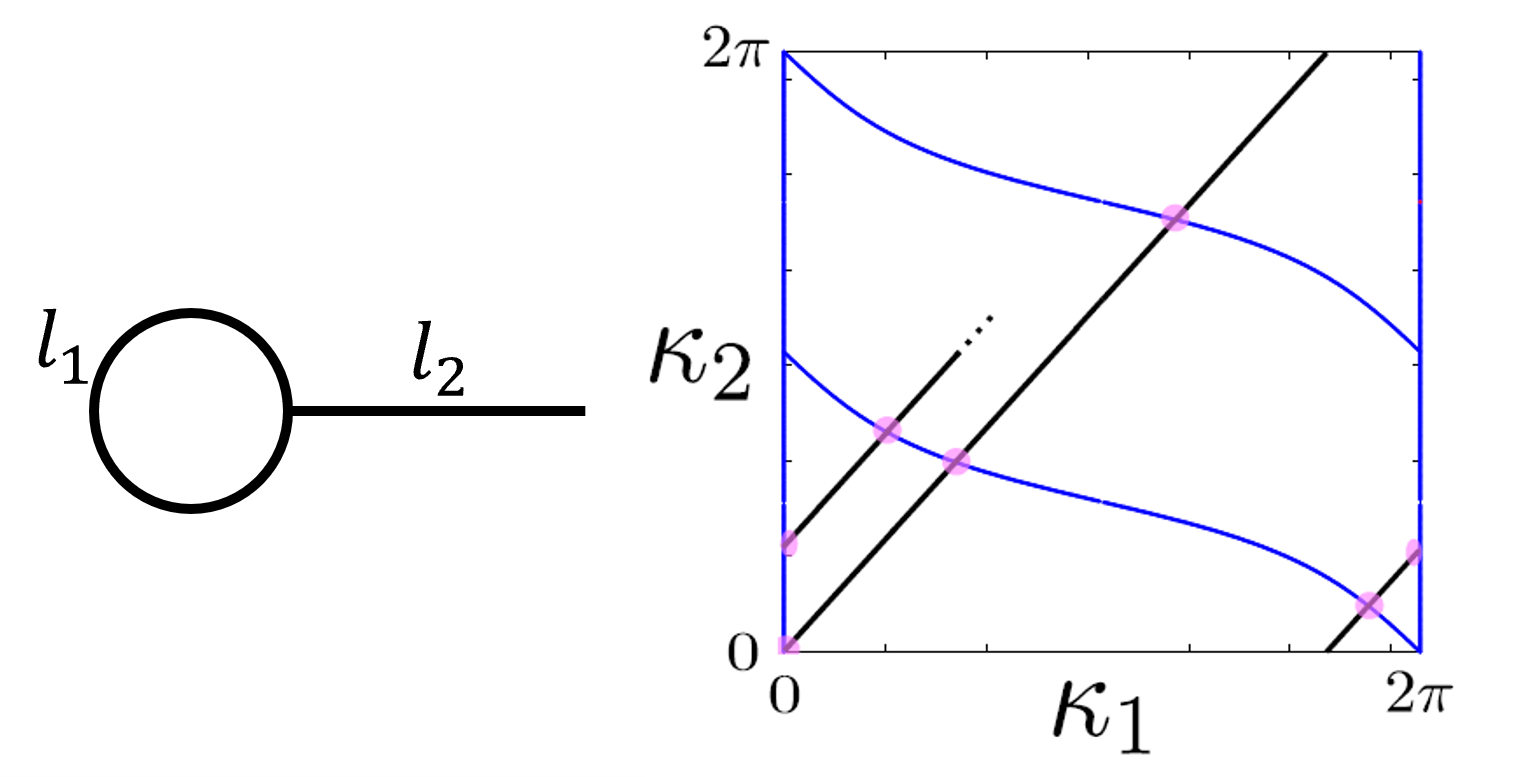}

\caption[A lasso graph and the corresponding secular manifold.]{A lasso graph and the corresponding secular manifold (blue) in the
two-dimensional torus. The black line is the torus flow $\phi\left(k\right)=k\vec{\ell}$.
Each intersection of the torus flow with the secular manifold corresponds
to an eigenvalue of $H_{0}$.\label{fig:torus flow}}
\end{figure}

$\Sigma$ is a compact algebraic subvariety of the torus known as
the \emph{secular manifold}\footnote{This is a slight misnomer, since generally the secular manifold could
have singular points.} (see \cite{AloBanBer_cmp18}). We can also define
\begin{equation}
\sreg=\left\{ \vec{\kappa}\in\Sigma:\dim\ker\left(I-Se^{i\kappa}\right)=1\right\} ,\label{eq:-39}
\end{equation}
with $\kappa=diag\left(\kappa_{1},\kappa_{1},...,\kappa_{E},\kappa_{E}\right)$.
$\sreg$ is a smooth submanifold of the torus of codimension one.

Note that if the combinatorial structure of the graph is fixed, then
any given point $\vec{\kappa}\in\sreg$ may be identified with infinitely
many different metric graphs. These are all metric graphs whose corresponding
secular function (which depends on their vector of edge lengths $\vec{\ell}$)
satisfies $\dim\ker\left(I-Se^{i\kappa}\right)=1$, where $\kappa=diag\left(k\ell_{1},k\ell_{1},...,k\ell_{E},k\ell_{E}\right)$.
Any point $\vec{\kappa}\in\sreg$ can be used to define a function
on this family of graphs by
\begin{equation}
\widetilde{f}_{j}^{\vec{\kappa}}\left(x\right)=a_{j}\left(\vec{\kappa}\right)e^{ix}+a_{\hat{j}}\left(\vec{\kappa}\right)e^{i\kappa_{j}}e^{-ix},\label{eq:-4}
\end{equation}

where $a_{j}\left(\vec{\kappa}\right)$ and $a_{\hat{j}}\left(\vec{\kappa}\right)$
are the corresponding entries of the $\mathbb{C}^{2E}$ normalized
vector $\vec{a}\in\ker\left(I-Se^{ik\L}\right)$ as in (\ref{eq:-2}).
These are also known as \emph{canonical eigenfunctions}, as discussed
in \cite{AloBanBer_cmp18,Alon_PhDThesis}. 

Note that if $k^{2}>0$ is an eigenvalue of $H_{0}$ with eigenfunction
$f$, then up to sign, $f$ and the canonical eigenfunction $f^{k\vec{\ell}}$
share the same amplitude vector $\vec{a}$, which means that they
attain the same values at vertices and have the same derivatives.
Due to this reason, the canonical eigenfunctions can be thought of
as a canonical choice of representatives for all eigenfunctions of
the family of metric graphs which correspond to the given point $\vec{\kappa}\in\sreg$.

The canonical eigenfunctions $\widetilde{f}^{\vec{\kappa}}$ can be
used to define various functions on $\sreg$, such as $\left|a_{j}\left(\vec{\kappa}\right)\right|^{2},\left|a_{\hat{j}}\left(\vec{\kappa}\right)\right|^{2},\left|\widetilde{f}_{j}^{\vec{\kappa}}\left(0\right)\right|^{2}$,
etc. By the above, we see that these functions on $\sreg$ can be
related to the eigenfunctions of our graph. Due to this useful property,
we would naturally like to be able to integrate such functions over
$\Sigma$. To do this, we shall apply the method first suggested by
Barra and Gaspard in \cite{BarGas_jsp00}, and further developed in
\cite{BerWin_tams10} by Berkolaiko and Winn and in \cite{CdV_ahp15}
by Colin de Verdière.
\begin{defn}
\label{def:bg-measure}(\cite{BarGas_jsp00,CdV_ahp15,BerWin_tams10}).
The Barra-Gaspard measure on the secular manifold $\Sigma$ is the
Radon probability measure given by
\begin{equation}
d\mu_{\vec{\ell}}=\frac{\pi}{L}\cdot\frac{1}{\left(2\pi\right)^{E}}\left|\hat{n}\left(\vec{\kappa}\right)\cdot\vec{\ell}\right|d\sigma,\label{eq:-5}
\end{equation}
where $d\sigma$ is the standard Lebesgue surface element and $\hat{n}$
is the unit normal to the secular manifold $\Sigma$. For a thorough
introduction to the Barra-Gaspard measure, see \cite{AloBanBer_cmp18,CdV_ahp15,AloBanBer_21arxiv,Alon_PhDThesis,BarGas_jsp00,BerWin_tams10}.
\end{defn}

The following ergodic theorem will be a main tool in proving our results:
\begin{thm}
\label{thm:ergodic-theorem}(\cite{BarGas_jsp00,CdV_ahp15,BerWin_tams10}).
Assume that the entries of the edge length vector $\vec{\ell}$ are
linearly independent over $\mathbb{Q}$. Let $g$ be a Riemann integrable
function on $\Sigma$ (equivalently -- $g$ is continuous almost
everywhere and bounded). Then
\begin{equation}
\left\langle g\right\rangle _{n}:=\lim_{N\rightarrow\infty}\frac{1}{N}\sum_{n=1}^{\infty}g\left(k_{n}\vec{\ell}\right)=\int_{\Sigma}gd\mu_{\vec{\ell}}.\label{eq:-7}
\end{equation}
\end{thm}

Thus, the ergodic theorem allows us to compute Cesaro means (see Definition
\ref{def:Cesaro}) of Riemann integrable functions on the secular
manifold when we sample them along the points where the torus flow
intersects the secular manifold.

In addition to Theorem \ref{thm:ergodic-theorem}, we will need the
following result: 
\begin{lem}
\label{lem:secular-normal}(\cite{CdV_ahp15}). Let $\widetilde{f}^{\vec{\kappa}}$
be a canonical eigenfunction such that $\vec{a}$ is $\mathbb{C}^{2E}$
normalized. Then the component of the unit normal to the secular manifold
$\Sigma$ at the point $\vec{\kappa}=k\vec{\ell}$ in the direction
of the edge $j$ is given by
\begin{equation}
\hat{n}_{j}=\left|a_{j}\right|^{2}+\left|a_{\hat{j}}\right|^{2}.\label{eq:-8}
\end{equation}
\end{lem}

Using these results, we can integrate some of the functions defined
above on the secular manifold.

\newpage{}

\section{Proofs for Theorems \ref{thm:1.RNG-Lipschitz}, \ref{thm:RNG-mean}
and \ref{thm:Weyl-law}\label{sec:proof-of-rng}}

\subsection{Proof of Theorem \ref{thm:1.RNG-Lipschitz}}

To prove Theorem \ref{thm:1.RNG-Lipschitz}, we need the following
lemmas:
\begin{lem}
\label{lem:hadamard}The RNG is given by the formula
\begin{equation}
d_{n}\left(\sigma\right)=\sum_{v\in\VR}\int_{0}^{\sigma}\left|f_{n}^{\left(t\right)}\left(v\right)\right|^{2}dt,\label{eq:-1}
\end{equation}
where $f_{n}^{\left(t\right)}$ is an $L^{2}$ normalized eigenfunction
of $H\left(t\right)$ as described in Section \ref{sec:RNG-tools}.
\end{lem}

\begin{lem}
\label{lem:bdd} For every $v\in\VR$, the quantity $\left|f_{n}^{\left(t\right)}\left(v\right)\right|^{2}$
is uniformly bounded in $n$ and $t\in\mathbb{R}$.
\end{lem}

\begin{rem}
\label{rem:additivity}Note that by Lemma \ref{lem:hadamard}, we
may from now on assume without loss of generality that our graph contains
a single $\delta$ interaction point placed at a vertex $v$. The
more general results of Theorem \ref{thm:1.RNG-Lipschitz} follow
from the additivity of Formula (\ref{eq:-1}). Moreover, for convenience,
we shall assume that the graph $\Gamma$ is parameterized so that
$v$ is located at $x=0$.

Before proving the lemmas, we use them in order to prove Theorem \ref{thm:1.RNG-Lipschitz}.
\end{rem}

\begin{proof}[Proof of Theorem \ref{thm:1.RNG-Lipschitz}]
 By Lemma \ref{lem:hadamard},
\begin{equation}
d_{n}\left(\sigma\right)=\sum_{v\in\VR}\int_{0}^{\sigma}\left|f_{n}^{\left(t\right)}\left(v\right)\right|^{2}dt.\label{eq:-41}
\end{equation}
By Remark \ref{rem:additivity}, we may prove the theorem for the
simpler case of $\VR$ consisting of a single vertex:
\begin{equation}
d_{n}\left(\sigma\right)=\int_{0}^{\sigma}\left|f_{n}^{\left(t\right)}\left(v\right)\right|^{2}dt.\label{eq:-42}
\end{equation}

By Lemma \ref{lem:bdd}, we know that $d_{n}'\left(\sigma\right)=\left|f_{n}^{\left(t\right)}\left(v\right)\right|^{2}$
is uniformly bounded in $n$ and $t\in\mathbb{R}$, which means that
the all functions in the sequence $\left(d_{n}\left(\sigma\right)\right)_{n=1}^{\infty}$
are Lipschitz continuous in $\mathbb{R}$ with a uniform Lipschitz
constant. In particular, this implies that the sequence of functions
$\left(d_{n}\left(\sigma\right)\right)_{n=1}^{\infty}$ is uniformly
bounded on any compact subset $A\subset\mathbb{R}$.

To prove the existence of a uniformly convergent subsequence on $A$,
we apply the Arzelà-Ascoli Theorem. We already know that that $\left(d_{n}\left(\sigma\right)\right)_{n=1}^{\infty}$
is uniformly bounded on $A$. Moreover, it is uniformly Lipschitz
and hence equicontinuous. By Arzelà-Ascoli, there exists a subsequence
$d_{n_{k}}\left(\sigma\right)$ which converges uniformly on $A$. 
\end{proof}
\begin{proof}[Proof of Lemma \ref{lem:hadamard}]
 We use a simple generalization of the Hadamard type formula which
is presented in Proposition $3.1.6$ in \cite{BerKuc_graphs} (the
formula allows for only a single $\delta$ vertex, while we allow
multiple vertices). The formula states that at any point $t$ such
that $\lambda_{n}\left(t\right)$ is a simple eigenvalue,
\begin{equation}
\frac{d\lambda_{n}\left(t\right)}{dt}=\sum_{v\in\VR}\left|f_{n}^{\left(t\right)}\left(v\right)\right|^{2}.\label{eq:-43}
\end{equation}

Note that unless $\lambda_{n}\left(t\right)$ is a multiple eigenvalue
for all $t\in\mathbb{R}$, then the set $D\subset\left[0,\sigma\right]$
of $t$ values for which $\lambda_{n}\left(t\right)$ is degenerate
must be finite. This is true since the spectral curves of the $\delta$
family are piecewise real analytic in $t$ (see Section $3.1.2$ in
\cite{BerKuc_graphs}). Indeed, if $D\subset\left[0,\sigma\right]$
was infinite, then two of the spectral curves would agree on a set
with an accumulation point, and thus agree everywhere. Furthermore,
if $\lambda_{n}\left(t\right)$ is a multiple eigenvalue for all $t$,
then by Lemma $3.1.15$ in \cite{BerKuc_graphs}, we have that $\lambda_{n}\left(t\right)=\lambda_{n}\left(0\right)$
for all $t$, and the corresponding eigenfunction vanishes at $\VR$
for all $t$. In this case, the RNG is zero, and Formula (\ref{eq:-1})
holds trivially, and so we may focus on the previous case.

Since Formula (\ref{eq:-43}) above holds at all but finitely many
values of $t$, which for the purpose of integration do not matter,
we conclude that
\begin{equation}
d_{n}\left(\sigma\right)=\lambda_{n}\left(\sigma\right)-\lambda_{n}\left(0\right)=\int_{0}^{\sigma}\frac{d\lambda_{n}\left(t\right)}{dt}dt=\sum_{v\in\VR}\int_{0}^{\sigma}\left|f_{n}^{\left(t\right)}\left(v\right)\right|^{2}dt.\label{eq:-44}
\end{equation}
\end{proof}
\begin{rem}
\label{rem:Simplicity-assumption} Due to arguments similar to the
ones in the proof above, we shall always assume from now on that Formula
(\ref{eq:-43}) holds everywhere (and for our purposes this will not
affect the proofs).
\end{rem}

\begin{proof}[Proof of Lemma \ref{lem:bdd}]
 Due to Remark \ref{rem:additivity}, we may assume that $v=0$.

Denote by $\widetilde{f}_{n}^{\left(t\right)}$ the eigenfunctions
of $H\left(t\right)$ whose coefficient vector $\vec{a}$ is $\mathbb{C}^{2E}$
normalized (see Section \ref{sec:RNG-tools}). We wish to express
the value of the $L^{2}$ normalized eigenfunction $\left|f_{n}^{\left(t\right)}\left(0\right)\right|^{2}$
in terms of the eigenfunction $\widetilde{f}_{n}^{\left(t\right)}$.

We treat the case where the corresponding eigenvalue is positive.
Since $H\left(t\right)$ can have only finitely many non-positive
eigenvalues, the boundedness result still holds by the same arguments
as below.

By Remark \ref{rem:Simplicity-assumption}, we may focus on the values
$t\in\mathbb{R}$ such that the eigenvalues of $H\left(t\right)$
are simple (and so the corresponding eigenspace is one-dimensional).
We can thus write $\widetilde{f}_{n}^{\left(t\right)}=\alpha_{n}f_{n}^{\left(t\right)}$
for some $\alpha_{n}\in\mathbb{R}$ (and for our purposes we may assume
$\alpha_{n}>0$).

Since $\left\Vert f_{n}^{\left(t\right)}\right\Vert _{L^{2}}=1$,
then $\left\Vert \widetilde{f}_{n}^{\left(t\right)}\right\Vert _{L^{2}}^{2}=\left|\alpha_{n}\right|^{2}$,
and we can compute $\alpha_{n}$ explicitly:
\begin{align}
1 & =\left\Vert f_{n}^{\left(t\right)}\right\Vert _{L^{2}}^{2}=\frac{1}{\left|\alpha_{n}\right|^{2}}\sum_{i=1}^{E}\int_{0}^{\ell_{i}}\left|a_{i}e^{ik_{n}x}+a_{\hat{i}}e^{ik_{n}\ell_{i}}e^{-ik_{n}x}\right|^{2}dx\label{eq:-45}\\
= & \frac{1}{\left|\alpha_{n}\right|^{2}}\sum_{i=1}^{E}\left[\ell_{i}\left(\left|a_{i}\right|^{2}+\left|a_{\hat{i}}\right|^{2}\right)+\frac{2}{k_{n}}Re\left(a_{i}\overline{a_{\hat{i}}}e^{-ik_{n}\ell_{i}}\right)\right],\label{eq:-47}
\end{align}

where we have used the expression for $\widetilde{f}_{n}^{\left(t\right)}$
given in (\ref{eq:-6}).

This means:
\begin{align}
 & \alpha_{n}=\frac{1}{\sqrt{\sum_{i=1}^{E}\left[\ell_{i}\left(\left|a_{i}\right|^{2}+\left|a_{\hat{i}}\right|^{2}\right)+\frac{2}{k_{n}}Re\left(a_{i}a_{\hat{i}}e^{ik_{n}\ell_{i}}\right)\right]}}\label{eq:-48}\\
\Rightarrow & \left|f_{n}^{\left(t\right)}\left(0\right)\right|^{2}=\left|\frac{\widetilde{f}_{n}^{\left(t\right)}\left(0\right)}{\alpha_{n}}\right|^{2}=\frac{\left|a_{j}+a_{\hat{j}}e^{ik_{n}\ell_{j}}\right|^{2}}{\sum_{i=1}^{E}\left[\ell_{i}\left(\left|a_{i}\right|^{2}+\left|a_{\hat{i}}\right|^{2}\right)+\frac{2}{k_{n}}Re\left(a_{i}\overline{a_{\hat{i}}}e^{-ik_{n}\ell_{i}}\right)\right]}.\label{eq:-9}
\end{align}

We want to show that the quantity above is bounded. From now on we
simply denote $\left|f_{n}^{\left(t\right)}\right|^{2}$.

Since $a_{i}$ and $a_{\hat{i}}$ are both bounded in absolute value
by one, the numerator of the expression above is bounded from above
by four. The first term in the denominator is bounded from below by
$\min_{i}\left\{ \ell_{i}\right\} $. Moreover, since $a_{i},a_{\hat{i}}$
are bounded in absolute value by one, there is some $k_{*}$ large
enough so that $\left|\frac{2}{k_{n}}Re\left(a_{i}\overline{a_{\hat{i}}}e^{-ik_{n}\ell_{i}}\right)\right|<\frac{1}{2}\min_{i}\left\{ \ell_{i}\right\} $
for all $k_{n}\geq k_{*}$ (and in particular, it is smaller than
the first term in the denominator).

This means that $\left|f_{n}^{\left(t\right)}\right|^{2}$ is bounded
in $t$ and $n$ for all but possibly finitely many values of $n$
(the ones for which $k<k_{*}$). Since there are only finitely many
wave numbers $k_{n}$ in $\left[0,k_{*}\right]$, it is enough to
show that for fixed $n$ (which corresponds to a fixed $k_{n}$),
$\left|f_{n}^{\left(t\right)}\right|^{2}$ is bounded in $t$.

Similar to the proof of Lemma \ref{lem:hadamard}, $\left|f_{n}^{\left(t\right)}\right|^{2}$
is a real analytic function of $t$. Moreover, note as $t\rightarrow\pm\infty$,
the $\delta$ condition corresponds to the Dirichlet condition at
$v$ (see also Section $1.4.4$ in \cite{BerKuc_graphs}), and so
$\lim_{t\rightarrow\pm\infty}\left|f_{n}^{\left(t\right)}\right|^{2}\rightarrow0$.
So $\left|f_{n}^{\left(t\right)}\right|^{2}$ is bounded as a real
analytic function which tends to zero at infinity.

Overall, we get that $\left|f_{n}^{\left(t\right)}\right|^{2}$ is
uniformly bounded in $n$ and $t\in\mathbb{R}$, as required.
\end{proof}
\bigskip

\subsection{Proof of Theorem \ref{thm:Weyl-law}}

From now and until the end of Subsection \ref{subsec:mean-proof},
we make the following assumption\textbf{:}
\begin{assumption}
\label{subsec:rationality}The edge lengths of the graph are linearly
independent over $\mathbb{Q}$.
\end{assumption}

This assumption is required in order to apply Theorem \ref{thm:ergodic-theorem}.
In Subsection \ref{subsec:rational-proof} we will show that this
assumption can be omitted.

Our proofs will heavily rely on the secular manifold. By Theorem \ref{thm:ergodic-theorem},
we know that for any Riemann integrable function $g:\Sigma\rightarrow\mathbb{R}$,
we have that
\begin{equation}
\lim_{N\rightarrow\infty}\frac{1}{N}\sum_{n=1}^{N}g\left(\vec{\kappa}_{n}\right)=\int_{\Sigma}gd\mu_{\vec{\ell}}.\label{eq:-58}
\end{equation}

Our strategy will thus be to compute the mean values by defining appropriate
functions on the secular manifold which we can integrate.
\begin{rem}
\label{rem:meas0} While many of the functions we consider for the
application of Theorem \ref{thm:ergodic-theorem} are only well defined
on $\sreg$, we know that $\Sigma\backslash\sreg$ is of codimension
at least one in $\Sigma$, and is thus of measure zero (since $\mu_{\vec{\ell}}$
is a Radon measure). Thus, for the purpose of integration, this does
not matter.
\end{rem}

\begin{rem}
We prove the three formulas in Theorem \ref{thm:Weyl-law} in the
terminology of the eigenfunctions whose coefficient vector $\vec{a}$
is $\mathbb{C}^{2E}$ normalized. This is equivalent to the statement
given in Theorem \ref{thm:Weyl-law}, after applying Formula (\ref{eq:-9})
obtained in the proof of Lemma \ref{lem:bdd}.
\end{rem}

\begin{lem}
\label{lem:Amean}On each edge $j$, the following holds:
\begin{equation}
\left\langle \frac{\left|a_{j}\right|^{2}}{\sum_{i=1}^{E}\ell_{i}\left(\left|a_{i}\right|^{2}+\left|a_{\hat{i}}\right|^{2}\right)}\right\rangle _{n}=\left\langle \frac{\left|a_{\hat{j}}\right|^{2}}{\sum_{i=1}^{E}\ell_{i}\left(\left|a_{i}\right|^{2}+\left|a_{\hat{i}}\right|^{2}\right)}\right\rangle _{n}=\frac{1}{2L}.\label{eq:-42-1}
\end{equation}
\end{lem}

\begin{proof}
Consider the transformation $\vec{\kappa}\mapsto-\vec{\kappa}$ on
the torus. This is a measure preserving isometry of the secular manifold
(see \cite{AloBanBer_cmp18}), and so for any Riemann integrable $g:\Sigma\rightarrow\mathbb{C}$,
\begin{equation}
\int_{\Sigma}g\left(\vec{\kappa}\right)d\mu_{\vec{\ell}}=\int_{\Sigma}g\left(-\vec{\kappa}\right)d\mu_{\vec{\ell}}.\label{eq:-10-2}
\end{equation}

Note that under our transformation, the canonical eigenfunctions change
in the following manner:
\begin{align}
 & \tilde{f}_{j}\left(x\right)=a_{j}e^{ikx}+a_{\hat{j}}e^{ik\ell_{j}}e^{-ikx},\label{eq:-59-1-1}\\
 & k\mapsto-k,\label{eq:-60-1-1}\\
 & \tilde{f}_{j}\left(x\right)\mapsto a_{j}e^{-ikx}+a_{\hat{j}}e^{-ik\ell_{j}}e^{ikx}.\label{eq:-29-2}
\end{align}

So up to a phase (which does not affect the absolute value), $a_{j}$
and $a_{\hat{j}}$ simply switch roles:
\begin{align}
 & a_{j}\mapsto a_{\hat{j}}e^{-ik\ell_{j}},\label{eq:-31-2}\\
 & a_{\hat{j}}\mapsto a_{j}e^{-ik\ell_{j}}.\label{eq:-32-2-1}
\end{align}

Now, consider the following function on $\sreg$:
\begin{equation}
g\left(\vec{\kappa}\right)=\frac{\left|a_{j}+a_{\hat{j}}e^{i\kappa_{j}}\right|^{2}}{\sum_{i=1}^{E}\ell_{i}\left(\left|a_{i}\right|^{2}+\left|a_{\hat{i}}\right|^{2}\right)}.\label{eq:-59}
\end{equation}
 As stated in Remark \ref{rem:meas0}, for the purpose of integration,
we can consider $g$ as a function on $\Sigma$ itself. The coefficients
$a_{i},a_{\hat{i}}$ in $g$ are defined by the coefficients of the
canonical eigenfunctions, as described in (\ref{eq:-4}). Recall that
the canonical eigenfunctions have the same amplitude vector as the
eigenfunctions we consider.

Moreover, under our measure preserving transformation, this function
becomes:
\begin{equation}
g\left(-\vec{\kappa}\right)=\frac{\left|a_{j}\right|^{2}}{\sum_{i=1}^{E}\ell_{i}\left(\left|a_{i}\right|^{2}+\left|a_{\hat{i}}\right|^{2}\right)}.\label{eq:-33-3}
\end{equation}

Since the integral of $g\left(-\vec{\kappa}\right)$ should remain
the same under our measure preserving change of variables, we conclude
from the ergodic theorem (Theorem \ref{thm:ergodic-theorem}) that
\begin{equation}
\left\langle \frac{\left|a_{j}\right|^{2}}{\sum_{i=1}^{E}\ell_{i}\left(\left|a_{i}\right|^{2}+\left|a_{\hat{i}}\right|^{2}\right)}\right\rangle _{n}=\left\langle \frac{\left|a_{\hat{j}}\right|^{2}}{\sum_{i=1}^{E}\ell_{i}\left(\left|a_{i}\right|^{2}+\left|a_{\hat{i}}\right|^{2}\right)}\right\rangle _{n}.\label{eq:-34-2}
\end{equation}

Now, to prove that both terms above are equal to $\frac{1}{2L}$,
we may apply Lemma \ref{lem:secular-normal} and simply prove that
\begin{equation}
\left\langle \frac{\left|a_{j}\right|^{2}+\left|a_{\hat{j}}\right|^{2}}{\sum_{i=1}^{E}\ell_{i}\left(\left|a_{i}\right|^{2}+\left|a_{\hat{i}}\right|^{2}\right)}\right\rangle _{n}=_{(\ref{eq:-8})}\left\langle \frac{\hat{n}_{j}}{\sum_{i=1}^{E}\ell_{i}\hat{n}_{i}}\right\rangle _{n}=\frac{1}{L}.\label{eq:-35-2}
\end{equation}

Due to the ergodic theorem and the definition of the Barra-Gaspard
measure (\ref{eq:-5},\ref{eq:-7}),
\begin{equation}
\left\langle \frac{\hat{n}_{j}}{\sum_{i=1}^{E}\ell_{i}\hat{n}_{i}}\right\rangle _{n}=\int_{\Sigma}\frac{\hat{n}_{j}}{\sum_{i=1}^{E}\ell_{i}\hat{n}_{i}}d\mu_{\vec{\ell}}=\frac{\pi}{L}\cdot\frac{1}{\left(2\pi\right)^{E}}\int_{\Sigma}\hat{n}_{j}d\sigma,\label{eq:-36-2}
\end{equation}
and so \ref{eq:-35-2} is equivalent to showing that $\int_{\Sigma}\hat{n}_{j}d\sigma=2\left(2\pi\right)^{E-1}$.

The given integral is exactly the flux through the secular manifold
of the vector field $\frac{\partial}{\partial x_{j}}$. Proposition
$3.1$ in \cite{CdV_ahp15} shows that the map $\pi_{j}:\Sigma\rightarrow\left(\mathbb{R}/2\pi\mathbb{Z}\right)^{E-1}$
which projects the secular manifold to the $j$th face of the torus
is a two to one map. Thus, the given integral should be equal to twice
the flux of the given vector field through the $j$th face of the
torus:
\begin{equation}
\int_{\Sigma}\left(1,0,...,0\right)\cdot d\hat{n}=2\int_{\left(\mathbb{R}/2\pi\mathbb{Z}\right)^{E-1}}1dn=2\left(2\pi\right)^{E-1},\label{eq:-37-2}
\end{equation}
as required.
\end{proof}
\begin{lem}
\label{lem:Uncorrelation} For every $i\neq j$, the following holds:
\begin{equation}
\left\langle \frac{a_{i}\overline{a_{j}}}{\sum_{i=1}^{E}\ell_{i}\left(\left|a_{i}\right|^{2}+\left|a_{\hat{i}}\right|^{2}\right)}\right\rangle _{n}=\left\langle \frac{a_{\hat{i}}\overline{a_{\hat{j}}}}{\sum_{i=1}^{E}\ell_{i}\left(\left|a_{i}\right|^{2}+\left|a_{\hat{i}}\right|^{2}\right)}\right\rangle _{n}=0.\label{eq:-43-2}
\end{equation}
\end{lem}

\begin{proof}
We write the proof for $\hat{i},\hat{j}$, the proof for $i,j$ is
identical.

We refer to Theorem $4.10$ in \cite{AloBanBer_21arxiv}, which gives
an alternative method for integrating Riemann integrable functions
over $\Sigma$:
\begin{equation}
\int_{\Sigma}gd\mu_{\vec{\ell}}=\int_{\mathbb{T}^{E}}\sum_{n=1}^{2E}g\left(\vec{\kappa}-\theta_{n}\cdot\vec{1}\right)\frac{\vec{a}_{n}^{*}\L\vec{a}_{n}}{\text{tr}\left(\L\right)}\frac{d\vec{\kappa}}{\left(2\pi\right)^{E}},\label{eq:-11-2}
\end{equation}
where $\theta_{n}$ are the eigenphases of the unitary matrix $Se^{i\kappa}$
in (\ref{eq:-2}), $\vec{a}_{n}$ are its $\mathbb{C}^{2E}$ normalized
eigenvectors, and $\vec{1}$ is the vector $\left(1,1,...,1\right)$.
In a sense, the result above states that instead of sampling $g$
directly along the torus flow defined in Subsection \ref{subsec:secular-manifold},
one can sample $g\left(\vec{\kappa}-\theta_{n}\cdot\vec{1}\right)$
along the entire torus.

Applying the ergodic theorem along with (\ref{eq:-11-2}) to the function
$g\left(\vec{\kappa}\right)=\frac{a_{\hat{j}}\overline{a_{\hat{i}}}}{\sum_{i=1}^{E}\ell_{i}\left(\left|a_{i}\right|^{2}+\left|a_{\hat{i}}\right|^{2}\right)}$,
we conclude:
\begin{equation}
\left\langle \frac{a_{\hat{j}}\overline{a_{\hat{i}}}}{\sum_{i=1}^{E}\ell_{i}\left(\left|a_{i}\right|^{2}+\left|a_{\hat{i}}\right|^{2}\right)}\right\rangle _{n}=\frac{1}{2L\left(2\pi\right)^{E}}\int_{T}\sum_{n=1}^{2E}a_{\hat{j}}\overline{a_{\hat{i}}}\left(\vec{\kappa}-\theta_{n}\cdot\vec{1}\right)d\vec{\kappa}.\label{eq:-38-2-1}
\end{equation}
Since the eigenvectors $\vec{a}_{n}$ of the unitary matrix $Se^{i\boldsymbol{\kappa}}$
form an orthonormal basis, the integrand itself (which is just the
$\mathbb{C}^{2E}$ inner product) is identically zero, and so the
integral is zero, as required.
\end{proof}
\begin{rem}
\label{rem:canonical-equiv} Equation (\ref{eq:-9}) in the proof
of Lemma \ref{lem:bdd} shows that
\begin{align}
 & \left|A_{j}\right|^{2}=\frac{\left|a_{j}\right|^{2}}{\sum_{i=1}^{E}\left[\ell_{i}\left(\left|a_{i}\right|^{2}+\left|a_{\hat{i}}\right|^{2}\right)+O\left(\frac{1}{k_{n}}\right)\right]},\label{eq:-101}\\
 & \left|A_{\hat{j}}\right|^{2}=\frac{\left|a_{\hat{j}}\right|^{2}}{\sum_{i=1}^{E}\left[\ell_{i}\left(\left|a_{i}\right|^{2}+\left|a_{\hat{i}}\right|^{2}\right)+O\left(\frac{1}{k_{n}}\right)\right]},\label{eq:-103}\\
 & A_{\hat{j}}\overline{A_{\hat{e}}}=\frac{a_{\hat{j}}\overline{a_{j}}}{\sum_{i=1}^{E}\ell_{i}\left(\left|a_{i}\right|^{2}+\left|a_{\hat{i}}\right|^{2}+O\left(\frac{1}{k_{n}}\right)\right)}.\label{eq:-104}
\end{align}

Since terms of the form $O\left(\frac{1}{k_{n}}\right)$ do not affect
the Cesaro sum (For more details, see proof of Lemma \ref{lem:fmean}
below), we see that the two lemmas above are equivalent to Equations
(\ref{eq:-16-2}) and (\ref{eq:-17-2}) in the statement of Theorem
\ref{thm:Weyl-law}.
\end{rem}

Combining the two lemmas we get:
\begin{lem}
\label{lem:fmean}At each vertex $v$ we have that
\begin{equation}
\left\langle \left|f\left(v\right)\right|^{2}\right\rangle _{n}=\frac{2}{\deg\left(v\right)L}.\label{eq:-44-2}
\end{equation}
\end{lem}

\begin{proof}
From now on we assume as usual that the point of evaluation is at
$v=0$ and replace the notation $f\left(v\right)$ with simply $f$.

From Equation (\ref{eq:-9}) in the proof of Lemma \ref{lem:bdd}
we know that
\begin{equation}
\left|f\right|^{2}=\frac{\left|a_{j}+a_{\hat{j}}e^{ik_{n}\ell_{j}}\right|^{2}}{\sum_{i=1}^{E}\left[\ell_{i}\left(\left|a_{i}\right|^{2}+\left|a_{\hat{i}}\right|^{2}\right)+O\left(\frac{1}{k_{n}}\right)\right]}.\label{eq:-32-2}
\end{equation}

Here we remind the reader that the coefficients of the vector $\vec{a}$
depend on $k_{n}$ as well. Since $\lim_{n\rightarrow\infty}k_{n}=\infty$,
note that
\begin{equation}
\lim_{n\rightarrow\infty}\left|\frac{\left|a_{j}+a_{\hat{j}}e^{ik_{n}\ell_{j}}\right|^{2}}{\sum_{i=1}^{E}\ell_{i}\left(\left|a_{i}\right|^{2}+\left|a_{\hat{i}}\right|^{2}\right)+O\left(\frac{1}{k_{n}}\right)}-\frac{\left|a_{j}+a_{\hat{j}}e^{ik_{n}\ell_{j}}\right|^{2}}{\sum_{i=1}^{E}\ell_{i}\left(\left|a_{i}\right|^{2}+\left|a_{\hat{i}}\right|^{2}\right)}\right|=0,\label{eq:-57}
\end{equation}
which means that the Cesaro sum will not change if we omit the additional
term in the denominator. Thus,
\begin{align}
 & \left\langle \left|f\right|^{2}\right\rangle _{n}=\left\langle \frac{\left|a_{j}+a_{\hat{j}}e^{ik_{n}\ell_{j}}\right|^{2}}{\sum_{i=1}^{E}\ell_{i}\left(\left|a_{i}\right|^{2}+\left|a_{\hat{i}}\right|^{2}\right)}\right\rangle _{n}\label{eq:-61-1-1}\\
= & \left\langle \frac{\left|a_{j}\right|^{2}+\left|a_{\hat{j}}\right|^{2}}{\sum_{i=1}^{E}\ell_{i}\left(\left|a_{i}\right|^{2}+\left|a_{\hat{i}}\right|^{2}\right)}\right\rangle _{n}+2\left\langle \frac{Re\left(e^{ik_{n}\ell_{j}}a_{\hat{j}}\overline{a_{j}}\right)}{\sum_{i=1}^{E}\ell_{i}\left(\left|a_{i}\right|^{2}+\left|a_{\hat{i}}\right|^{2}\right)}\right\rangle _{n}\label{eq:-14-2}
\end{align}

Due to Lemma \ref{lem:Amean}, the first term is $\frac{1}{L}$. We
now want to evaluate the second term. From Equation (\ref{eq:-32-2})
we have that
\begin{equation}
a_{j}=e^{ik_{n}\ell_{j}}\sum_{i=1}^{\deg\left(v\right)}\left(\frac{2}{\deg\left(v\right)}-\delta_{ji}\right)a_{\hat{i}}.\label{eq:-39-2}
\end{equation}

Thus,
\begin{align}
 & \left\langle \frac{Re\left(e^{ik_{n}\ell_{j}}a_{\hat{j}}\overline{a_{j}}\right)}{\sum_{i=1}^{E}\ell_{i}\left(\left|a_{i}\right|^{2}+\left|a_{\hat{i}}\right|^{2}\right)}\right\rangle _{n}=\left(\frac{2}{\deg\left(v\right)}-1\right)\left\langle \frac{\left|a_{\hat{j}}\right|^{2}}{\sum_{i=1}^{E}\ell_{i}\left(\left|a_{i}\right|^{2}+\left|a_{\hat{i}}\right|^{2}\right)}\right\rangle _{n}\label{eq:-62-1-1}\\
 & +\frac{2}{\deg\left(v\right)}\sum_{i\neq j\in E_{v}}Re\left\langle \frac{a_{\hat{j}}\overline{a_{\hat{i}}}}{\sum_{i=1}^{E}\ell_{i}\left(\left|a_{i}\right|^{2}+\left|a_{\hat{i}}\right|^{2}\right)}\right\rangle _{n}.\label{eq:-40-2-1}
\end{align}

From Lemmas \ref{lem:Amean} and \ref{lem:Uncorrelation}, the first
term is equal to $\frac{1}{2L}\left(\frac{2}{\deg\left(v\right)}-1\right)$,
while the second term is equal to zero. Plugging this into (\ref{eq:-14-2})
we obtain
\begin{equation}
\left\langle \left|f\right|^{2}\right\rangle _{n}=\frac{1}{L}+\frac{1}{L}\left(\frac{2}{\deg\left(v\right)}-1\right)=\frac{2}{\deg\left(v\right)L}.\label{eq:-41-1-1}
\end{equation}
\end{proof}
\bigskip

\subsection{\label{subsec:mean-proof}Proof of Theorem \ref{thm:RNG-mean}}
\begin{prop}
\label{prop:mean-conv} For any $\sigma\in\mathbb{R}$, $\left\langle d\right\rangle _{n}\left(\sigma\right)$
exists and satisfies
\begin{equation}
\left\langle d\right\rangle _{n}\left(\sigma\right)=\sum_{v\in\VR}\left\langle \left|f^{\left(0\right)}\left(v\right)\right|^{2}\right\rangle _{n}\sigma.\label{eq:-49}
\end{equation}
\end{prop}

Applying the local Weyl law from Theorem \ref{thm:Weyl-law}, the
proposition above immediately proves Theorem \ref{thm:RNG-mean} as
a corollary.
\begin{rem}
By Remark \ref{rem:additivity}, we may once again without loss of
generality prove Proposition \ref{prop:mean-conv} for a single $\delta$
vertex placed at $v=0$. As done before, we replace the notation $f_{n}^{\left(t\right)}\left(0\right)$
with $f_{n}^{\left(t\right)}$ for convenience.
\end{rem}

To prove Proposition \ref{prop:mean-conv}, we need the following
lemmas:
\begin{lem}
\label{lem:eigenfunc-conv}$\left|f_{n}^{\left(t\right)}-f_{n}^{\left(0\right)}\right|\underset{_{n\rightarrow\infty}}{\rightarrow}0$
uniformly in $t\in\left[0,\sigma\right]$.
\end{lem}

\begin{lem}
\label{lem:eigmean-conv}$\left\langle \left|f^{\left(t\right)}\right|^{2}\right\rangle _{n}=\left\langle \left|f^{\left(0\right)}\right|^{2}\right\rangle _{n}$
for all $t\in\mathbb{R}$. 
\end{lem}

Before proving the lemmas, we use them to prove Proposition \ref{prop:mean-conv}.
\begin{proof}[Proof of Proposition \ref{prop:mean-conv}]
 Fix $\sigma$. Then
\begin{align}
 & \left\langle d\right\rangle _{n}\left(\sigma\right)=\lim_{N\rightarrow\infty}\frac{1}{N}\sum_{n=1}^{N}d_{n}\left(\sigma\right)\label{eq:-50}\\
 & =_{(\ref{eq:-1})}\lim_{N\rightarrow\infty}\frac{1}{N}\sum_{n=1}^{N}\int_{0}^{\sigma}\left|f_{n}^{\left(t\right)}\right|^{2}dt=\lim_{N\rightarrow\infty}\int_{0}^{\sigma}\frac{1}{N}\sum_{n=1}^{N}\left|f_{n}^{\left(t\right)}\right|^{2}dt.
\end{align}

We know by Lemma \ref{lem:eigmean-conv} that $\left\langle \left|f^{\left(t\right)}\right|^{2}\right\rangle _{n}=\left\langle \left|f^{\left(0\right)}\right|^{2}\right\rangle _{n}$,
and that the convergence is uniform in $t\in\left[0,\sigma\right]$.
We can thus insert the limit into the integral and take $t=0$:
\begin{equation}
\left\langle d\right\rangle _{n}\left(\sigma\right)=\int_{0}^{\sigma}\lim_{N\rightarrow\infty}\frac{1}{N}\sum_{n=1}^{N}\left|f^{\left(t\right)}\right|^{2}dt=\int_{0}^{\sigma}\left\langle \left|f^{\left(0\right)}\right|^{2}\right\rangle _{n}dt=\left\langle \left|f^{\left(0\right)}\right|^{2}\right\rangle _{n}\sigma.\label{eq:-51}
\end{equation}
\end{proof}
\begin{proof}[Proof of Lemma \ref{lem:eigenfunc-conv}]
 Recall that by Formula (\ref{eq:-32}), the scattering matrix of
the $\delta$ condition is given by
\begin{equation}
S_{j'j}^{\left(t\right)}=\begin{cases}
\frac{2}{\deg\left(v\right)+\frac{it}{k}}-1 & j'=\hat{j}\\
\frac{2}{\deg\left(v\right)+\frac{it}{k}} & j\rightarrow j'\text{ at \ensuremath{v} and }j'\neq\hat{j}\\
0 & \text{Otherwise}
\end{cases}\label{eq:-3-1-1-1-1-1}
\end{equation}
Moreover, denote $U^{\left(t\right)}\left(k\right)=S^{\left(t\right)}e^{ik\L}$.
From (\ref{eq:-2}), we know that $k^{2}$ is an eigenvalue of $H\left(t\right)$
if and only if one is an eigenvalue of $U^{\left(t\right)}\left(k\right)$.
We would like to show that the $k$ values for which this happens
with $U^{\left(t\right)}\left(k\right)$ get close to the $k$ values
for which this happens with $U^{\left(0\right)}\left(k\right)$ as
$k\rightarrow\infty$ .

Note that by definition of $U^{\left(t\right)}$,
\begin{equation}
\left\Vert U^{\left(t\right)}\left(k\right)-U^{\left(0\right)}\left(k\right)\right\Vert _{\infty}=\left\Vert e^{ik\L}\left(S^{\left(t\right)}-S^{\left(0\right)}\right)\right\Vert _{\infty}\leq\frac{2t}{\deg\left(v\right)\left|\deg\left(v\right)+\frac{it}{k}\right|k},\label{eq:-52}
\end{equation}
and this expression approaches zero uniformly in $t\in\left[0,\sigma\right]$
as $k\rightarrow\infty$ . Since the supremum norm of the difference
tends to zero, so does the operator norm of the difference. This means
that as $k\rightarrow\infty$, the eigenvalues of $U^{\left(t\right)}\left(k\right)$
converge to those of $U^{\left(0\right)}\left(k\right)$ uniformly
in $t\in\left[0,\sigma\right]$.

Denote the eigenvalues of the unitary matrix $U^{\left(t\right)}\left(k\right)$
by $\left(e^{i\theta_{j}^{t}\left(k\right)}\right)_{j=1}^{2E}$, so
that the eigenphases $\left(\theta_{j}^{t}\left(k\right)\right)_{j=1}^{2E}$
are the lifts of these eigenvalues from $S^{1}$ to the universal
cover $\mathbb{R}$. We have that $k^{2}>0$ is an eigenvalue of $H\left(t\right)$
if and only if $\theta_{j}^{t}\left(k\right)\in2\pi\mathbb{Z}$ for
some $j$. Denote by $\left(k_{n}^{t}\right)_{n=1}^{\infty}$ the
$k$ values for which this happens (these are exactly the roots of
the secular function (\ref{eq:-3})).

We know that $\left(\theta_{j}^{t}\left(k\right)\right)_{j=1}^{2E}$
increase monotonically with $k$ at a rate which is bounded from below
by some $c>0$, and that $k_{n}^{0}<k_{n}^{t}$ for all $n$ (see
Lemma $4.5$ in \cite{BolEnd_ahp09}). Then by applying the mean value
theorem we get
\begin{align}
 & \theta_{j}^{0}\left(k_{n}^{0}\right)=2\pi m=\theta_{j}^{t}\left(k_{n}^{t}\right)\geq\theta_{j}^{t}\left(k_{n}^{0}\right)+c\left(k_{n}^{t}-k_{n}^{0}\right)\label{eq:-53}\\
\Rightarrow & k_{n}^{t}-k_{n}^{0}\leq\frac{1}{c}\left(\theta_{j}^{0}\left(k_{n}^{0}\right)-\theta_{j}^{t}\left(k_{n}^{0}\right)\right).\label{eq:-54}
\end{align}

As $n\rightarrow\infty$ (which is equivalent to $k\rightarrow\infty$),
we know by Equation (\ref{eq:-52}) that the expression above goes
to zero, and so we conclude that as $n\rightarrow\infty$, $\left|k_{n}^{t}-k_{n}^{0}\right|\rightarrow0$
uniformly in $t\in\left[0,\sigma\right]$.

Since the roots of the secular function (which are exactly $k_{n}^{t}$)
determine the coefficients of the eigenfunction via the (real analytic)
matrix equation (\ref{eq:-2}), we get that as $n\rightarrow\infty$,
\begin{equation}
\left|f_{n}^{\left(t\right)}\left(0\right)-f_{n}^{\left(0\right)}\left(0\right)\right|\rightarrow0,\label{eq:-55}
\end{equation}
and that this convergence is uniform in $t\in\left[0,\sigma\right]$,
as required.
\end{proof}
\begin{proof}[Proof of Lemma \ref{lem:eigmean-conv}]
 We first note that for $t=0$, $\left\langle \left|f^{\left(0\right)}\right|^{2}\right\rangle _{n}$
exists by Lemma \ref{lem:fmean}. We denote this mean value by $C$
for brevity.

For $t\neq0$, we claim that $\left\langle \left|f^{\left(t\right)}\right|^{2}\right\rangle _{n}$
exists as well, and that it is actually equal to the same constant
$C$. To show this, we use the fact that $\left|f_{n}^{\left(t\right)}-f_{n}^{\left(0\right)}\right|\underset{n\rightarrow\infty}{\rightarrow}0$
(Lemma \ref{lem:eigenfunc-conv}). Note that
\begin{equation}
\frac{1}{N}\sum_{n=1}^{N}\left|f_{n}^{\left(t\right)}\right|^{2}\leq\frac{1}{N}\sum_{n=1}^{N}\left|f_{n}^{\left(0\right)}\right|^{2}+\frac{1}{N}\sum_{n=1}^{N}\left|\left(f_{n}^{\left(t\right)}\right)^{2}-\left(f_{n}^{\left(0\right)}\right)^{2}\right|.\label{eq:-60}
\end{equation}

As $N\rightarrow\infty$, the first term converges to $C$. We claim
that the second term converges to zero. To show this, it is enough
to show that the summands themselves tend to zero (since the Cesaro
sum of a converging sequence is the limit itself).

Recall by the proof of Theorem \ref{thm:1.RNG-Lipschitz} that the
expressions $\left|f_{n}^{\left(t\right)}\right|^{2}$ are all uniformly
bounded in $t\in\mathbb{R}$ by some $M>0$. Since $\left|f_{n}^{\left(t\right)}-f_{n}^{\left(0\right)}\right|\underset{n\rightarrow\infty}{\rightarrow}0$,
we have that
\begin{align}
 & \left|\left(f_{n}^{\left(t\right)}\right)^{2}-\left(f_{n}^{\left(0\right)}\right)^{2}\right|=\left|f_{n}^{\left(t\right)}-f_{n}^{\left(0\right)}\right|\cdot\left|f_{n}^{\left(t\right)}+f_{n}^{\left(0\right)}\right|\label{eq:-61}\\
 & \leq2M\left|f_{n}^{\left(t\right)}-f_{n}^{\left(0\right)}\right|\rightarrow0,\label{eq:-62}
\end{align}
and so the second term indeed goes to zero. Overall we get that $C$
is the Cesaro mean in this case as well.
\end{proof}
\bigskip

\subsection{Omitting the assumption of independence over $\mathbb{Q}$ \label{subsec:rational-proof}}

Recall that in order to apply the ergodic theorem, we added the assumption
that the entries of the vector of edge lengths $\vec{\ell}$ are linearly
independent over $\mathbb{Q}$. We now show that the result of Theorems
\ref{thm:RNG-mean} and \ref{thm:Weyl-law} in fact holds without
this assumption.
\begin{prop}
Assumption \ref{subsec:rationality} can be omitted in Theorems \ref{thm:RNG-mean}
and \ref{thm:Weyl-law}.
\end{prop}

\begin{proof}
Fix $\sigma\in\mathbb{R}$, a combinatorial graph $G=\left(\mathcal{V},\mathcal{E}\right)$
and a Robin set $\VR\subset\mathcal{V}$. Denote:
\begin{equation}
\mathbb{R}_{+}^{E}=\left\{ \vec{x}\in\mathbb{R}^{E}:x_{i}>0,\forall i\in\left\{ 1,..,E\right\} \right\} .\label{eq:-11}
\end{equation}

For $\vec{\ell}\in\mathbb{R}_{+}^{E}$, denote by $\Gamma_{\vec{\ell}}$
the metric graph obtained by assigning the vector of edge lengths
$\vec{\ell}$ to the fixed combinatorial graph $G$ (recall Definition
\ref{fig:A-metric-graph.}). Furthermore, denote by $P^{E}$ the subset
of $\mathbb{R}_{+}^{E}$ of vectors whose coordinates are rationally
independent. This is a dense subset of $\mathbb{R}_{+}^{E}$. Lastly,
denote the set of Cesaro summable sequences by $\mathcal{C}$.

Define the following function:
\begin{align}
 & \phi_{1}:\mathbb{R}_{+}^{E}\rightarrow\mathbb{R}^{\mathbb{N}},\label{eq:-10}\\
 & \phi_{1}\left(\vec{\ell}\right)=\left(d_{1}^{\vec{\ell}}\left(\sigma\right),d_{2}^{\vec{\ell}}\left(\sigma\right),...\right),\label{eq:-14}
\end{align}

where $d_{n}^{\vec{\ell}}\left(\sigma\right)$ is the RNG for the
graph $\Gamma_{\vec{\ell}}$ with corresponding Robin set $\mathcal{V}_{R}$.
Furthermore, define the following additional functions:
\begin{align}
 & \phi_{2}:\mathcal{C}\rightarrow\mathbb{R},\label{eq:-63}\\
 & \phi_{2}\left(\left(a_{n}\right)_{n=1}^{\infty}\right)=\lim_{N\rightarrow\infty}\frac{1}{N}\sum_{n=1}^{N}a_{n}\label{eq:-64}\\
 & \phi:P^{E}\rightarrow\mathbb{R},\label{eq:-65}\\
 & \phi=\phi_{2}\circ\left(\phi_{1}|_{P^{E}}\right)\label{eq:-66}
\end{align}

We know by the version we proved for Theorem \ref{thm:RNG-mean} that
$\phi$ is a well defined function on $P^{E}$, which assigns to each
vector of edge lengths the mean value of the RNG for the corresponding
graph.

Note that $\phi$ is locally uniformly continuous, since it is simply
given by the expression
\begin{equation}
\phi\left(\vec{\ell}\right)=\frac{2\sigma}{\sum_{i=1}^{E}\ell_{i}}\sum_{v\in\VR}\frac{1}{\deg\left(v\right)}.\label{eq:-67}
\end{equation}

Since $P^{E}$ is dense in $\mathbb{R}_{+}^{E}$ and $\phi$ is locally
uniformly continuous, it can be extended into a continuous function
$\tilde{\phi}$ on $\mathbb{R}_{+}^{E}$.

Assuming that the composition $\phi_{2}\circ\phi_{1}$ is well defined
and continuous on $\mathbb{R}_{+}^{E}$, we can in fact say that $\tilde{\phi}=\phi_{2}\circ\phi_{1}$,
since the two functions are continuous and identify on a dense subset.
To show that the composition $\phi_{2}\circ\phi_{1}$ is well defined
(meaning, that $Im\left(\phi_{1}\right)\subset Dom\left(\phi_{2}\right)=\mathcal{C}$),
we can simply show that $\phi_{1}$ is continuous, since then
\begin{equation}
\phi_{1}\left(\mathbb{R}_{+}^{E}\right)=\phi_{1}\left(\overline{P^{E}}\right)\subset_{\text{Continuity of }\phi}\overline{\phi_{1}\left(P^{E}\right)}\subset\overline{\mathcal{C}}\subset\mathcal{C},\label{eq:-70}
\end{equation}
where we have used the fact that the set of Cesaro summable sequences
is closed. This continuity will also show the continuity of the composition,
since $\phi_{1}$ and $\phi_{2}$ are both continuous.

We thus want to show that $\phi_{1}$ is continuous. By Theorem $3.1.2$
in \cite{BerKuc_graphs}, the functions $\lambda_{n}^{\vec{\ell}}\left(\sigma\right)$
and $\lambda_{n}^{\vec{\ell}}\left(0\right)$ are continuous in $\vec{\ell}$.
Moreover, by definition,
\begin{equation}
d_{n}^{\vec{\ell}}\left(\sigma\right)=\lambda_{n}^{\vec{\ell}}\left(\sigma\right)-\lambda_{n}^{\vec{\ell}}\left(0\right).\label{eq:-71}
\end{equation}

Then $\phi_{1}$ is continuous in each of its components as the difference
of two continuous functions, and is thus continuous.

Now that we know that $\tilde{\phi}=\phi_{2}\circ\phi_{1}$, our proposition
follows from continuity of $\tilde{\phi}$; For every $\vec{\ell}\in\mathbb{R}_{+}^{E}$,
we can choose a sequence $\vec{\ell}_{n}$ such that $\vec{\ell}_{n}\rightarrow\vec{\ell}$,
and then
\begin{align}
 & \lim_{N\rightarrow\infty}\frac{1}{N}\sum_{n=1}^{N}d_{n}^{\vec{\ell}}\left(\sigma\right)=\phi_{2}\left(\left(d_{n}^{\vec{\ell}}\left(\sigma\right)\right)_{n=1}^{\infty}\right)=\phi_{2}\left(\phi_{1}\left(\vec{\ell}\right)\right)\label{eq:-68}\\
 & =\tilde{\phi}\left(\vec{\ell}\right)=\tilde{\phi}\left(\lim_{n\rightarrow\infty}\vec{\ell}_{n}\right)=\lim_{n\rightarrow\infty}\tilde{\phi}\left(\vec{\ell}_{n}\right)\\
 & =\lim_{n\rightarrow\infty}\frac{2\sigma}{\sum_{i=1}^{E}\ell_{i}^{\left(n\right)}}\sum_{v\in\VR}\frac{1}{\deg\left(v\right)}=\frac{2\sigma}{\sum_{i=1}^{E}\ell_{i}}\sum_{v\in\mathcal{V}_{R}}\frac{1}{\deg\left(v\right)}.\label{eq:-69}
\end{align}

This completes the proof.

\end{proof}
\begin{rem}
It is worth noting that while the mean value converges to the same
value for the rationally dependent case, the behavior of the RNG might
be drastically different than in the rationally independent case.

For instance, for the case of an equilateral star graph, one can show
that the RNG accumulates around two values, and does not get close
to the mean value at all. Nevertheless, the Cesaro sum still converges
to the same mean value, as displayed in Figure \ref{rational-fig}.
\begin{figure}
\includegraphics[scale=0.6]{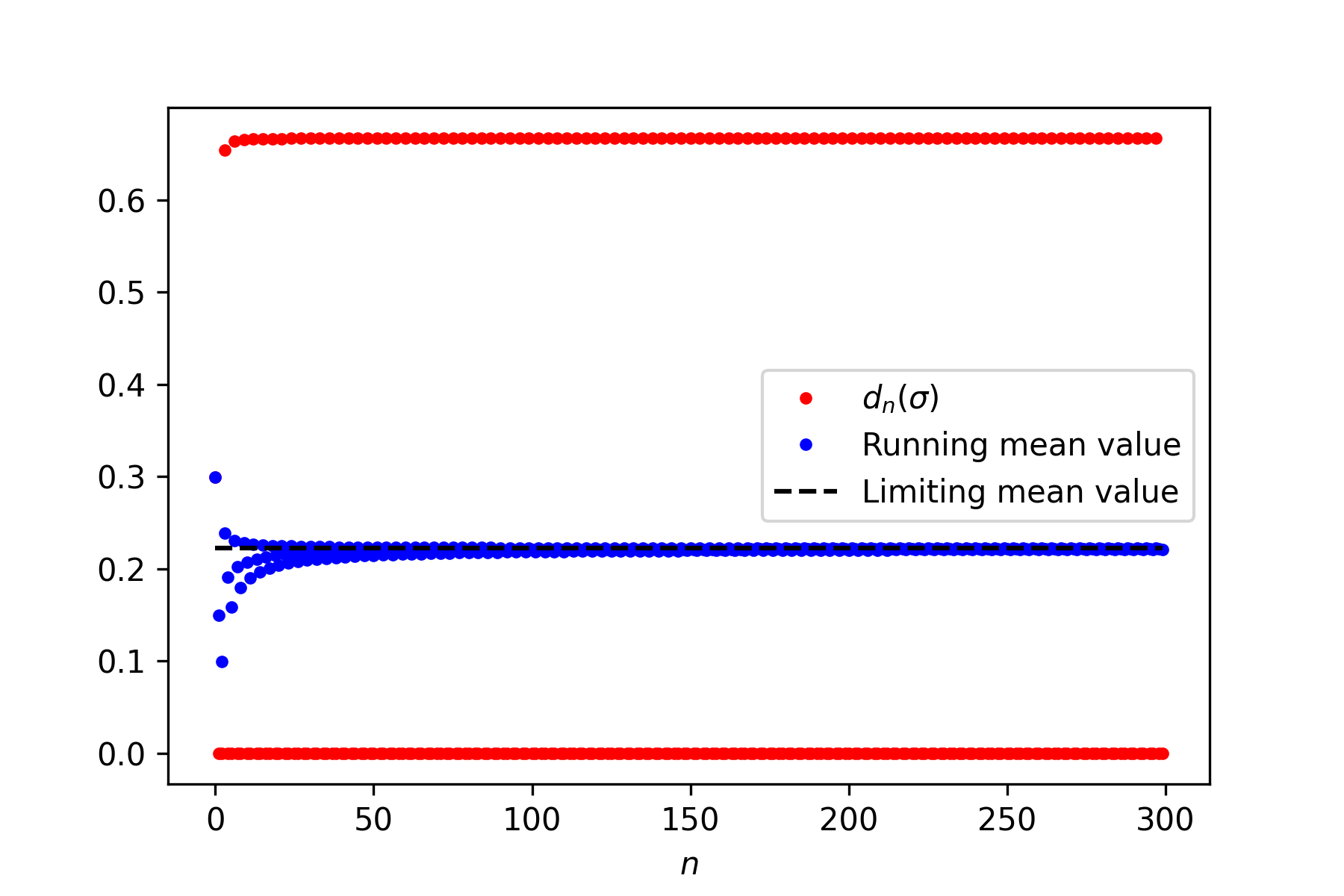}

\caption[RNG on an equilateral star graph.]{\label{rational-fig}While the RNG accumulates around two values for
an equilateral star graph, the mean value still converges to the value
$\frac{2}{L\cdot\deg\left(v\right)}$ which lies between them.}
\end{figure}
\end{rem}

\newpage{}

\section{Tools for proofs of Theorems \ref{thm:SF-index}, \ref{prop:SF-points},
\ref{thm:SF-Betti1} and \ref{prop:SF-Betti2} \label{sec:SF-tools}}

In this section we prove several useful properties of the $\delta_{s}$
family and its spectral curves $\lambda_{n}^{s}\left(t\right)$.

A useful fact that will be used throughout this section is an alternative
definition of the $\delta_{s}$ family from Subsection \ref{subsec:delta-s-definition}
via sesquilinear forms (the relevant computation may be found in Appendix
\ref{sec:Appendices}).

For $s\neq\infty$, one can define the self-adjoint family $\left(H^{s}\left(t\right)\right)_{t\neq0,\infty}$
using the following sesquilinear form on $H^{1}\left(\Gamma\right)$:
\begin{align}
 & \,\,\,\,\,\,\,\,\,\,\,\,\,\,\,\,L_{t}^{s}\left(f,g\right)=\int_{\Gamma}\frac{df}{dx}\overline{\frac{dg}{dx}}dx\label{eq:-4-1}\\
 & -\frac{1}{\sin^{2}\left(\alpha\right)t}\sum_{v\in B}\left(\begin{array}{c}
f_{1}\left(v\right)\\
f_{2}\left(v\right)
\end{array}\right)^{*}\left(\begin{array}{cc}
1+\frac{1}{2}\sin\left(2\alpha\right)t & -1\\
-1 & 1-\frac{1}{2}\sin\left(2\alpha\right)t
\end{array}\right)\left(\begin{array}{c}
g_{1}\left(v\right)\\
g_{2}\left(v\right)
\end{array}\right).
\end{align}

In the case $t=\infty$, the form is given by

\begin{align}
 & \,\,\,\,\,\,\,\,\,\,\,\,\,\,\,\,L_{\infty}^{s}\left(f,g\right)=\int_{\Gamma}\frac{df}{dx}\overline{\frac{dg}{dx}}dx\label{eq:-4-1-2}\\
 & -\frac{1}{\sin^{2}\left(\alpha\right)}\sum_{v\in B}\left(\begin{array}{c}
f_{1}\left(v\right)\\
f_{2}\left(v\right)
\end{array}\right)^{*}\left(\begin{array}{cc}
\frac{1}{2}\sin\left(2\alpha\right) & 0\\
0 & -\frac{1}{2}\sin\left(2\alpha\right)t
\end{array}\right)\left(\begin{array}{c}
g_{1}\left(v\right)\\
g_{2}\left(v\right)
\end{array}\right),
\end{align}
which can be thought of as taking the limit $t\rightarrow\infty$
in (\ref{eq:-4-1}).

In the case $t=0$ the domain of $L_{0}^{s}$ is given by all functions
in $H^{1}\left(\Gamma\right)$ which are continuous at $B$, and the
corresponding form is given by
\begin{equation}
L_{0}^{s}\left(f,g\right)=\int_{\Gamma}\frac{df}{dx}\overline{\frac{dg}{dx}}dx,\label{eq:-100}
\end{equation}

which can also be thought of as a special case of (\ref{eq:-4-1})
by identifying the second term as zero due to the continuity at $B$.

For the case $s=\infty$ we define the $\delta_{\infty}$ family using
the following sesquilinear form on $H^{1}\left(\Gamma\right)$:
\begin{equation}
L_{t}^{\infty}\left(f,g\right)=\int_{\Gamma}\overline{\frac{df}{dx}}\frac{dg}{dx}dx+t\sum_{v\in B}\overline{f\left(v\right)}g\left(v\right),\label{eq:-8-1-1}
\end{equation}

where in the case $t=\infty$ the domain of $L_{\infty}^{\infty}$
is $H_{0}^{1}\left(\Gamma\right)$ (here the boundary of $\Gamma$
is chosen as $B$), and so the second term vanishes.

In other words, $H^{s}\left(t\right)$ is the maximal self-adjoint
extension of the Laplacian which satisfies
\begin{equation}
L_{t}^{s}\left(u,v\right)=\left\langle H^{s}\left(t\right)u,v\right\rangle _{L^{2}},\,\,\forall u\in Dom\left(H^{s}\left(t\right)\right),v\in Dom\left(L_{t}^{s}\right).\label{eq:-77}
\end{equation}

\begin{lem}
\label{lem:sf-bdd-below} For every $s\neq\infty$, the spectral curves
$\lambda_{n}^{s}\left(t\right)$ are real analytic at any $t_{0}\in\R\backslash\left\{ 0\right\} $
such that $\lambda_{n}^{s}\left(t_{0}\right)$ is a simple eigenvalue.
Moreover, for every $\epsilon>0$, they are uniformly bounded from
below for $t\in\R\backslash\left(0,\epsilon\right)$.\\
For $s=\infty$, the spectral curves $\lambda_{n}^{\infty}\left(t\right)$
are real analytic at any $t_{0}\in\R$ such that $\lambda_{n}^{\infty}\left(t_{0}\right)$
is a simple eigenvalue. Moreover, for every $M>0$, they are uniformly
bounded from below for $t\in[-M,\infty)$.
\end{lem}

\begin{proof}
We first prove the analyticity.

On the given two regions, $Dom\left(L_{t}^{s}\right)$ is independent
of $t$. For fixed $f\in Dom\left(L_{t}^{s}\right)$, the complex
valued function $\phi\left(z\right)=L_{z}^{s}\left(f,f\right)$ is
clearly analytic (on $\mathbb{C}$ for $s=\infty$ and on $\mathbb{C}\backslash\left\{ 0\right\} $
for $s\neq\infty$). Thus, by standard perturbation theory (see, for
instance, chapter VII-4 in \cite{Kato_book}), the $\delta_{s}$ family
is a holomorphic family of operators, and the spectral curves $\lambda_{n}^{s}\left(t\right)$
are real analytic on the given regions.

For the boundedness from below, we write the proof for the case $s\neq\infty$,
while the proof for $s=\infty$ is analogous. The quadratic form associated
with the operator $H^{s}\left(t\right)$ is given by
\begin{align}
 & L_{t}^{s}\left(f,f\right)=\int_{\Gamma}\left|\frac{df}{dx}\right|^{2}dx-\sum_{v\in B}\frac{1}{\sin^{2}\left(\alpha\right)t}\Biggl[\left(1+\frac{1}{2}\sin\left(2\alpha\right)t\right)\left|f_{1}\left(v\right)\right|^{2}\label{eq:-30-1-1}\\
 & -2Re\left(f_{1}\overline{f_{2}}\left(v\right)\right)+\left(1-\frac{1}{2}\sin\left(2\alpha\right)t\right)\left|f_{2}\left(v\right)\right|^{2}\Biggr].
\end{align}

By the min-max characterization of the ground state:
\begin{align}
 & \lambda_{1}\left(H^{s}\left(t\right)\right)=\min_{\left\{ f\in Dom\left(L_{t}^{s}\right):\left\Vert f\right\Vert =1\right\} }L_{t}^{s}\left(f,f\right)\label{eq:-78-1}\\
= & \min_{\left\{ f\in Dom\left(L_{t}^{s}\right):\left\Vert f\right\Vert =1\right\} }\int_{\Gamma}\left|\frac{df}{dx}\right|^{2}dx-\frac{1}{\sin^{2}\left(\alpha\right)t}\sum_{v\in B}\Biggl[\left(1+\frac{1}{2}\sin\left(2\alpha\right)t\right)\left|f_{1}\left(v\right)\right|^{2}\label{eq:-31-1-1}\\
 & -2Re\left(f_{1}\overline{f_{2}}\left(v\right)\right)+\left(1-\frac{1}{2}\sin\left(2\alpha\right)t\right)\left|f_{2}\left(v\right)\right|^{2}\Biggr].
\end{align}

Given $\epsilon>0$, the quantity above is clearly uniformly bounded
from below for $t\in\R\backslash\left(0,\epsilon\right)$, which gives
the result.
\end{proof}
\begin{rem}
At the points where $\lambda_{n}^{s}\left(t\right)$ is not a simple
eigenvalue, the results in \cite{Kato_book} show that the spectral
curves are still continuous. Nevertheless, at the forbidden region
$\left(0,\epsilon\right)$, one can show that some of the spectral
curves diverge to $-\infty$ near $t=0^{+}$ (see, for instance, Figures
\ref{fig:Betti} and \ref{monoton-fig}).\\
Theorem \ref{prop:SF-points} shows that the number of curves which
tend to $-\infty$ is equal to the number of points where the $\delta_{s}$
condition is imposed.

\begin{figure}
\includegraphics[scale=0.7]{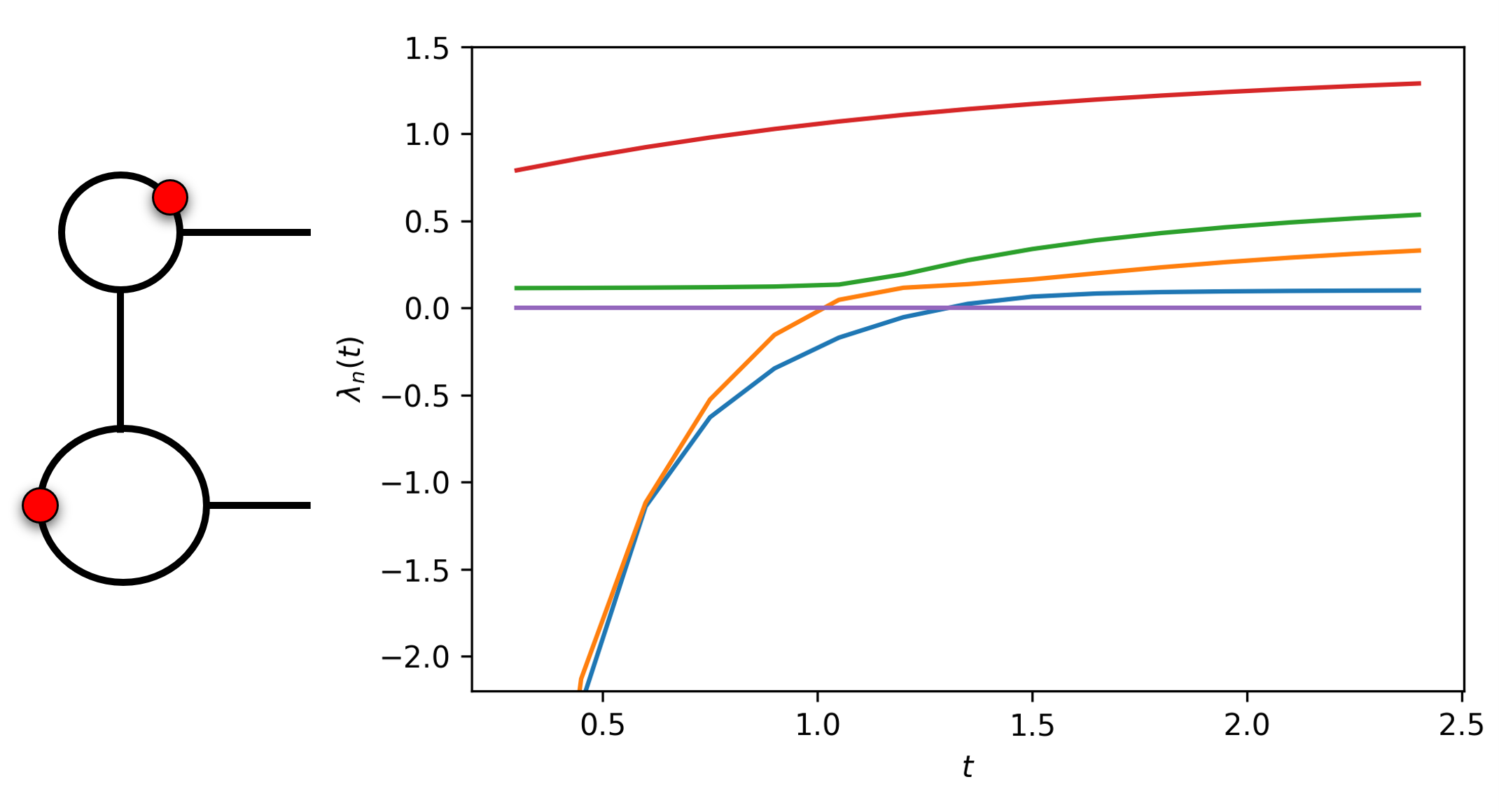}

\caption[Monotonicity and boundedness from below of spectral curves.]{\label{monoton-fig}Demonstration of Lemmas \ref{lem:sf-bdd-below}
and \ref{lem:monotonicity} for the $\delta_{0}$ flow on glasses
graph. The spectral curves are monotone increasing with $t$, and
are uniformly bounded from below on $\left(0,\epsilon\right)$, with
two spectral curves approaching $-\infty$ near $0^{+}$. }
\end{figure}
\end{rem}

\begin{lem}
\label{lem:monotonicity}The spectral curves $\lambda_{n}^{s}\left(t\right)$
are monotone increasing with $t$ on every region where they are continuous
(either $\mathbb{R}$ for $s=\infty$ or the two half lines for $s\neq\infty$).
\end{lem}

\begin{proof}
By Lemma \ref{lem:sf-bdd-below}, the spectral curves $\lambda_{n}^{s}\left(t\right)$
are real analytic for every $t_{0}$ such that $\lambda_{n}^{s}\left(t_{0}\right)$
is a simple eigenvalue. Fix such $n,t_{0}$ and denote by $f_{n}$
the $L^{2}$ normalized eigenfunction corresponding to $\lambda_{n}^{s}\left(t_{0}\right)$.

By definition of $L_{t}^{s}$, we have that
\begin{equation}
L_{t_{0}}^{s}\left(f_{n},g\right)=\lambda_{n}^{s}\left(t_{0}\right)\left\langle f_{n},g\right\rangle ,\,\,\forall g\in Dom\left(L_{t_{0}}^{s}\right).\label{eq:-9-1}
\end{equation}

By Theorem $3.1.2$ in \cite{BerKuc_graphs}, $f_{n}$ depends analytically
on $t$ at $t_{0}$ as well. We can thus differentiate both sides
with respect to $t$ and get
\begin{equation}
\frac{\partial L_{t_{0}}^{s}\left(f_{n},g\right)}{\partial t}+L_{t_{0}}^{s}\left(\frac{\partial f_{n}}{\partial t},g\right)=\frac{\partial\lambda_{n}^{s}\left(t_{0}\right)}{\partial t}\left\langle f_{n},g\right\rangle +\lambda_{n}^{s}\left(t_{0}\right)\left\langle \frac{\partial f_{n}}{\partial t},g\right\rangle .\label{eq:-10-1}
\end{equation}

Choosing $g=f_{n}$ we obtain
\begin{equation}
\frac{\partial L_{t_{0}}^{s}\left(f_{n},f_{n}\right)}{\partial t}+L_{t_{0}}^{s}\left(\frac{\partial f_{n}}{\partial t},f_{n}\right)=\frac{\partial\lambda_{n}^{s}\left(t_{0}\right)}{\partial t}\left\langle f_{n},f_{n}\right\rangle +\lambda_{n}^{s}\left(t_{0}\right)\left\langle \frac{\partial f_{n}}{\partial t},f_{n}\right\rangle .\label{eq:-11-1}
\end{equation}

Recall that $\left\langle f_{n},f_{n}\right\rangle =1$ and that $L_{t_{0}}^{s}\left(\frac{\partial f_{n}}{\partial t},f_{n}\right)=\lambda_{n}^{s}\left(t_{0}\right)\left\langle \frac{\partial f_{n}}{\partial t},f_{n}\right\rangle $,
and so
\begin{equation}
\frac{\partial\lambda_{n}^{s}\left(t_{0}\right)}{\partial t}=\frac{\partial L_{t_{0}}^{s}\left(f_{n},f_{n}\right)}{\partial t}=\begin{cases}
\sum_{v\in B}\frac{1}{\sin^{2}\left(\alpha\right)t^{2}}\left|f_{n,2}\left(v\right)-f_{n,1}\left(v\right)\right|^{2} & s\neq\infty\\
\sum_{v\in B}\left|f_{n}\left(v\right)\right|^{2} & s=\infty
\end{cases}\label{eq:-12-1}
\end{equation}

The expression above is non-negative, which gives the desired result
for any point such that $\lambda_{n}^{s}\left(t\right)$ is simple.

Once again by standard theory, for every given $n$, $\lambda_{n}^{s}\left(t\right)$
can only be a multiple eigenvalue on a discrete set of points (see
also proof of Lemma \ref{lem:hadamard}), in which $\lambda_{n}^{s}\left(t\right)$
is still continuous, which means that the monotonicity in fact holds
for all $t$ in the given regions.
\end{proof}
The following useful relationship between the $\delta_{s}$ family
and the Robin map (see Subsection \ref{subsec:DTN-definition}) will
be a key argument in the proof of Theorem \ref{thm:SF-index}.
\begin{lem}
\label{lem:DTN-correspondence}Let $c\notin\spec{H^{s}\left(\infty\right)}$.
Then $-t\in\mathbb{R}$ is an eigenvalue of $\Lambda_{s}\left(c\right)$
if and only if $c$ is an eigenvalue of $H^{s}\left(t\right)$ (with
identical multiplicity).\\
Consequently,
\begin{equation}
\dim\ker\left(H^{s}\left(t\right)-c\right)=\dim\ker\left(\Lambda_{s}\left(c\right)+t\right).\label{eq:-83}
\end{equation}
\end{lem}

\begin{proof}
First, assume that $-t$ is an eigenvalue of $\Lambda_{s}\left(c\right)$
with eigenvector $w\in\ell^{2}\left(B\right)$:
\begin{equation}
\Lambda_{s}\left(c\right)w=-tw.\label{eq:-18-1}
\end{equation}
 Then the solution $u$ to the corresponding boundary value problem
satisfies
\begin{equation}
-\frac{\partial^{2}u}{\partial x^{2}}=cu,\label{eq:-19-1}
\end{equation}
and on the boundary $B$ it satisfies
\begin{align}
 & \gamma_{1}^{s}\left(u\right)=\gamma_{2}^{s}\left(u\right),\label{eq:-20-1}\\
 & \gamma_{1}^{s\star}\left(u\right)-\gamma_{2}^{s\star}\left(u\right)=\Lambda_{s}\left(c\right)\gamma^{s}\left(u\right)=-t\gamma^{s}\left(u\right),\label{eq:-21-1}
\end{align}
which means that $u$ is an eigenfunction of $H^{s}\left(t\right)$
with eigenvalue $c$.

For the contrary, if $c$ is an eigenvalue of $H^{s}\left(t\right)$
with eigenfunction $u$, then $u$ must satisfy the $\delta_{s}$
vertex condition at $B$:
\begin{equation}
\gamma_{2}^{s\star}\left(u\right)-\gamma_{1}^{s\star}\left(u\right)=t\gamma^{s}\left(u\right).\label{eq:-13}
\end{equation}

Choosing $w=\gamma^{s}\left(u\right)$ as input to $\Lambda_{s}\left(c\right)$
gives
\begin{equation}
\Lambda_{s}\left(c\right)w=\gamma_{1}^{s\star}\left(u\right)-\gamma_{2}^{s\star}\left(u\right)=-t\gamma^{s}\left(u\right)=-tw,\label{eq:-22-1}
\end{equation}
and so $-t$ is an eigenvalue of $\Lambda_{s}\left(c\right)$ with
eigenvector $w$, assuming that $w\neq0$.

Assume by contradiction that $w=0$. Then $w$ corresponds to the
vertex condition $\gamma^{s}\left(u\right)=0$ on $B$. We thus conclude
that $u$ (which is the solution to the corresponding boundary value
problem) is an eigenfunction of $H^{s}\left(\infty\right)$ with the
same eigenvalue $c$. But we assumed that $c\notin\spec{H^{s}\left(\infty\right)}$,
which gives a contradiction.
\end{proof}
\newpage{}

\section{Proofs for Theorems \ref{thm:SF-index}, \ref{prop:SF-points}, \ref{thm:SF-Betti1}
and \ref{prop:SF-Betti2} \label{sec:proof_for_sf}}

\subsection{Proof of Theorem \ref{thm:SF-index}}

Fix a generic eigenpair $\left(\lambda_{n},f_{n}\right)$ of $H_{0}:=H^{s}\left(0\right)$
as in the statement of Theorem \ref{thm:SF-index}. Now, consider
the $\delta_{s}$ family of Hamiltonians, where the $\delta_{s}$
points are placed at the $s$ points of the eigenfunction $f_{n}$.
\begin{lem}
\label{lem:const}For every $t$, $f_{n}$ is an eigenfunction of
$H^{s}\left(t\right)$ with eigenvalue $\lambda_{n}$.
\end{lem}

\begin{proof}
Since $\left(\lambda_{n},f_{n}\right)$ is an eigenpair of $H_{0}$,
then naturally
\begin{equation}
-\frac{d^{2}f_{n}}{dx^{2}}=\lambda_{n}f_{n}.\label{eq:-106}
\end{equation}

We thus only need to verify that $f_{n}\in Dom\left(H^{s}\left(t\right)\right)$
for every $t$. But since the $\delta_{s}$ points for the family
$H^{s}\left(t\right)$ were chosen as the set of $s$ points of $f_{n}$
(recall Definition \ref{def:s-points}), then
\begin{align}
 & \gamma_{1}^{s}\left(f\right)=\gamma_{2}^{s}\left(f\right)=0\label{eq:-107}\\
 & \gamma_{2}^{s\star}\left(f\right)-\gamma_{1}^{s\star}\left(f\right)=t\gamma^{s}\left(f\right)=0,\label{eq:-108}
\end{align}

and so $f_{n}\in Dom\left(H^{s}\left(t\right)\right)$ by definition.
\end{proof}
For brevity, we denote the number of $s$ points of $f_{n}$ by $\phi_{s}$
and the number of $s$ domains of $f_{n}$ by $\nu_{s}$. Moreover,
we denote the multiplicity of the eigenvalue $\lambda_{n}$ is $\spec{H^{s}\left(t\right)}$
by $\mult{\lambda_{n}}t$.
\begin{lem}
\label{lem:multiplicity}The following holds:
\begin{equation}
\mult{\lambda_{n}}{\infty}=\nu_{s}.\label{eq:-46}
\end{equation}
\end{lem}

\begin{proof}
Denote $m=\mult{\lambda_{n}}{\infty}$. First, note that $m\geq\nu_{s}$.
This is true since to any $s$ domain $\Omega$ (which is just a subgraph
of $\Gamma$), we can match a linearly independent eigenfunction of
$H^{s}\left(\infty\right)$:
\begin{equation}
f_{\Omega}\left(x\right)=\begin{cases}
f_{n}\left(x\right) & x\in\Omega\\
0 & \text{Otherwise}
\end{cases}\label{eq:-26-1}
\end{equation}

We remind the reader that domain of the operator $H^{s}\left(\infty\right)$
consists of functions which satisfy the $s$ condition at the selected
subset of points $B$:
\begin{equation}
f'\left(x\right)=sf\left(x\right).\label{eq:-105}
\end{equation}

In our case, $B$ is just the collection of $s$ points of the function
$f_{n}$. We thus see that by construction, $f_{\Omega}$ are all
indeed eigenfunctions of $H^{s}\left(\infty\right)$ with eigenvalue
$\lambda_{n}$, and there are thus at least $\nu_{s}$ such linearly
independent eigenfunctions. So $m\geq\nu_{s}$.

Assume by contradiction that $m>\nu_{s}$. Then there is an eigenfunction
$g$ which is linearly independent of all $f_{\Omega}$ above. In
particular, there is some $s$ domain $\Omega$ such that $g|_{\Omega}\not\equiv0$
and $g|_{\Omega}\not\propto f_{\Omega}|_{\Omega}$. This gives two
linearly independent eigenfunctions for the Laplacian on the subgraph
$\Omega$ -- $f_{\Omega}|_{\Omega}$ and $g|_{\Omega}$. This means
that $\lambda_{n}$ is not a simple eigenvalue of $H_{0}$ on $\Omega$.

Since $\lambda_{n}>\frac{\pi}{L_{min}}$, $\Omega$ is a star graph
(see Proposition $3.9$ in \cite{AloBan_21ahp}). By a simple adaption
of Corollary $3.1.9$ in \cite{BerKuc_graphs}, there exists an internal
vertex $v\in\Omega$ such that $f_{n}\left(v\right)=0$. This contradicts
the genericity of $f_{n}$ (recall Definition \ref{def:generic}).
\end{proof}
\begin{figure}
\includegraphics[scale=0.6]{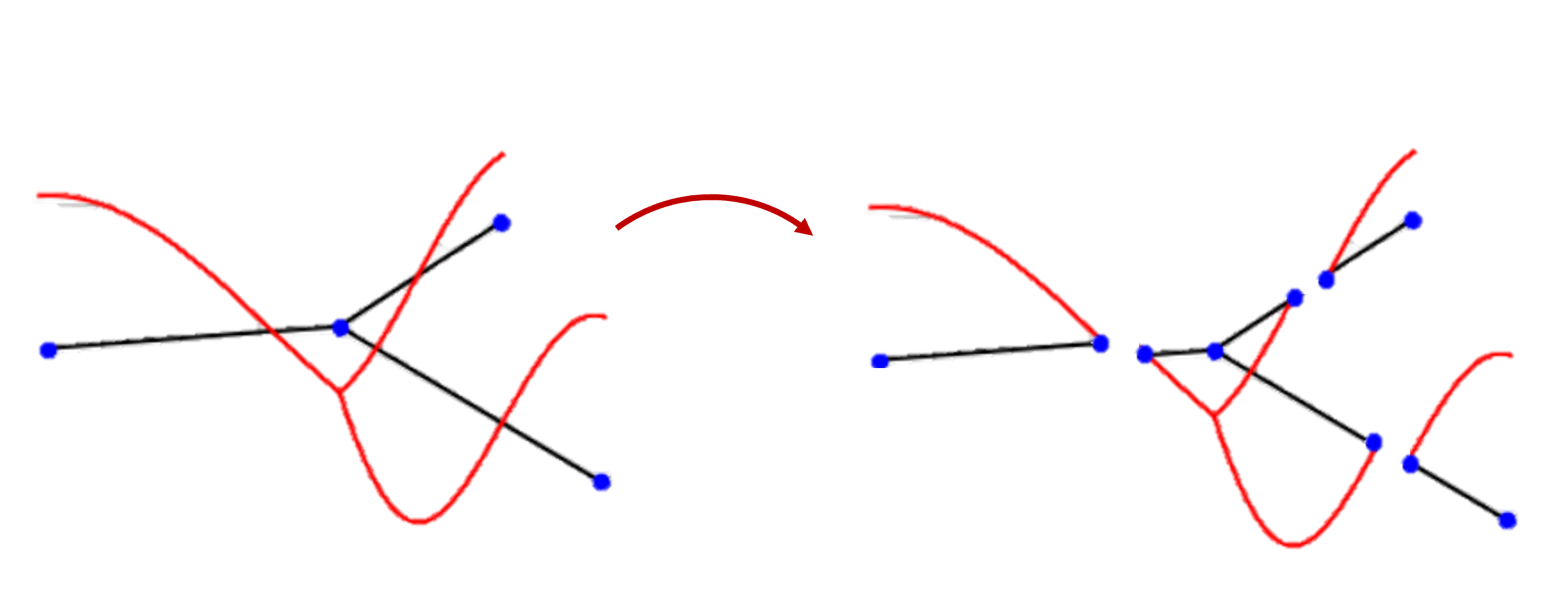}

\caption[Idea of the proof of Lemma \ref{lem:multiplicity}.]{Idea of the proof of Lemma \ref{lem:multiplicity} for the case $s=\infty$.
$t=0$ corresponds to imposing the Dirichlet condition at the selected
set of nodal points, which naturally matches an eigenfunction of $H^{\infty}\left(\infty\right)$
with eigenvalue $\lambda_{n}$ to each nodal domain.\label{fig:nodal-domains}}
\end{figure}

\begin{lem}
\label{lem:s-position}Denote the spectral position of $\lambda_{n}$
at $t=-\infty$ by $p$. Meaning,
\begin{equation}
p:=1+\#\left\{ \lambda\in\spec{H^{s}\left(\infty\right)}:\lambda<\lambda_{n}\right\} .\label{eq:-96}
\end{equation}

Then
\begin{equation}
p=\begin{cases}
1 & s=\infty\\
1+\phi_{\infty} & s\neq\infty
\end{cases}\label{eq:-78}
\end{equation}
\end{lem}

\begin{proof}
Let us first consider the case $s=\infty$, which corresponds to the
Dirichlet condition. Let $\Omega$ be an arbitrary nodal domain of
$f_{n}$. Consider the function
\begin{equation}
f_{\Omega}\left(x\right)=\begin{cases}
f_{n}\left(x\right) & x\in\Omega\\
0 & \text{Otherwise}
\end{cases}\label{eq:-28-1-1}
\end{equation}

Then as shown in Lemma \ref{lem:multiplicity}, $f_{\Omega}$ is an
eigenfunction of $H^{\infty}\left(\infty\right)$ with eigenvalue
$\lambda_{n}$.

Moreover, it has a single nodal domain -- $\Omega$. By a well known
consequence of Courant's nodal theorem (see Theorem $5.2.6$ and Remark
$5.2.9$ in \cite{BerKuc_graphs}), since $f_{\Omega}$ is an eigenfunction
with a single nodal domain, then it is a ground state. So $\lambda_{n}$
is the smallest eigenvalue of $H^{\infty}\left(\infty\right)$ and
$p=1$. 

We now prove the case $s\neq\infty$. By definition, $p$ is equal
to the number of eigenvalues $\lambda'$ of $H^{s}\left(\infty\right)$
such that $\lambda'<\lambda_{n}$, plus one. Our strategy for counting
these eigenvalues will be to partition $\Gamma$ into subgraphs according
to the $s$ domains of $f_{n}$, and counting the eigenfunctions of
each such subgraph.

Denoting the set of $s$ domains of $f_{n}$ by $D$ and the spectral
position of $\lambda_{n}$ as an eigenvalue of the subgraph $\Omega$
by $p\left(\Omega\right)$, we have the following formula for the
spectral position $p$:
\begin{equation}
p=1+\sum_{\Omega\in D}\left(p\left(\Omega\right)-1\right).\label{eq:-27-1}
\end{equation}

The formula above is obtained by using the eigenfunctions of $H^{s}\left(\infty\right)$
on each subgraph to define an eigenfunction of $H^{s}\left(\infty\right)$
on the entire graph (similar to Lemma \ref{lem:multiplicity}). On
each $s$ domain $\Omega$, there are $p\left(\Omega\right)-1$ eigenvalues
of $H^{s}\left(\infty\right)$ smaller than $\lambda_{n}$, which
we denote by $\left(\lambda_{k}^{\Omega}\right)_{k=1}^{p\left(\Omega\right)-1}$
with eigenfunctions $\left(\tilde{f}_{k}^{\Omega}\right)_{k=1}^{p\left(\Omega\right)-1}$.
Note that by the same argument as in Lemma \ref{lem:multiplicity},
each $\lambda_{k}^{\Omega}$ is also an eigenvalue of $H^{s}\left(\infty\right)$
on $\Gamma$ with eigenfunction
\begin{equation}
f_{k}^{\Omega}\left(x\right)=\begin{cases}
\tilde{f}_{k}^{\Omega}\left(x\right) & x\in\Omega\\
0 & \text{Otherwise}
\end{cases}\label{eq:-28-1}
\end{equation}

This overall gives us $\sum_{\Omega\in D}\left(p\left(\Omega\right)-1\right)$
linearly independent eigenfunctions for eigenvalues of $H^{s}\left(\infty\right)$
which are smaller than $\lambda_{n}$.

We claim that there are no other eigenfunctions of $\Gamma$ with
eigenvalues smaller than $\lambda_{n}$. Indeed, assume that $g$
is an eigenfunction with eigenvalue $\lambda_{g}<\lambda_{n}$. Choose
some $s$ domain $\Omega$ such that $g|_{\Omega}\not\equiv0$. Note
that $g|_{\Omega}$ is also an eigenfunction of $H_{0}$ on $\Omega$
with eigenvalue $\lambda_{g}$. But this means that we have already
counted $g|_{\Omega}$ as one of the eigenfunctions $\tilde{f}_{k}^{\Omega}\left(x\right)$
from before. This proves the formula for $p$ in (\ref{eq:-27-1}).

Since $\lambda_{n}>\frac{\pi}{L_{min}}$, then on each $s$ domain
$\Omega$ (which is now a star subgraph of $\Gamma$), the spectral
position of $\lambda_{n}$ as an eigenvalue of $H_{0}$ is equal to
$\phi_{\infty}\left(f_{n}|_{\Omega}\right)+1$ (see Lemma $3.1$ in
\cite{AloBan_21ahp}). Plugging this into (\ref{eq:-27-1}) we get
\begin{equation}
p=1+\sum_{\Omega\in D}\left(p\left(\Omega\right)-1\right)=1+\sum_{\Omega\in D}\phi_{\infty}\left(f_{n}|_{\Omega}\right)=1+\phi_{\infty},\label{eq:-29-1}
\end{equation}
which completes the proof.
\end{proof}
\begin{rem}
We remind the reader that for the case $s=\infty$, the assumptions
in Theorem \ref{thm:SF-index} that $\lambda_{k}>\frac{\pi}{\ell_{min}}$
and that $f_{k}$ is generic are unnecessary (see Remark \ref{rem:unnecessary}).
These assumptions were used in Lemmas \ref{lem:multiplicity} and
\ref{lem:s-position} above to guarantee that each $s$ domain of
the eigenfunction $f_{n}$ is a star graph and to apply Corollary
$3.1.9$ in \cite{BerKuc_graphs}. For the case $s=\infty$, these
lemmas hold without this additional assumption.
\end{rem}

\begin{proof}[Proof of Theorem \ref{thm:SF-index}]
 We first write the proof for $s\neq\infty$. Consider the $\delta_{s}$
family placed on the set of $s$ points of $f_{n}$. The idea of the
proof will be to compute the spectral flow along $\left[-\infty,0\right]$
(see also Figure \ref{fig:index-demo}).

By Lemmas \ref{lem:multiplicity} and \ref{lem:s-position}, at $t=-\infty$
there are exactly $\nu_{s}+\phi_{\infty}$ spectral curves below $\lambda_{n}+\epsilon$.
On the other hand, there are exactly $n$ spectral curves below $\lambda_{n}+\epsilon$
at $t=0$.

Since the spectral curves are monotone increasing and continuous on
all of $(-\infty,0]$ (Lemma \ref{lem:sf-bdd-below}), we conclude
that along $\left(-\infty,0\right)$, exactly $\nu_{s}+\phi_{\infty}-n$
spectral curves intersect $\lambda_{n}+\epsilon$, giving rise to
$\nu_{s}+\phi_{\infty}-n$ positive eigenvalues of $\Lambda_{s}\left(\lambda_{n}+\epsilon\right)$
in the process (Lemma \ref{lem:DTN-correspondence}). Thus,
\begin{align}
 & Pos\left(\Lambda_{s}\left(\lambda_{n}+\epsilon\right)\right)=\nu_{s}+\phi_{\infty}-n\label{eq:-72}\\
 & \Rightarrow\mathcal{D}_{s}\left(f_{n}\right)=n-\nu_{s}=\phi_{\infty}\left(f_{n}\right)-Pos\left(\Lambda_{s}\left(\lambda_{n}+\epsilon\right)\right).\label{eq:-32-1}
\end{align}

For $s=\infty$, we repeat the same procedure, counting the spectral
flow through $\lambda_{n}+\epsilon$ along $\left[0,\infty\right]$
instead of $\left[-\infty,0\right]$. This time, Lemma \ref{lem:s-position}
says that at $t=\infty$, $\lambda_{n}$ is the ground state.

The same intersection counting argument as before now gives
\begin{equation}
\mathcal{D}_{\infty}\left(f_{n}\right)=Mor\left(\Lambda_{\infty}\left(\lambda_{n}+\epsilon\right)\right),\label{eq:-8-1}
\end{equation}
which finishes the proof.
\end{proof}
\begin{rem}
If $\lambda_{n}$ is not assumed to be a simple eigenvalue, then by
the same proof as above, Formulas (\ref{eq:-32-1}) and (\ref{eq:-8-1})
easily generalize to
\begin{align}
 & \mathcal{D}_{s}=\phi_{\infty}+1-\dim\left(Ker\left(H_{0}-\lambda_{n}\right)\right)-Pos\left(\Lambda_{s}\left(\lambda_{n}+\epsilon\right)\right),\label{eq:-73}\\
 & \mathcal{D}_{\infty}\left(f_{n}\right)=Mor\left(\Lambda_{s}\left(\lambda_{n}+\epsilon\right)\right)+1-\dim\left(Ker\left(H_{0}-\lambda_{n}\right)\right).\label{eq:-74}
\end{align}
\end{rem}

A demonstration of Theorem \ref{thm:SF-index} can be found in Figure
\ref{fig:index-demo} for $s=0$ (Neumann points). We fix the fourth
eigenfunction $f_{4}$, which has three Neumann points (red) and four
nodal points (blue). Since $f_{4}$ has two Neumann domains, $\mathcal{D}_{0}\left(f_{4}\right)=4-2=2$.
The $\delta_{0}$ family is placed on the Neumann points of $f_{4}$.

By Lemma \ref{lem:DTN-correspondence}, $Pos\left(\Lambda_{0}\left(\lambda_{4}+\epsilon\right)\right)$
is equal to the spectral flow along $\left[-\infty,0\right]$ through
the horizontal line $y=\lambda_{4}+\epsilon$ (dashed black line),
which is seen to be two. Combining everything we get
\begin{equation}
\mathcal{D}_{0}\left(f_{4}\right)=n-\nu_{0}\left(f_{4}\right)=4-2=\phi_{\infty}\left(f_{n}\right)-Pos\left(\Lambda_{0}\left(\lambda_{4}+\epsilon\right)\right),\label{eq:-119}
\end{equation}
as anticipated by the Theorem.

\begin{figure}
\includegraphics[scale=0.7]{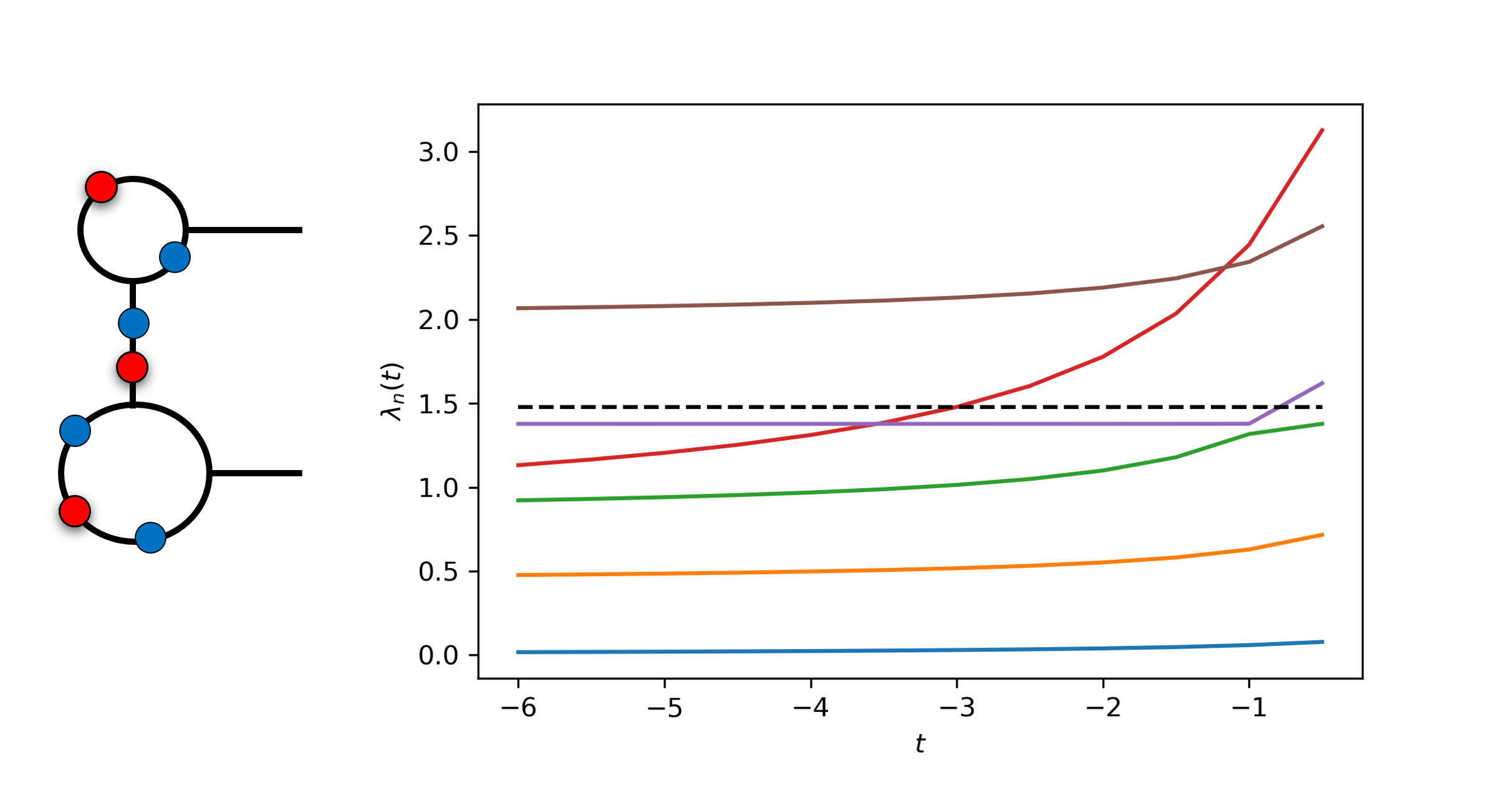}

\caption{\label{fig:index-demo}Demonstration of Theorem \ref{thm:SF-index}
for $s=0$ (Neumann points).}
\end{figure}

\begin{rem}
By Remark \ref{rem:DTN-matrix}, when the Robin map $\Lambda_{s}\left(\lambda_{n}+\epsilon\right)$
is placed at the $s$ points of $f_{n}$, then , $\Lambda_{s}\left(\lambda_{n}+\epsilon\right)$
be written as a square matrix of size $\phi_{s}\left(f_{n}\right)$.
This means that
\begin{equation}
Mor\left(\Lambda_{s}\left(\lambda_{n}+\epsilon\right)\right)+Pos\left(\Lambda_{s}\left(\lambda_{n}+\epsilon\right)\right)=\phi_{s}\left(f_{n}\right).\label{eq:-82}
\end{equation}

From this we see that taking $s\rightarrow\infty$ in Formula (\ref{eq:-32-1}),
we once again obtain Formula (\ref{eq:-8-1}). Colloquially:
\begin{equation}
\lim_{s\rightarrow\infty}\mathcal{D}_{s}\left(f_{n}\right)=\mathcal{D}_{\infty}\left(f_{n}\right).\label{eq:-92}
\end{equation}

In this sense, the index formula for the $s$ deficiency is ``continuous''
in $s$.
\end{rem}

\bigskip

\subsection{Proof of Theorems \ref{prop:SF-points}, \ref{thm:SF-Betti1} and
\ref{prop:SF-Betti2}}
\begin{proof}[Proof of Theorem \ref{prop:SF-points}]
 The theorem is an immediate corollary of Lemma \ref{lem:DTN-correspondence}.
Due to the lemma, $\sf c^{\delta_{s}}\left[-\infty,\infty\right]$
is equal to the number of eigenvalues of $\Lambda_{s}\left(c\right)$.
Since $\Lambda_{s}\left(c\right)$ can be written as a square matrix
whose size is equal to the number of selected $\delta_{s}$ points
(see Remark \ref{rem:DTN-matrix}), we conclude
\begin{equation}
\sf c^{\delta_{s}}\left[-\infty,\infty\right]=\text{no. of }\delta_{s}\text{ points}=\left|B\right|.\label{eq:-80}
\end{equation}
\end{proof}
\begin{proof}[Proof of Theorem \ref{thm:SF-Betti1}]
 By Theorem \ref{prop:SF-points}, we know that $\sf{\epsilon}^{\delta_{0}}\left[-\infty,\infty\right]=\left|B\right|$
for $\epsilon>0$ small enough. Recall that the spectral curves are
continuous and monotone increasing. Then by an intersection counting
argument of the spectral curves, for $\epsilon>0$ small enough we
have that
\begin{equation}
\sf{+\varepsilon}^{\delta_{0}}\left[-\infty,0\right]-\sf{-\varepsilon}^{\delta_{0}}\left[-\infty,0\right]=\mult{\lambda=0}{H^{0}\left(-\infty\right)}-\mult{\lambda=0}{H^{0}\left(0\right)}.\label{eq:-34-1}
\end{equation}

Moreover, for $t<0$ the operator $H^{0}\left(t\right)$ is non-negative,
as can be seen from the quadratic form $L_{t}^{0}\left(f,f\right)$,
see proof of Lemma \ref{lem:sf-bdd-below}. This means that there
are no spectral curves in the third quadrant, and so $\sf{-\varepsilon}^{\delta_{0}}\left[-\infty,0\right]=0$. 

Note that $\mult{\lambda=0}{H^{0}\left(0\right)}$ is equal to one,
with a constant eigenfunction. Similarly, $\mult 0{H^{0}\left(-\infty\right)}$
is equal to the number of connected components of the cut graph (denoted
by $\pi_{0}\left(\Gamma_{cut}\right)$), with piecewise constant eigenfunctions
supported on each individual connected component.

Combining all of these we get
\begin{align}
 & \sf{+\varepsilon}^{\delta_{0}}\left[0,\infty\right]=\sf{+\varepsilon}^{\delta_{0}}\left[-\infty,\infty\right]-\sf{+\varepsilon}^{\delta_{0}}\left[-\infty,0\right]=\left|B\right|-\pi_{0}\left(\Gamma_{cut}\right)+1.\label{eq:-35-1}
\end{align}

Using the definition of the first Betti number for the graphs $\Gamma$
and $\Gamma_{cut}$\footnote{For a graph with $C$ connected components, we have that $\beta=E-V+C$}
we note that
\begin{equation}
\pi_{0}\left(\Gamma_{cut}\right)-\pi_{0}\left(\Gamma\right)=\left|B\right|+\beta_{\Gamma_{cut}}-\beta_{\Gamma}.\label{eq:-36-1}
\end{equation}

Collecting all of the above we finally obtain
\begin{equation}
\sf{+\epsilon}^{\delta_{0}}\left[0,\infty\right]=\beta_{\Gamma}-\beta_{\Gamma_{cut}}.\label{eq:-37-1}
\end{equation}
\end{proof}
\begin{proof}[Proof of Theorem \ref{prop:SF-Betti2}]
 Since we assume that $\lambda_{n}>\frac{\pi}{\ell_{min}}$, then
the following relation holds (see \cite{AloBan_21ahp}):
\begin{equation}
\nu_{s}\left(f_{n}\right)=\phi_{s}\left(f_{n}\right)-\beta_{\Gamma}+1.\label{eq:-56}
\end{equation}

By Lemma \ref{lem:multiplicity},
\begin{equation}
\mult{\lambda_{n}}{H^{s}\left(-\infty\right)}=\nu_{s}=\phi_{s}-\beta_{\Gamma}+1.\label{eq:-81}
\end{equation}

Since the $\delta_{s}$ points are placed on the $s$ points of $f_{n}$,
then $\lambda_{n}\in\spec{H^{s}\left(t\right)}$ for all $t$, which
means that $y=\lambda_{n}$ is a constant spectral curve (see Lemma
\ref{lem:const}).

Since $\lambda_{n}$ is simple (due to the genericity assumption),
there is exactly one such flat spectral curve. This means that all
other $\phi_{s}-\beta_{\Gamma}$ spectral curves that are equal to
$\lambda_{n}$ at $t=-\infty$ must intersect $\lambda_{n}+\epsilon$
at $\left(-\infty,\infty\right)$. This gives $\phi_{s}-\beta_{\Gamma}$
intersections of spectral curves with $\lambda_{n}+\epsilon$.

By Theorem \ref{prop:SF-points}, there should overall be $\phi_{s}$
such intersections along $\left[-\infty,\infty\right]$. We thus conclude
that $\beta_{\Gamma}$ additional spectral curves must cross through
$\lambda_{n}+\epsilon$ as well.

By assumption, these curves are not equal to $\lambda_{n}$ at $t=\pm\infty$,
and so they must also intersect $\lambda_{n}$ in $\left(-\infty,\infty\right)$,
which gives us exactly $\beta_{\Gamma}$ such intersections.
\end{proof}
\newpage{}

\section{Discussion and further remarks \label{sec:Discussion}}

The results presented in this work are naturally classified into two
types -- the ones involving the Robin-Neumann gap and the ones involving
the spectral flow. While these two types are different in nature,
they are both related to the collective behavior of the spectral curves,
each from a different point of view.

Combining these two quantities provides us with information not just
about the spectral curves themselves, but also about the metric graph
(the edge lengths, the vertex degrees, the Betti number) and the corresponding
eigenfunctions (through the $s$ deficiency).

\subsection{Discussion of RNG results}

The properties of the Robin-Neumann gap presented in this work tell
us that the spectral curves display some uniform regularity; They
are uniformly Lipschitz, they have some common accumulation points,
and on average, their growth is linear in the Robin parameter $\sigma$.

\bigskip

By \textbf{Theorem \ref{thm:RNG-mean}}, the mean value of the RNG
is given by:
\begin{equation}
\left\langle d\right\rangle _{n}\left(\sigma\right)=\frac{2\sigma}{L}\sum_{v\in\VR}\frac{1}{\deg\left(v\right)}.\label{eq:-98}
\end{equation}

This expression bears obvious similarity to the result introduced
in \cite{RudWigYes_arxiv21}, which states that for a bounded planar
domain $\Omega$ and $\sigma>0$
\begin{equation}
\left\langle d\right\rangle _{n}\left(\sigma\right)=\frac{2\text{Length}\left(\partial\Omega\right)}{\text{Area\ensuremath{\left(\Omega\right)}}}\sigma.\label{eq:-33-2}
\end{equation}

This is also the result proven in \cite{RudWig_amq21} for the hemisphere.

This gives an interesting analogy between the two-dimensional setting
and the quantum graph setting. The term $\frac{2\sigma}{L}$ takes
naturally the place of $\frac{2\sigma}{\text{Area}\left(\Omega\right)}$
(as $L$ is the Lebesgue measure of our space). More interestingly,
the boundary term $\text{Length}\left(\partial\Omega\right)$ is replaced
by a discrete measure on the graph boundary (which in our case is
exactly the set of Robin points), which assigns to each vertex total
weight which is inversely proportional to its degree. Heuristically,
the higher the vertex degree, the less it ``feels'' the $\delta$
perturbation on average.

While it seems like this dependence on degree only appears in the
graph setting, we believe that at least in a sense, a similar behavior
also exists for two-dimensional domains. To see how, note that the
boundary $\partial\Omega$ of the planar domains considered in \cite{RudWigYes_arxiv21}
only intersects $\Omega$ on one side (see Figure \ref{fig: bndry}),
which can be thought of as dividing the expression in (\ref{eq:-33-2})
by one. On the other hand, each Robin vertex $v$ intersects the graph
$\Gamma$ through $\deg\left(v\right)$ different edges, which might
be the reason for dividing by $\deg\left(v\right)$. We conjecture
that if the Robin boundary for the domain $\Omega$ is chosen to have
a two sided intersection with $\Omega$ (see again Figure \ref{fig: bndry}),
the corresponding expression in (\ref{eq:-33-2}) should be divided
by two. In this sense, the degree of the vertex can be replaced by
the number of sides of the boundary which are in contact with the
domain.

\begin{figure}
\includegraphics[scale=0.5]{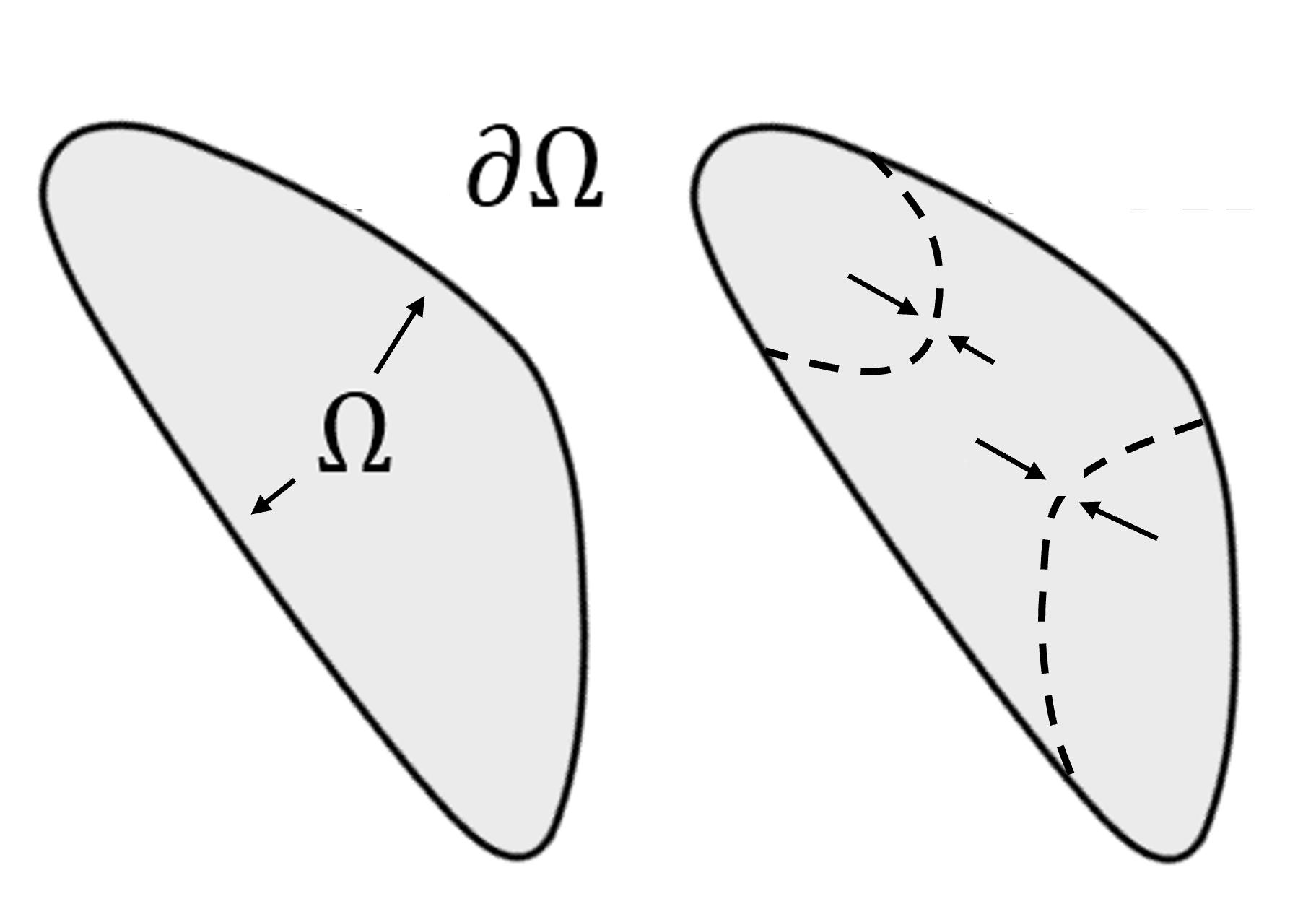}

\caption[Domains with different ``degrees''.]{The boundary of the left domain only touches the domain $\Omega$
from one side. The boundary of the right domain contains components
(the two dashed curves) which touch the domain from both sides. In
formula (\ref{eq:-33-2}), the contribution of these components to
the boundary term will contain a factor of $\frac{1}{2}$, which is
analogous to the degree term from the graph setting. \label{fig: bndry}}
\end{figure}

\bigskip

The fact that the Robin-Neumann gaps are uniformly bounded (\textbf{Theorem
\ref{thm:1.RNG-Lipschitz}}) is remarkable, since it no longer holds
when one passes to the two-dimensional setting. The sequence of RNG
is known to be unbounded for the hemisphere (\cite{RudWig_amq21}),
and also conjectured to be unbounded for certain planar domains, like
the disk (\cite{RudWigYes_arxiv21}).

\textbf{Theorem \ref{thm:1.RNG-Lipschitz}} also shows the existence
of a uniformly convergent subsequence of gaps. The result in \cite{RivRoy_jphys20}
shows that for star graphs, there exists a subsequence which converges
pointwise to zero. Although the proof is not included in this work,
this can be in fact shown for the general graph setting as well (see
\cite{Banda}).

For the two-dimensional case, it is shown in \cite{RudWigYes_arxiv21}
that under the assumption that the billiard dynamics which correspond
to the domain $\Omega$ are ergodic, there exists a subsequence which
converges pointwise to the mean value.

The example in Subsection \ref{subsec:rational-proof} shows that
it is not necessarily the case for an equilateral star graph. We conjecture
that the statement is in fact true if one adds the assumption of rationally
independent edge lengths, which is analogous to the assumption of
ergodic billiards.

\bigskip

While the local Weyl law presented in \textbf{Theorem \ref{thm:Weyl-law}}
was originally obtained as part of the proof of the RNG mean value
estimate, we note that it is interesting on its own right. Firstly,
it shows that on average, the value of an eigenfunction at a given
vertex only depends on the total length of the graph and the vertex
degree (the structure of the graph itself does not affect this mean
value). Moreover, it shows that in a sense, the scattering amplitudes
at different edges are uncorrelated (since their product averages
to zero). Again, this result is rather remarkable, since it does not
depend on the structure of the graph itself.

\bigskip

Lastly, we note that while the results presented in this work for
the RNG are concerned with the Laplacian, all of these results are
easily generalizable to Schrödinger operators of the form $H=-\frac{d^{2}}{dx^{2}}+V\left(x\right)$
where $V\left(x\right)\in L^{\infty}\left(\Gamma\right)$ (with the
same $\delta$ vertex conditions).

\subsection{Discussion of spectral flow results}

The $\delta_{s}$ family can be thought of as a gradual perturbation
of our graph. By increasing the parameter $t$, it extrapolates between
the unperturbed Neumann-Kirchhoff condition at $t=0$, and the condition
at $t=\infty$, which corresponds to cutting the original graph at
the given set of points and imposing the Robin condition with parameter
$s$.

\bigskip

\textbf{Theorem \ref{thm:SF-index}} continues the sequence of nodal
index theorems proven in recent years. Not only does it give a metric
graph version for an existing index formula for the nodal deficiency
on domains (\cite{BerCoxMar_lmp19}); But it in fact provides much
more information about the eigenfunctions than the original theorem,
by considering its $s$ domains. This is a more general notion than
the usual nodal (or Neumann) domains studied in previous works, and
it is studied for the first time in this work.

Note that interestingly, although the theorem is concerned with the
$s$ deficiency, the nodal count still appears in the index formula,
even for $s\neq\infty$. In this sense, it seems that the nodal count
is `special', as it holds certain information about the $s$ count
for all values of $s$.

\bigskip

\textbf{Theorem \ref{thm:SF-index}} shows that the Robin map holds
data about the behavior of the graph eigenfunctions. While the Dirichlet
to Neumann map (which is a special case of the Robin map) has been
connected to studying the nodal behavior of eigenfunctions in past
works (\cite{BerCoxMar_lmp19,Berkolaiko2022}), this is the first
time where an operator of this type is used to obtain information
which is not strictly nodal.

In fact, the definition we give for the Robin map $\Lambda_{s}\left(c\right)$
can be extended even further, so that the parameter $c$ can take
any real value (even values within the spectrum of $H^{s}\left(\infty\right)$),
although we did not present this generalization in this work.

We hope that by gaining a better understanding of the Robin map itself,
one may further use it as a `mediator' in gaining information about
the eigenfunctions. While we have only defined the Robin map for graph
eigenfunctions, we believe that our definition may be generalized
for domains in $\mathbb{R}^{N}$ as well, and can then be used to
study eigenfunctions on Euclidean domains.

\bigskip

A quantity related to the spectral flow is the spectral shift presented
in \cite{BerKuc_arxiv21}, which in essence describes how the position
of the spectral curves with respect to one another changes along some
lateral perturbation.

It was shown that the spectral shift of a certain spectral curve is
equal to the stability index of the given eigenvalue with respect
the lateral perturbation (\cite{BerKuc_arxiv21}), and this observation
can in fact be used to give an alternative proof of the Nodal Magnetic
Theorem (\cite{Ber_apde13,BerWey_ptrsa14}). This result is similar
in spirit to the idea of the proof of Theorem \ref{thm:SF-index},
and we hope that it is possible to recast the theorem in the language
of the spectral shift.

\bigskip

\textbf{Theorem \ref{prop:SF-points}} gives information about the
spectral curves themselves. It tells us that the number of spectral
curves that cross each horizontal cross section is completely determined
by the number of $\delta_{s}$ interaction points.

\bigskip

\textbf{Theorems \ref{thm:SF-Betti1} and \ref{prop:SF-Betti2}} show
that the spectral flow can also tell us about topological features
of the graph.

\bigskip

We shall note that the spectral flow formalism can be recast in a
more abstract setting, in the language of symplectic geometry, Lagrangian
subspaces and the Maslov index (See, for instance, \cite{LatSuk_ams20,B.BoossBavnbek2018,B.BoossBavnbek2013}).

We did not focus on this perspective in this work, and instead gave
a more elementary approach. Nevertheless, this more general formalism
can be used to give alternative proofs to some of our results, and
in certain cases even shed additional light on the collective behavior
of the $\delta_{s}$ family. Moreover, it provides several additional
results which we cannot prove with only the tools presented in this
work (see \cite{Band}).

Our long term goal is to combine these two approaches in order to
gain a better understanding of the spectral curves of the $\delta_{s}$
family, and consequently gain a better understanding of the eigenfunctions
by studying their $s$ domains.

\subsection{Ideas for future work}

The present work suggests many possible new directions to explore.
We give suggestions to several of them, which arose naturally along
the research. The first three involve the RNG, while the other two
involve the spectral flow and the Robin map.
\begin{enumerate}
\item Theorem \ref{thm:RNG-mean} gives an expression for the expectation
(or first moment) of the RNG with respect to the natural density.
A natural question that arises is -- can the higher moments be computed
similarly as well? What is the general probability distribution of
the RNG? Does it hold further geometric information about the graph,
which is not seen from the first moment alone? (see Figure \ref{fig:RNG-hist})\\
\begin{figure}
\includegraphics[scale=0.7]{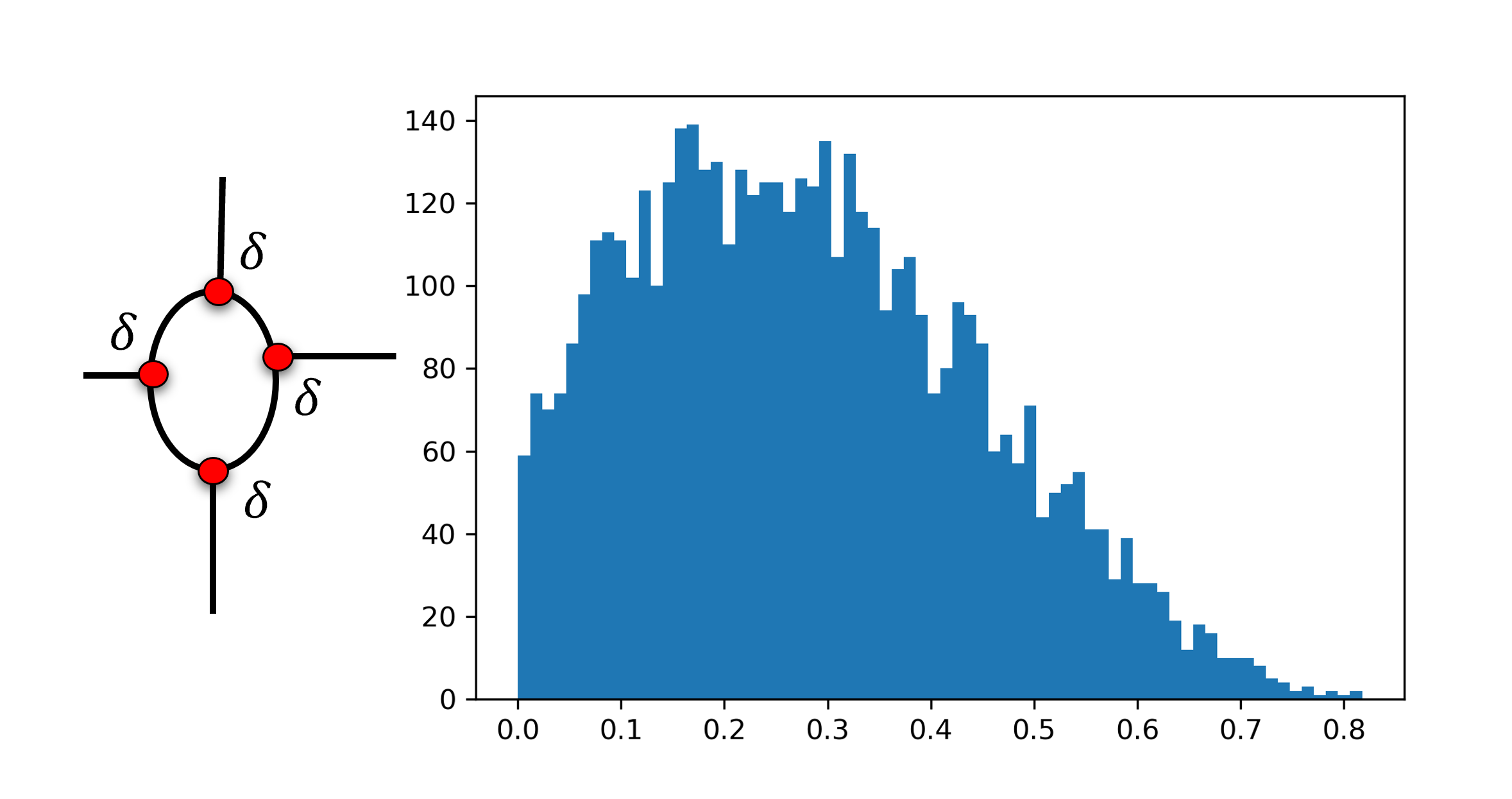}

\caption[Histogram of Robin-Neumann gaps.]{Histogram of the first five-thousand values of the Robin-Neumann gap
for a graph with four $\delta$ points, which gives a suggestion for
the possible probability distribution of the RNG. \label{fig:RNG-hist}}
\end{figure}
\\
Naively, the computation of the higher moments could be carried out
by an approach similar to before -- defining the higher moments as
functions on the secular manifold, and then computing the corresponding
integral. Yet, it turns out that the higher moments cannot be expressed
as well defined functions on the secular manifold. \\
Since this approach fails for the higher moments, this problem holds
an additional challenge of finding a different way to perform the
computation.
\item The RNG was defined using the $\delta$ vertex condition, which is
a special case of the $\delta_{s}$ family of vertex conditions. The
results for the RNG given in this work immediately raise the question
whether similar results can be obtained for the $\delta_{s}$ family
(and whether further geometric information can be obtained from studying
the corresponding gaps).\\
For instance, for the $\delta_{0}$ type condition, it seems (at least
numerically) that the mean value of the corresponding gap converges
as well (see Figure \ref{fig:DPRNG}) Nevertheless, it is not clear
what it converges to, and what additional properties this gap possesses.\\
\begin{figure}
\includegraphics[scale=0.6]{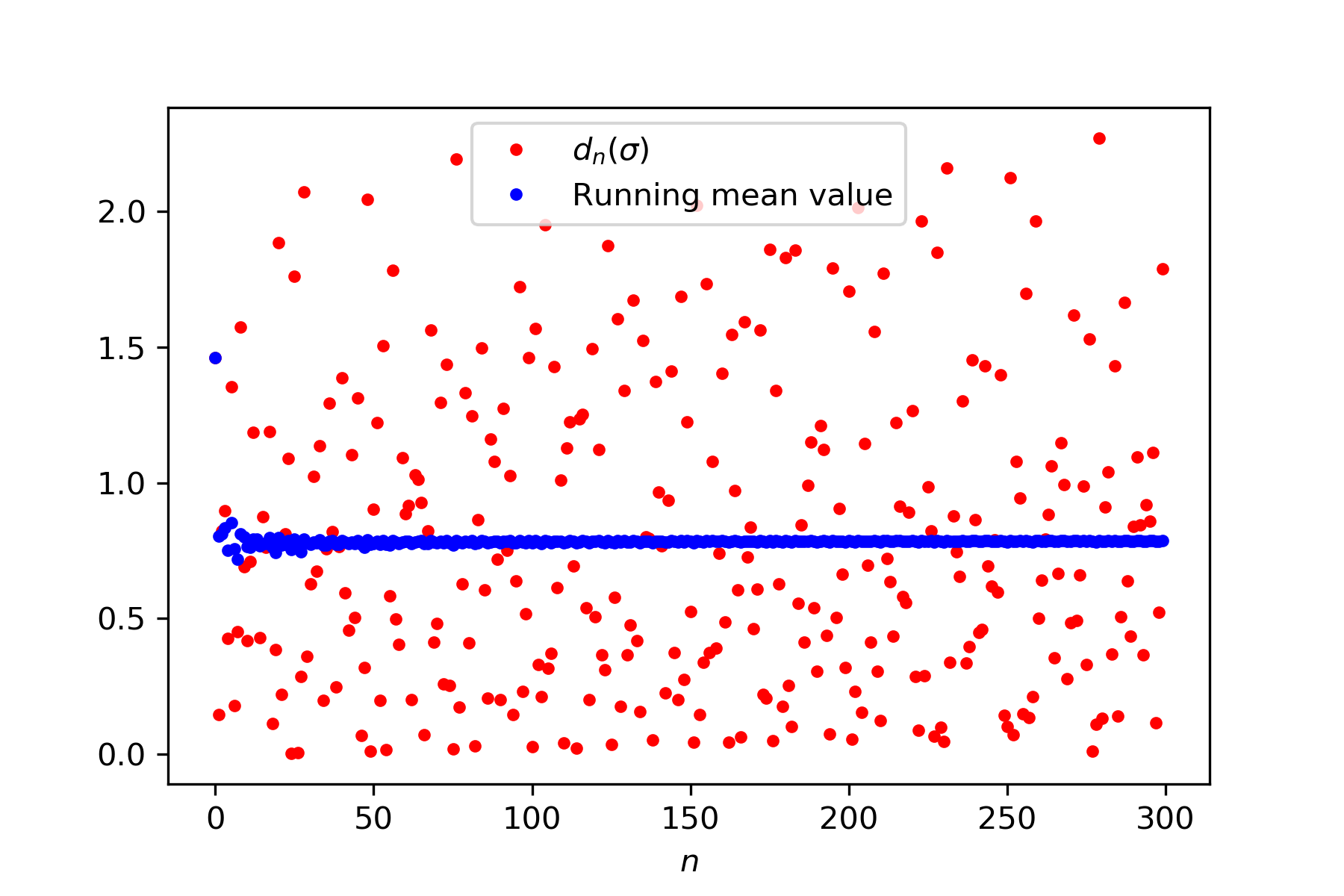}

\caption[Spectral gaps for the $\delta_{0}$ condition on a star graph.]{The first three-hundred values of the spectral gap for the $\delta_{0}$
condition on a star graph (red points), and the running mean value
of the sequence (blue points). Similar to the RNG, the sequence seems
to be bounded, and the mean value seems to converge.\label{fig:DPRNG}}
\end{figure}
When trying to repeat the computation shown for the RNG with different
vertex conditions from the $\delta_{s}$ family, one comes across
a problem similar to the one described in the previous bullet. It
turns out that for $s\neq\infty$, it is not simple to define the
corresponding spectral gap as a function on the secular manifold,
which makes it difficult to apply the approach presented in this work.\\
So once again, an interesting challenge could be to find an alternative
way to perform this computation, and see if any additional geometric
information can be derived from the statistics of the corresponding
gap.
\item In this work, we have assumed $\Gamma$ to be a compact metric graph.
In the case where $\Gamma$ is not compact, but periodic, then the
spectrum is no longer discrete, and consists of bands and gaps.\\
If we periodically place a Robin parameter $\sigma$ along the graph
(as done in the famous Kronig-Penney model, see \cite{Kronig}), then
as $\sigma$ changes, these spectral bands change as well. What one
obtains are not exactly spectral curves (since they have `width'),
but many of the usual notions from compact graphs still apply (including
the secular manifold).\\
It could be interesting to try and define a notion of a RNG for such
systems, which will quantify how the structure of the spectral bands
and gaps changes with respect to the Robin parameter. Then, one could
study its properties and see if any analogies to the compact case
exist.
\item Fix a certain graph eigenfunction $f_{n}$ (as in Figure \ref{fig:sjump}).
If we follow the value of the quantity $s\left(x\right):=\frac{f_{n}'\left(x\right)}{f_{n}\left(x\right)}$
near a vertex, we see that its behavior changes suddenly when we reach
the vertex. Denote by $s_{e}\left(v\right)$ the limit of the quantity
$s\left(x\right)$ as we approach the vertex $v$ from the edge $e$.
Then we see that for $s$ value slightly smaller than $s_{e}\left(v\right)$,
the number of $s$ points of $f_{n}$ is usually different than that
for $s$ value slightly larger than $s_{e}\left(v\right)$. These
so called critical $s$ values are exactly the $s$ values for which
the $s$ count of $f_{n}$ changes (see Figure \ref{fig:s-surfaces}),
and they hold a particular interest in studying the $s$ points of
$f_{n}$. \\
\begin{figure}
\includegraphics[scale=0.55]{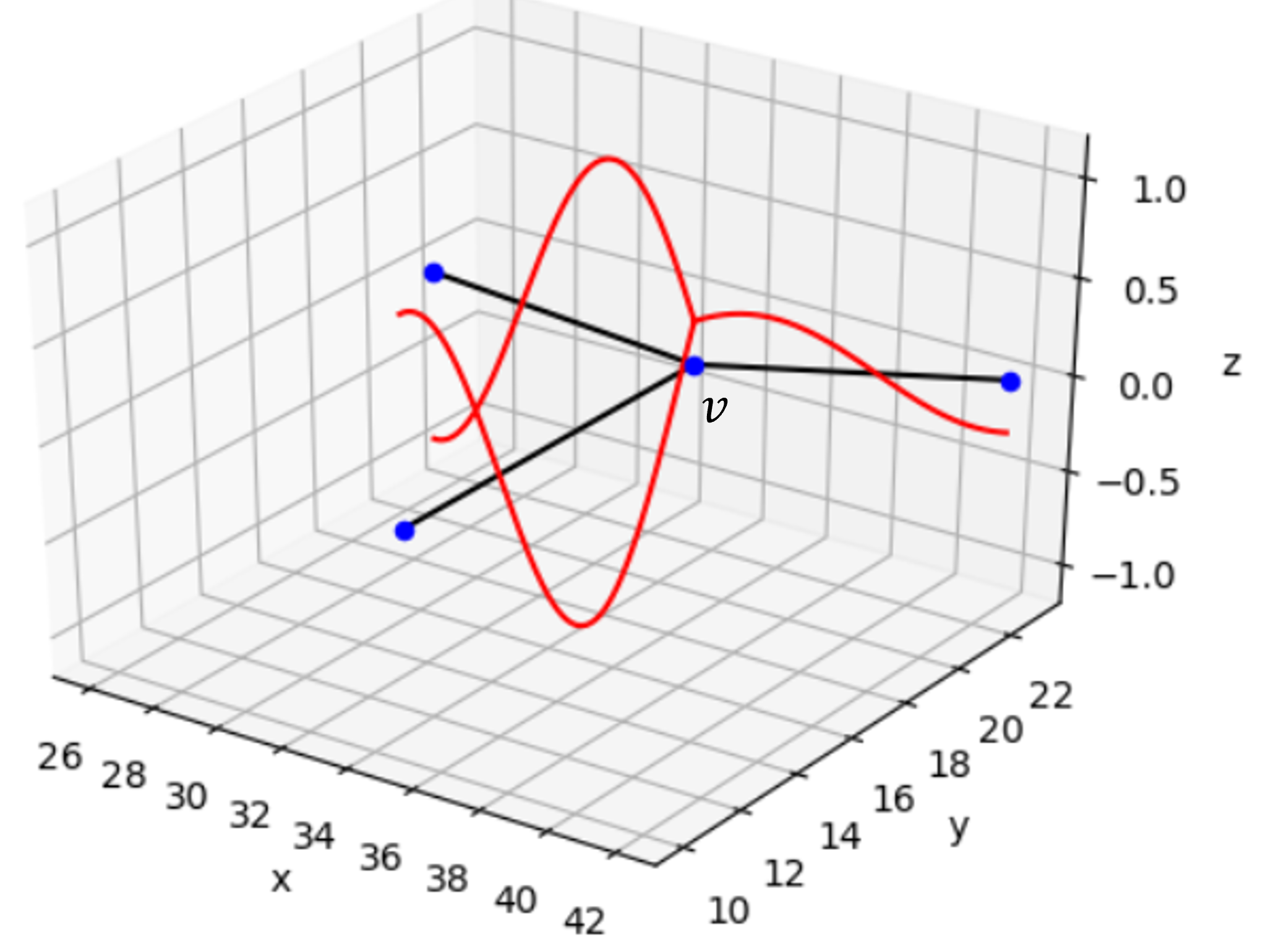}

\caption[The jump in the quantity $s\left(x\right)$.]{When approaching the central vertex through each individual edge,
the limit of the quantity $s\left(x\right)$ is different. These three
critical $s$ values result in $f_{n}$ having a different number
of $s$ points for different values of $s$.\label{fig:sjump}}
\end{figure}
Theorem \ref{prop:SF-points} shows that the number of $s$ points
of a given eigenfunction is equal to the spectral flow of the corresponding
$\delta_{s}$ family. Thus, at these critical $s$ values, the spectral
flow for the $\delta_{s}$ family changes suddenly. This can be seen
by considering not the spectral curves, but in fact the spectral surfaces
of the two-parameter family $\delta_{s}\left(t\right)$. 
\begin{figure}
\includegraphics[scale=0.45]{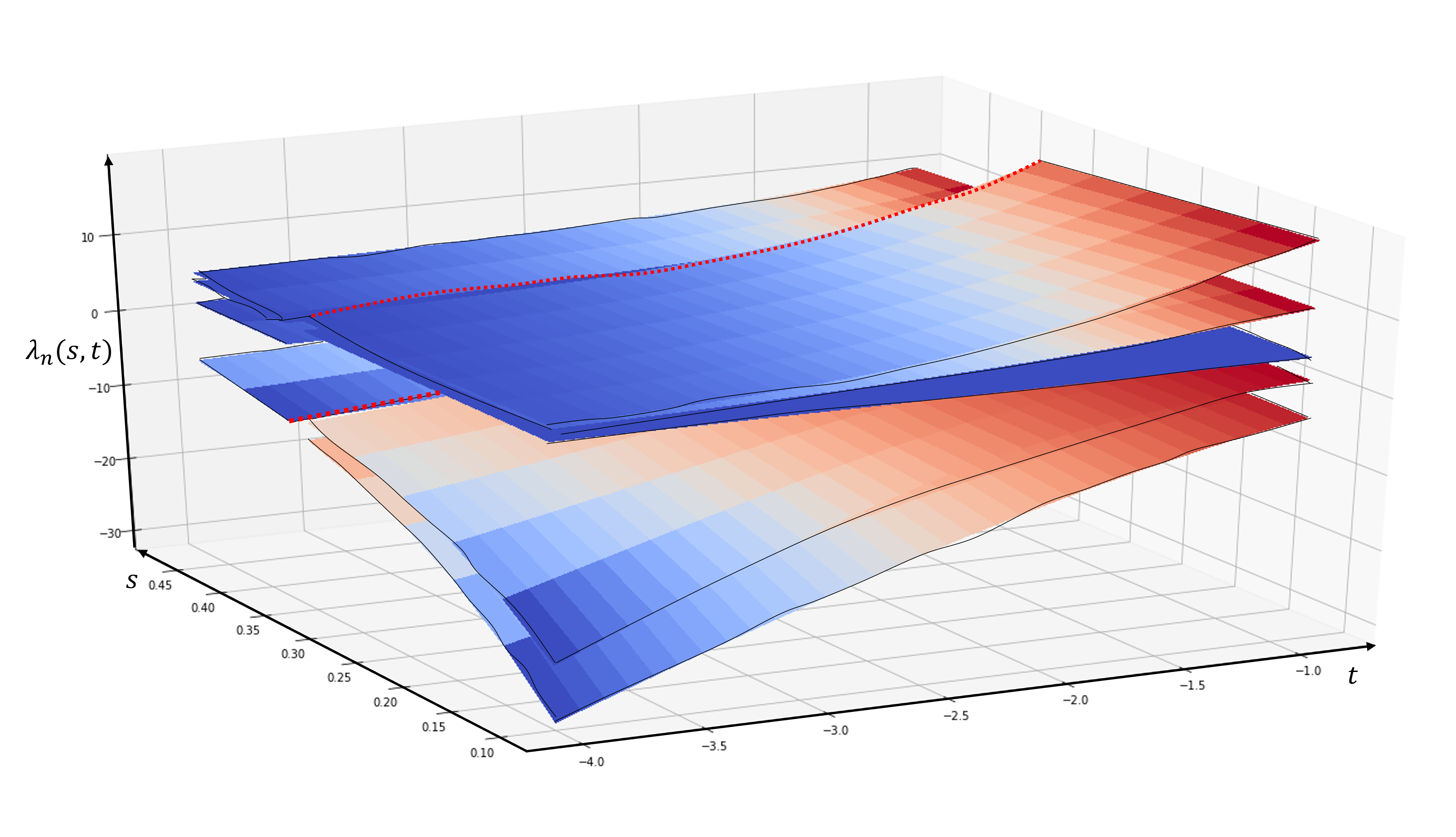}

\caption[The spectral surfaces for the $\delta_{s}\left(t\right)$ family.]{The spectral surfaces for the $\delta_{s}\left(t\right)$ family for
a fixed eigenfunction $f$ of a star graph. For each value of $s$,
we place the $\delta_{s}$ points on the $s$ points of $f$. At $s=0.3$,
the number of $s$ points of $f$ jumps from two to one, which results
in branching of the spectral surfaces (red dashed lines). \label{fig:s-surfaces}}
\end{figure}
\\
Whenever we reach a critical $s$ value, the spectral flow changes
(either increases by one or decreases by one). This means that at
the critical $s$ values, a single branch of the spectral surfaces
is created/destroyed, as seen in Figure \ref{fig:s-surfaces}.\\
It could be interesting to study these spectral surfaces in order
to gain insight on the possible number of $s$ points an eigenfunction
could have.
\item The study of the statistical behavior of the nodal count on quantum
graphs (such as the mean nodal deficiency) was the subject of several
works in recent years, such as \cite{AloBanBer_cmp18,AloBanBer_21arxiv,Alon_PhDThesis,Ban_ptrsa14}.
The newly introduced definition of $s$ domains opens a door for studying
the statistics of $s$ domains, rather than just nodal domains.\\
The study of Neumann domains (\cite{AloBan_21ahp,AloBanBerEgg_lms20})
shows that different $s$ values can display different statistics.
It could thus be interesting to try and study the collective statistical
behavior of $s$ points, in order to gain a better understanding of
the behavior of the graph eigenfunctions.
\end{enumerate}
\newpage{}

\appendix

\section{\label{sec:Appendices} Derivation of the sesquilinear form for the
$\delta_{s}$ condition}

We wish to compute the sesquilinear form $L_{t}^{s}$ which corresponds
to the $\delta_{s}\left(t\right)$ condition. To do this, we apply
the method presented in Theorem $1.4.11$ in \cite{BerKuc_graphs}.

We perform the computation for $s\neq\infty$. The computation for
$s=\infty$ (which is just the $\delta$ condition) is standard and
appears in the reference above.

Let us start with the case $t\neq0,\infty$. Working in the standard
coordinates on the graph (where the derivative of a function $u$
is taken into the edge), the $\delta_{s}$ vertex conditions presented
in formulas (\ref{eq:-6-1-1},\ref{eq:-7-1-1}) can be written at
each individual vertex $v$ in the following form:
\begin{align}
 & A_{v}\left(\begin{array}{c}
u_{1}\left(v\right)\\
u_{2}\left(v\right)
\end{array}\right)+B_{v}\left(\begin{array}{c}
u_{1}'\left(v\right)\\
u_{2}'\left(v\right)
\end{array}\right)=0,\label{eq:-43-1}\\
 & A_{v}=\left(\begin{array}{cc}
\cos\left(\alpha\right) & -\cos\left(\alpha\right)\\
\frac{1}{2}\cos\left(\alpha\right)t+\sin\left(\alpha\right) & \frac{1}{2}\cos\left(\alpha\right)t-\sin\left(\alpha\right)
\end{array}\right),\label{eq:-89}\\
 & B_{v}=\left(\begin{array}{cc}
\sin\left(\alpha\right) & \sin\left(\alpha\right)\\
\frac{1}{2}\sin\left(\alpha\right)t-\cos\left(\alpha\right) & -\frac{1}{2}\sin\left(\alpha\right)t-\cos\left(\alpha\right)
\end{array}\right),\label{eq:-90}
\end{align}
where $\alpha$ is the Prüfer angle, $s=\cot\left(\alpha\right)$.

By the result presented in \cite{BerKuc_graphs}, the sesquilinear
form is given by
\begin{equation}
L_{t}^{s}\left(f,g\right)=\int_{\Gamma}\frac{df}{dx}\overline{\frac{dg}{dx}}dx-\sum_{v\in B}\left\langle \Lambda_{v}P_{R,v}F,P_{R,v}G\right\rangle ,\label{eq:-44-1}
\end{equation}

where:
\begin{align}
 & P_{R,v}=I-P_{D,v}-P_{N,v},\label{eq:-91}\\
 & \Lambda_{v}=B_{v}^{-1}A_{v}P_{R,v},\label{eq:-45-2}
\end{align}
and $P_{D,v},P_{N,v}$ are the orthogonal projections onto $\ker\left(B_{v}\right)$
and $\ker\left(A_{v}\right)$ correspondingly.

A straightforward computation gives that for $t\neq0$
\begin{align}
 & P_{R,v}=I,\label{eq:-88}\\
 & \Lambda_{v}=\frac{1}{\sin^{2}\left(\alpha\right)t}\left(\begin{array}{cc}
1+\frac{1}{2}\sin\left(2\alpha\right)t & -1\\
-1 & 1-\frac{1}{2}\sin\left(2\alpha\right)t
\end{array}\right).\label{eq:-46-1}
\end{align}

Plugging this into formula \ref{eq:-44-1} we get
\begin{align}
 & \,\,\,\,\,\,\,\,\,\,\,\,\,\,\,\,L_{t}^{s}\left(f,g\right)=\int_{\Gamma}\frac{df}{dx}\overline{\frac{dg}{dx}}dx\label{eq:-47-1}\\
 & -\frac{1}{\sin^{2}\left(\alpha\right)t}\sum_{v\in B}\left(\begin{array}{c}
f_{1}\left(v\right)\\
f_{2}\left(v\right)
\end{array}\right)^{*}\left(\begin{array}{cc}
1+\frac{1}{2}\sin\left(2\alpha\right)t & -1\\
-1 & 1-\frac{1}{2}\sin\left(2\alpha\right)t
\end{array}\right)\left(\begin{array}{c}
g_{1}\left(v\right)\\
g_{2}\left(v\right)
\end{array}\right).
\end{align}

Note that while the computation was done in the non-oriented coordinates,
since the sesquilinear form only depends on the value of the function
(and not the orientation of the derivative), the same formula holds
in the oriented coordinates we use in this work.

By \cite{BerKuc_graphs}, the domains of the corresponding operators
consist of all functions in $H^{1}\left(\Gamma\right)$ which satisfy
$P_{D,v}=0$ at each vertex $v$. Since $P_{D,v}=0$, we conclude
that the domain is $H^{1}\left(\Gamma\right)$.

Now for the case $t=\infty$. This time, the corresponding matrices
are given by
\begin{align}
 & A_{v}=\left(\begin{array}{cc}
\cos\left(\alpha\right) & -\cos\left(\alpha\right)\\
\cos\left(\alpha\right) & \cos\left(\alpha\right)
\end{array}\right),\label{eq:-93}\\
 & B_{v}=\left(\begin{array}{cc}
\sin\left(\alpha\right) & \sin\left(\alpha\right)\\
\sin\left(\alpha\right) & -\sin\left(\alpha\right)
\end{array}\right).\label{eq:-94}
\end{align}

The same computation as before gives

\begin{align}
 & P_{R,v}=I,\label{eq:-91-1}\\
 & \Lambda_{v}=\frac{1}{2\sin^{2}\left(\alpha\right)}\left(\begin{array}{cc}
\sin\left(2\alpha\right) & 0\\
0 & -\sin\left(2\alpha\right)
\end{array}\right),\label{eq:-45-2-1}
\end{align}

which overall gives
\begin{align}
 & \,\,\,\,\,\,\,\,\,\,\,\,\,\,\,\,L_{\infty}^{s}\left(f,g\right)=\int_{\Gamma}\frac{df}{dx}\overline{\frac{dg}{dx}}dx\label{eq:-47-1-1}\\
 & -\frac{1}{\sin^{2}\left(\alpha\right)}\sum_{v\in B}\left(\begin{array}{c}
f_{1}\left(v\right)\\
f_{2}\left(v\right)
\end{array}\right)^{*}\left(\begin{array}{cc}
\frac{1}{2}\sin\left(2\alpha\right)t & 0\\
0 & -\frac{1}{2}\sin\left(2\alpha\right)t
\end{array}\right)\left(\begin{array}{c}
g_{1}\left(v\right)\\
g_{2}\left(v\right)
\end{array}\right).
\end{align}

Once again, $P_{D,v}=0$ and so the domain is given by $H^{1}\left(\Gamma\right)$.

Lastly, in the case $t=0$, we have that

\begin{align}
 & A_{v}=\left(\begin{array}{cc}
\cos\left(\alpha\right) & -\cos\left(\alpha\right)\\
\sin\left(\alpha\right) & -\sin\left(\alpha\right)
\end{array}\right),\label{eq:-93-1}\\
 & B_{v}=\left(\begin{array}{cc}
\sin\left(\alpha\right) & \sin\left(\alpha\right)\\
-\cos\left(\alpha\right) & -\cos\left(\alpha\right)
\end{array}\right).\label{eq:-94-1}
\end{align}

This time, the Dirichlet and Neumann projections are non-trivial:

\begin{align}
 & \ker\left(B_{v}\right)=sp\left\{ \left(1,-1\right)\right\} ,\label{eq:-91-1-1}\\
\Rightarrow & P_{D,v}=\frac{1}{\sqrt{2}}\left(\begin{array}{cc}
1 & -1\\
-1 & 1
\end{array}\right),\\
 & \ker\left(A_{v}\right)=sp\left\{ \left(1,1\right)\right\} ,\\
\Rightarrow & P_{N,v}=\frac{1}{\sqrt{2}}\left(\begin{array}{cc}
1 & 1\\
1 & 1
\end{array}\right),
\end{align}

Which gives $\Lambda=0$ and the following sesquilinear form:

\begin{align}
 & L_{0}^{s}\left(f,g\right)=\int_{\Gamma}\frac{df}{dx}\overline{\frac{dg}{dx}}dx.\label{eq:-47-1-1-1}
\end{align}

By the expression above for $P_{D,v}$, we also see that the domain
of $L_{0}^{s}$ consists of all functions in $H^{1}\left(\Gamma\right)$
which are continuous at the selected set of vertices.

\section{\label{sec:DTN-comp} Explicit computation of $\Lambda_{\infty}\left(c\right)$
for specific graph}

We give an example of a computation for the Robin map $\Lambda_{\infty}\left(c\right)$
for the graph displayed in Figure \ref{fig:DTN-graph}. For simplicity
of the computation, we take all edge lengths to be equal to $L$.

\begin{figure}
\includegraphics[scale=0.8]{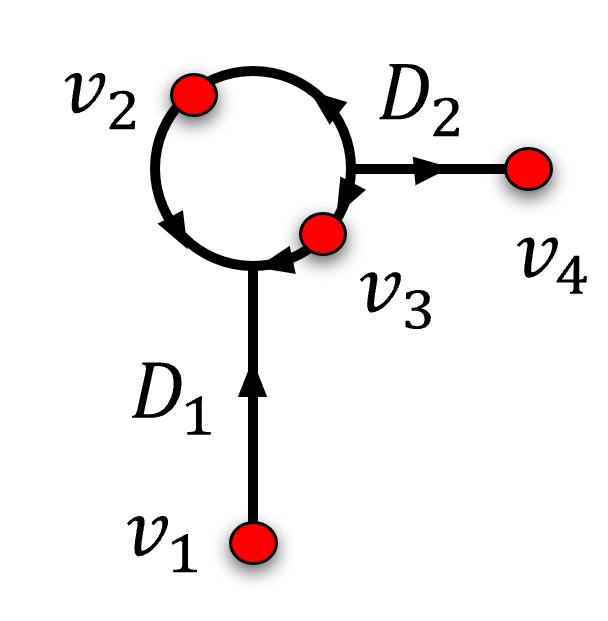}

\caption[Computation of $\Lambda_{\infty}\left(c\right)$ for specific graph.]{Computation of the Robin map $\Lambda_{\infty}\left(c\right)$ for
the graph above, where the vertex set $B$ is marked in red. The graph
is decomposed into to star graphs, $D_{1}$ and $D_{2}$, and the
computation of $\Lambda_{\infty}\left(c\right)$ is done by computing
the matrices for $\Lambda_{\infty}^{D_{1}}\left(c\right)$ and $\Lambda_{\infty}^{D_{1}}\left(c\right)$
and taking their block sum. \label{fig:DTN-graph}}
\end{figure}

To compute the Robin map, we first compute the map $\Lambda_{\infty}^{D}\left(c\right)$
for a star graph with three edges, where the set $B$ is taken to
be its exterior vertices (and the Neumann-Kirchhoff condition is imposed
at the center). We will then construct the matrix for $\Lambda_{\infty}\left(c\right)$
by using the decomposition of the given graph into two stars. For
convenience, we parameterize our star graph so that the interior vertices
are located at $x=0$ and the center is located at $x=L$.

By using the definition of our trace maps, we see that for the case
$s=\infty$, the map $\Lambda_{\infty}^{D}\left(c\right)$ is simply
the Dirichlet to Neumann map, which sends the Dirichlet trace data
to the Neumann trace data.

Let $k^{2}\notin\spec{H^{\infty}\left(\infty\right)}$. For simplicity,
we choose $c>0$ and write $c=k^{2}$. Let $w\in\mathbb{R}^{3}$ be
given. We are interested in solving the following boundary value problem:
\begin{align}
 & -\frac{d^{2}f}{dx^{2}}=k^{2}f,\label{eq:-111}\\
 & f_{e}\left(0\right)=w_{e},e\in\left\{ 1,2,3\right\} ,\\
 & f_{e}\left(L\right)=f_{e'}\left(L\right),\forall e,e'\in\left\{ 1,2,3\right\} ,\\
 & \sum_{e=1}^{3}f_{e}'\left(L\right)=0.
\end{align}

For each edge $e\in\left\{ 1,2,3\right\} $, we suggest a solution
of the form
\begin{equation}
f_{e}\left(x\right)=A_{e}\sin\left(kx\right)+B_{e}\cos\left(kx\right).\label{eq:-113}
\end{equation}

Then the condition at the outer vertices gives
\begin{equation}
B_{e}=w_{e}.\label{eq:-110}
\end{equation}

The continuity condition at the center gives that for every $e\in\left\{ 1,2,3\right\} $:
\begin{align}
 & A_{1}\sin\left(kL\right)+w_{1}\cos\left(kL\right)=A_{e}\sin\left(kL\right)+w_{e}\cos\left(kL\right)\label{eq:-112}\\
\Rightarrow & A_{e}=\frac{1}{\sin\left(kL\right)}\left(A_{1}\sin\left(kL\right)+w_{1}\cos\left(kL\right)-u_{e}\cos\left(kL\right)\right)\\
\Rightarrow & f_{e}\left(x\right)=\left(A_{1}+w_{1}\cot\left(kL\right)-u_{e}\cot\left(kL\right)\right)\sin\left(kx\right)+w_{e}\cos\left(kx\right).
\end{align}

Plugging this into the condition for the sum of derivatives at the
center, a straightforward computation yields that for each edge $e$
\begin{equation}
A_{e}=\frac{\sum_{j=1}^{3}\left(\frac{1}{\sin\left(kL\right)}\left(w_{j}\cos\left(kL\right)-w_{e}\cos\left(kL\right)\right)\cos\left(kL\right)+w_{j}\sin\left(kL\right)\right)}{3\cos\left(kL\right)}.\label{eq:-114}
\end{equation}

This overall gives
\begin{align}
 & f_{e}\left(x\right)=\left[\frac{1}{\sin\left(kL\right)\cos\left(kL\right)}w_{j}-w_{e}\sin\left(kL\right)\right]\sin\left(kx\right)+w_{e}\cos\left(\lambda x\right),\label{eq:-115}
\end{align}

and the corresponding Neumann trace is thus
\begin{equation}
f_{e}'\left(0\right)=\frac{k}{\sin\left(kL\right)\cos\left(kL\right)}w_{j}-kw_{e}\sin\left(kL\right).\label{eq:-116}
\end{equation}

We thus see that $\Lambda_{\infty}^{D}\left(k^{2}\right)$ can be
written in the following matrix form:
\begin{equation}
\left(\Lambda_{\infty}^{D}\left(k^{2}\right)\right)_{ee'}=\begin{cases}
\frac{k}{\sin\left(kL\right)\cos\left(kL\right)} & e\neq e'\\
\frac{k}{\cos\left(kL\right)} & e=e'
\end{cases}\label{eq:-117}
\end{equation}

We now construct the complete matrix $\Lambda_{\infty}\left(k^{2}\right)$,
which is a square matrix of size four acting on $\ell^{2}\left(\left\{ v_{1},...,v_{4}\right\} \right)$.
From the orientation on the edges (and definition of our weight functions
$\chi_{D}$), we see that $\Lambda_{\infty}^{D_{1}}\left(k^{2}\right)$
comes with a minus sign while $\Lambda_{\infty}^{D_{2}}\left(k^{2}\right)$
comes with a plus sign, which means that most terms cancel. Taking
the block sum (and noting that $D_{1}\cap D_{2}=\left\{ v_{2},v_{3}\right\} $)
we finally get
\begin{equation}
\Lambda_{\infty}\left(k^{2}\right)=\left(\begin{array}{cccc}
-\frac{k}{\cos\left(kL\right)} & 0 & 0 & 0\\
0 & 0 & 0 & 0\\
0 & 0 & 0 & 0\\
0 & 0 & 0 & \frac{k}{\cos\left(kL\right)}
\end{array}\right).\label{eq:-118}
\end{equation}

\newpage{}

\bibliographystyle{plain}
\bibliography{GlobalBib_210709}

\begin{thebibliography}{10}

\bibitem{Alon_PhDThesis}
L.~Alon.
\newblock {\em Quantum graphs - Generic eigenfunctions and their nodal count
  and Neumann count statistics}.
\newblock PhD thesis, Mathamtics Department, Technion - Israel Institute of
  Technology, 2020.

\bibitem{AloBan_21ahp}
L.~Alon and R.~Band.
\newblock Neumann domains on quantum graphs.
\newblock {\em Annales Henri Poincar\'{e}}, 2021.

\bibitem{AloBanBer_cmp18}
L.~{Alon}, R.~{Band}, and G.~{Berkolaiko}.
\newblock {Nodal Statistics On Quantum Graphs}.
\newblock {\em Comm. Math. Phys.}, 2018.

\bibitem{AloBanBer_21arxiv}
L.~Alon, R.~Band, and G.~Berkolaiko.
\newblock Universality of nodal count distribution in large metric graphs.
\newblock 2021.

\bibitem{AloBanBerEgg_lms20}
L.~Alon, R.~Band, M.~Bersudsky, and S.~Egger.
\newblock Neumann domains on graphs and manifolds.
\newblock In {\em Analysis and Geometry on Graphs and Manifolds}, volume 461 of
  {\em London Math. Soc. Lecture Note Ser.} 2020.

\bibitem{AvrExnLas_prl94}
J.~E. Avron, P.~Exner, and Y.~Last.
\newblock Periodic {S}chr{\"o}dinger operators with large gaps and
  {W}annier-{S}tark ladders.
\newblock {\em Phys. Rev. Lett.}, 72(6):896--899, 1994.

\bibitem{Ban_ptrsa14}
R.~Band.
\newblock The nodal count {$\{0,1,2,3,\dots\}$} implies the graph is a tree.
\newblock {\em Philos. Trans. R. Soc. A}, 372, 2014.

\bibitem{BanHarJoy_prep12}
R.~Band, J.~M. Harrison, and C.~H. Joyner.
\newblock Finite pseudo orbit expansions for spectral quantities of quantum
  graphs.
\newblock preprint {\tt arXiv:1205.4214}, 2012.

\bibitem{Band}
R.~Band, M.~Prokhorova, and G.~Sofer.
\newblock Spectral flow and the generalized nodal deficiency.
\newblock {\em (In writing)}.

\bibitem{Banda}
R.~Band, H.~Schanz, U.~Smilansky, and G.~Sofer.
\newblock Differences between robin and neumann eigenvalues on metric graphs.
\newblock {\em (In writing)}.

\bibitem{BarGas_jsp00}
F.~Barra and P.~Gaspard.
\newblock On the level spacing distribution in quantum graphs.
\newblock {\em J. Statist. Phys.}, 101(1--2):283--319, 2000.

\bibitem{Ber_cmp08}
G.~Berkolaiko.
\newblock A lower bound for nodal count on discrete and metric graphs.
\newblock {\em Comm. M.Phys.}, 278(3):803--819, 2008.

\bibitem{Ber_apde13}
G.~Berkolaiko.
\newblock Nodal count of graph eigenfunctions via magnetic perturbation.
\newblock {\em Anal. PDE}, 6, 2013.

\bibitem{Berkolaiko_qg-intro17}
G.~Berkolaiko.
\newblock An elementary introduction to quantum graphs.
\newblock In {\em Geometric and computational spectral theory}, volume 700 of
  {\em Contemp. Math.}, pages 41--72. Amer. Math. Soc., Providence, RI, 2017.

\bibitem{Berkolaiko2022}
G.~Berkolaiko, G.~Cox, B.~Helffer, and M.~Sundqvist.
\newblock Computing nodal deficiency with a refined dirichlet-to-neumann map.
\newblock 2022.

\bibitem{BerCoxMar_lmp19}
G.~Berkolaiko, G.~Cox, and J.~Marzuola.
\newblock Nodal deficiency, spectral flow, and the {D}irichlet-to-{N}eumann
  map.
\newblock {\em Lett. Math. Phys.}, 109(7):1611--1623, 2019.

\bibitem{BerKuc_arxiv21}
G.~Berkolaiko and P.~Kuchment.
\newblock Spectral shift via "lateral" perturbation.
\newblock {\tt arXiv:2011.11142}.

\bibitem{BerKuc_graphs}
G.~Berkolaiko and P.~Kuchment.
\newblock {\em Introduction to Quantum Graphs}, volume 186 of {\em Math. Surv.
  and Mon.}
\newblock AMS, 2013.

\bibitem{BerWey_ptrsa14}
G.~Berkolaiko and T.~Weyand.
\newblock Stability of eigenvalues of quantum graphs with respect to magnetic
  perturbation and the nodal count of the eigenfunctions.
\newblock {\em Philos. Trans. R. Soc. A}, 372(2007):20120522, 2014.

\bibitem{BerWin_tams10}
G.~Berkolaiko and B.~Winn.
\newblock Relationship between scattering matrix and spectrum of quantum
  graphs.
\newblock {\em Trans. Amer. Math. Soc.}, 362(12):6261--6277, 2010.

\bibitem{BolEnd_ahp09}
J.~Bolte and S.~Endres.
\newblock The trace formula for quantum graphs with general self adjoint
  boundary conditions.
\newblock {\em Ann. Henri Poincar\'e}, 10(1):189--223, 2009.

\bibitem{B.BoossBavnbek2013}
B.~Booss-Bavnbek and C.~Zhu.
\newblock The maslov index in weak symplectic functional analysis.
\newblock {\em Ann Glob Anal Geom}, 2013.

\bibitem{B.BoossBavnbek2018}
B.~Booss-Bavnbek and C.~Zhu.
\newblock The maslov index in symplectic banach spaces.
\newblock {\em Memoirs of the American Mathematical Society}, 2018.

\bibitem{Borthwick2022}
D.~Borthwick, K.~Jones, and E.~M.~Harrell II.
\newblock The heat kernel on the diagonal for a compact metric graph.
\newblock 2022.

\bibitem{CdV_ahp15}
Y.~Colin~de Verdi\`{e}re.
\newblock Semi-classical measures on quantum graphs and the {G}au\ss{} map of
  the determinant manifold.
\newblock {\em Annales Henri Poincar\'{e}}, 16(2):347--364, 2015.
\newblock also {\tt arXiv:1311.5449}.

\bibitem{Cou_ngwgmp23}
R.~Courant.
\newblock Ein allgemeiner {S}atz zur {T}heorie der {E}igenfuktionen
  selbstadjungierter {D}ifferentialausdr\"ucke.
\newblock {\em Nachr. Ges. Wiss. G\"ottingen Math Phys}, pages 81--84, 1923.

\bibitem{Courant23}
R.~Courant.
\newblock Ein allgemeiner {S}atz zur {T}heorie der {E}igenfunktione
  selbstadjungierter {D}ifferentialausdr{\"u}cke.
\newblock {\em Nachr. Ges. Wiss. G\"ottingen Math Phys}, July K1:81--84, 1923.

\bibitem{Cox2017}
G.~Cox, C.~Jones, and J.~Marzuola.
\newblock Manifold decompositions and indices of schrodinger operators.
\newblock 2017.

\bibitem{Exner1995}
P.~Exner.
\newblock Lattice kronig-penney models.
\newblock {\em Physical Review Letters}, 1995.

\bibitem{GnuSmi_ap06}
S.~Gnutzmann and U.~Smilansky.
\newblock Quantum graphs: Applications to quantum chaos and universal spectral
  statistics.
\newblock {\em Adv. Phys.}, 55(5--6):527--625, 2006.

\bibitem{Gnutzmann2003}
S.~Gnutzmann, U.~Smilansky, and J.~Weber.
\newblock Nodal domains on quantum graphs.
\newblock 2003.

\bibitem{Had_book08}
J.~Hadamard.
\newblock {\em M\'{e}moire sur le probl\`{e}me d'analyse relatif \`{a}
  l'\'{e}quilibre des plaques elastiques encastr\'{e}es}.
\newblock Mem. Acad. Sci. Inst. de France, 1908.

\bibitem{Kato_book}
T.~Kato.
\newblock {\em Perturbation theory for linear operators}.
\newblock Springer-Verlag, Berlin, second edition, 1976.
\newblock Grundlehren der Mathematischen Wissenschaften, Band 132.

\bibitem{KotSmi_ap99}
T.~Kottos and U.~Smilansky.
\newblock Periodic orbit theory and spectral statistics for quantum graphs.
\newblock {\em Ann. Physics}, 274(1):76--124, 1999.

\bibitem{KotSmi_prl00}
T.~Kottos and U.~Smilansky.
\newblock Chaotic scattering on graphs.
\newblock {\em Phys. Rev. Lett.}, 85(5):968--971, 2000.

\bibitem{Kronig}
R~D.~L. Kronig and W.~G. Penney.
\newblock Quantum mechanics of electrons in crystal lattices.
\newblock {\em Proceedings of the Royal Society}, 1931.

\bibitem{LatSuk_ams20}
Y.~Latushkin and S.~Sukhtaiev.
\newblock An index theorem for {S}chr\"{o}dinger operators on metric graphs.
\newblock In {\em Analytic trends in mathematical physics}, volume 741 of {\em
  Contemp. Math.}, pages 105--119. Amer. Math. Soc., [Providence], RI, 2020.

\bibitem{RivRoy_jphys20}
G.~Rivi\'{e}re and J.~Royer.
\newblock Spectrum of a non-selfadjoint quantum star graph.
\newblock {\em J. Phys. A}, 53(49):495202, 2020.

\bibitem{RudWig_amq21}
Z.~Rudnick and I.~Wigman.
\newblock On the robin spectrum for the hemisphere.
\newblock {\em Annales math\'{e}matiques du Qu\'{e}bec}, 2021.

\bibitem{RudWigYes_arxiv21}
Z.~Rudnick, I.~Wigman, and N.~Yesha.
\newblock Differences between {R}obin and {N}eumann eigenvalues.
\newblock {\tt arXiv:2008.07400}.

\end{thebibliography}

\end{document}